\documentclass[sigconf]{acmart}

\usepackage[linesnumbered,ruled,vlined]{algorithm2e}
\usepackage{amsmath,graphics,float,subfigure}
\usepackage{balance}
\usepackage{pifont}
\usepackage{wasysym}
\usepackage{multirow}
\usepackage{color}
\usepackage{xcolor}
\newcommand{\stitle}[1]{\vspace{1ex} \noindent{\bf #1}}
\long\def\comment#1{}
\newcommand{\kw}[1]{{\ensuremath{\mathsf{#1}}}\xspace}

\newcommand{\minCliques}{\kw{mineCliques}}
\newcommand{\inset}{\kw{insert}}
\newcommand{\acorn}{\kw{ACORN}}

\newcommand{\mci}{\kw{MCI}}
\newcommand{\mcc}{\kw{MCC}}
\newcommand{\tmcc}{$\tau$-\kw{MCC}}
\newcommand{\knng}{$k$-NNG\xspace}

\begin{document}


\title{MCI: A Maximal Clique Index for Efficient Arbitrary-Filtered Approximate Nearest Neighbor Search}

\settopmatter{authorsperrow=3}

\author{Xiaowei Ye}
\affiliation{%
	\institution{Beijing Institute of Technology}
	\city{Beijing}
	\country{China}
}
\email{yexiaowei@bit.edu.cn}

\author{ Rong-Hua Li}
\affiliation{%
	\institution{Beijing Institute of Technology}
	\city{Beijing}
	\country{China}
}
\email{ lironghuabit@126.com}

\author{Guoren Wang}
\affiliation{%
		\institution{Beijing Institute of Technology}
		\city{Beijing}
		\country{China}
	}
\email{ wanggrbit@126.com}

\author{ Kaiwen Xue}
\affiliation{%
	\institution{Huawei}
	\city{Dongguan}
	\country{China}
}
\email{ xuekaiwen6@huawei.com}

\author{ Daiyin Wang}
\affiliation{%
	\institution{Huawei}
	\city{Dongguan}
	\country{China}
}
\email{ wangdaiyin@huawei.com}

\author{ Xubin Li}
\affiliation{%
	\institution{Huawei}
	\city{Dongguan}
	\country{China}
}
\email{ lixubin@huawei.com}

\renewcommand{\shortauthors}{Rong-Hua Li et al.}

\begin{abstract}
	Approximate Nearest Neighbor Search with arbitrary filtering predicates (AFANNS) is essential for modern data applications, yet existing methods often incur substantial storage and computational costs. In this work, we introduce the Maximal Clique Index (\mci), a novel graph-based index designed for robust and efficient AFANNS. The core idea of \mci is to approximate a dense Nearest Neighbor Graph (NNG) through a compact, clique-based representation. We propose two key techniques: (1) Maximal Clique Cover (\mcc), which exploits the geometric transitivity of high-dimensional spaces to encode dense neighborhoods as maximal cliques, achieving an index with high compression and  connectivity; and (2) Local Neighborhood Graph Geometric Densification, a strategy that constructs an index approximating a large NNG from a sparse initial NNG,  recovers global connectivity by progressively increasing distance thresholds to locally densify the structure. The index is built in a lock-free parallel manner for scalability and queried via a carefully-designed multi-seed strategy to handle fragmented predicate-induced subgraphs. Extensive experiments on 10 datasets show that \mci significantly outperforms state-of-the-art methods by up to one order of magnitude in QPS at high recall while using substantially smaller space, and remains competitive even on range/keyword filtering tasks, demonstrating robust general-purpose performance.
\end{abstract}

\maketitle


\section{Introduction} \label{sec:intro}
With the proliferation of powerful vector embeddings, Approximate Nearest Neighbor Search (ANNS) has become a foundational component in numerous real-world applications, including multimedia retrieval \cite{DBLP:conf/kdd/BorisyukMM0LMBR21,DBLP:journals/pacmmod/MohoneyPCMIMPR23}, recommendation systems \cite{DBLP:conf/www/DasDGR07}, and Retrieval-Augmented Generation (RAG) \cite{DBLP:conf/kdd/FanDNWLYCL24}. In this landscape, Arbitrary Filtered Approximate Nearest Neighbor Search (AFANNS) is particularly crucial. For instance, in a RAG system, a user may need to retrieve a document chunk that satisfies specific metadata constraints (e.g., keywords, publication date) while simultaneously relying on vector similarity for semantic matching.

State-of-the-art (SOTA) ANNS methods predominantly rely on proximity graphs such as the Relative Neighborhood Graph (RNG) \cite{RNG1980}. A fundamental requirement for effective search in such graphs is connectivity, ensuring a traversal path exists from any starting node to the target. However, this property is inherently compromised in AFANNS due to query-specific predicates. Proximity graphs limit each node to a small, fixed number of out-neighbors (e.g., $M$ in HNSW \cite{HNSW} or NSG \cite{NSG}). Under a predicate with selectivity $s$ (the fraction of vectors that satisfy the predicate), the expected number of valid neighbors per node reduces to $M \cdot s$. When $s$ is small (e.g., $s=0.01$, $M=32$), many nodes are likely to have zero valid out-neighbors, fragmenting the induced subgraph and halting traversal. Consequently, traditional proximity-graph indexes with limited neighbor lists are fundamentally ill-suited for AFANNS.

To address this, the state-of-the-art ACORN algorithm \cite{ACORN} proposes a predicate-agnostic proximity graph index, extending the Hierarchical Navigable Small World (HNSW) structure. While standard HNSW maintains $M$ neighbors per node, ACORN expands this list to $M \cdot \gamma$, with $\gamma \approx 1/s$, to statistically preserve $M$ valid neighbors after filtering. To mitigate the resulting memory explosion, ACORN approximates the expanded neighborhood via 2-hop connections. Despite this optimization, ACORN still suffers from significant storage overhead and suboptimal search performance in practice \cite{iRangeGraph,MochengSurvey,UNG}. Alternative strategies like pre-filtering (filter first, search after) and post-filtering (search first, filter after) also exhibit well-known efficiency bottlenecks, especially under varying selectivity \cite{MochengSurvey} (detailed in Section~\ref{subsec:exist-methods}).

The limitations of existing approaches reveal that traditional graph structures are inadequate for AFANNS. We therefore propose a fundamentally new index: the Maximal Clique Index (\mci). The core idea is to approximate a dense $k$-Nearest Neighbor Graph ($k$-NNG)—where $k$ is large enough to ensure connectivity under filtering (e.g., $k \approx M/s$)—but to do so efficiently via a much smaller $k'$-NNG ($k' < k$). This is made possible by a key geometric insight: in high-dimensional spaces, neighbors of a node tend to be neighbors of each other. Formally, we prove that nodes within a local neighborhood of a sparse $k'$-NNG are highly likely to form dense cliques (Theorem~\ref{the:nnarenn}). Thus, a large $k$-NNG's connectivity can be inferred from the maximal cliques within a sparse $k'$-NNG.

Based on this observation, we propose two pivotal techniques for constructing \mci from a sparse $k'$-NNG:
(1) Maximal Clique Cover (\mcc), which selects a covering set of maximal cliques (via a linear-time greedy algorithm \cite{CoreCliqueRemoval}) to avoid redundancy and ensure every node is covered; and (2) Geometric Densification of Neighborhood Graphs, which iteratively expands local neighborhood graphs by increasing a distance threshold, mining cliques progressively until full coverage is achieved. This prioritizes the early discovery of tight, high-quality cliques. For query processing, we design a beam-search algorithm augmented with a carefully-designed multiple seeds selection strategy, ensuring robust traversal even when the predicate-induced subgraph is fragmented.

\mci offers three key advantages over prior work. First, it eliminates the need for impractically large neighbor lists (e.g., $M\cdot\gamma$ in ACORN), constructing instead from a modest $k'$-NNG (e.g., $k'\le 200$), which remains computationally tractable at scale. Notably, the final index discards the original $k'$-NNG, incurring no dependency overhead. Second, it achieves a superior QPS–Recall trade-off. Our experiments show that \mci can deliver up to an order of magnitude higher QPS at high recall (e.g., recall@10 $> 0.95$) under mixed selectivity regimes. Third, our index is storage-efficient, with a footprint comparable to the highly compact IVFPQ algorithm and up to $10\times$ smaller than ACORN.

In summary, our main contributions are:
\begin{itemize}
    \item \textbf{A novel maximal clique index (\mci):} We introduce the first clique-based index for AFANNS, which compactly approximates a dense $k$-NNG via a maximal clique cover mined from a sparse $k'$-NNG, grounded in a formal geometric theorem. To the best of our knowledge, our work pioneers the use of maximal cliques for indexing in high-dimensional vector retrieval.

    \item \textbf{Efficient construction algorithm:} We design an algorithm that geometrically expands local neighborhood graphs to efficiently mine a covering set of cliques, ensuring complete node coverage with near-linear practical complexity.
    
    \item \textbf{Robust search algorithm:} We develop a search algorithm combining beam search with a multi-seed initialization, guaranteeing robust performance even over disconnected predicate subgraphs.
    
    \item \textbf{Comprehensive evaluation:} We conduct extensive experiments on 10 benchmark datasets. Results demonstrate that \mci consistently outperforms SOTA methods (ACORN, HNSW, IVFPQ, and specialized filters) in QPS at high recall, maintains compact index size, and exhibits robustness across all selectivity regimes.
    \item \textbf{Reproducibility :} Our source code is available at the anonymous repository \url{https://anonymous.4open.science/r/MCI_ANNS}.
\end{itemize}

\section{Preliminaries} \label{sec:preliminaries}
\subsection{Notations and problem definition}
Let $D = \{(v_i, f_i)\}_{i=1}^n$ be a dataset of $n$ vector-feature pairs, where each vector $v_i \in \mathbb{R}^{D}$ is associated with a corresponding feature $f_i$. Denote by $V=\{v_i\}_{i=1}^n$ and $F=\{f_i\}_{i=1}^{n}$. A query \( Q = (v_q, f_q, P_q) \) contains a query vector \( v_q \), a feature \( f_q \), and a predicate \( P_q \). For a data pair \( (v_i, f_i) \), the predicate evaluates \( P_q(f_i, f_q) \in \{\text{True}, \text{False}\} \). The pair \( (v_i, f_i) \) is \emph{valid} to \( Q \) if and only if \( P_q(f_i, f_q) \) holds. We denote the Euclidean distance between two vectors by \( d(\cdot, \cdot) \).

The \textit{selectivity} of a query on the dataset, denoted as $s(D,Q)$, is the fraction of elements that satisfy the predicate:
\begin{equation}
    s(D,Q) = \frac{\sum_{i=1}^{n}{1[P_q(f_i,f_q)=True]}}{n}. 
\end{equation}
Note that $0 \le s(D,Q) \le 1$. We use $s$ instead of $s(D,Q)$ when the context is clear.

Let $G(U,E)$ be a graph where $U$ is the set of nodes and $E\subseteq U\times U$ is the set of edges. A clique is a subgraph of $G$ that all pairs of nodes are connected by an edge. A maximal clique is a clique that cannot add any other node to reach a larger clique.

A \textit{$k$-nearest neighbor graph} ($k$-NNG) is a directed graph where an edge exists from $u$ to $v$ if $v$ is among the $k$ nearest neighbors of $u$ \cite{NN-Decent, survey21}. Here, the \emph{$k$-nearest neighbors} of \( u \) with respect to \( {V} \), denoted \( \mathcal{N}_k(u) \), is a set \( \mathcal{N}_k(u) \subseteq {V} \) of size \( k \) that satisfies:
\begin{equation}
\max_{v \in \mathcal{N}_k(u)} d(u, v) \ \le \ \min_{v' \in \mathcal{V} \setminus \mathcal{N}_k(u)} d(u, v').
\end{equation}

\stitle{Arbitrary-Filtered Approximate Nearest Neighbor Search (AFANNS).} 
Given a dataset $D=(V,F)$, a query $Q(v_q, f_q, P_q)$, and an integer $k$, the goal is to find the $k$ nearest neighbors of $v_q$ that satisfy the predicate $P_q(f_i, f_q) = True$.
Let $R$ be the set of the exact $k$ closest \textit{True} elements to $v_q$. The AFANNS problem aims to retrieve a set $R'$ of $k$ \textit{True} elements to maximize recall and efficiency, where Recall@$k = \frac{|R \cap R'|}{k}$.

\subsection{Challenges of AFANNS}
The \emph{arbitrary filtering} nature of AFANNS brings two basic challenges that together make it hard to design efficient indexing and query methods, as discussed below.

\comment{
\stitle{The ``Unhappy Middle'' of Selectivity Regimes.}
AFANNS confronts the ``unhappy middle'' problem across varying selectivity ranges \cite{SIEVE, UNG, MochengSurvey}:

At \textit{very low selectivity} (where few elements are valid), brute-force scanning of candidates is feasible due to the small subset size. At \textit{high selectivity} (where most elements satisfy the predicate), the top nearest neighbors in the full dataset are highly likely to include sufficient valid elements, rendering standard ANNS algorithms effective.

However, the challenge becomes acute at \textit{moderate selectivity} \cite{MochengSurvey,UNIFY,SIEVE}. This regime offers no clear optimization path: (1) brute-force scanning is computationally prohibitive due to the larger volume of valid elements compared to the low-selectivity case; (2) standard ANNS is inefficient because the top unfiltered neighbors often lack sufficient valid elements; and (3) searching solely within the predicate-induced subgraph is ineffective, as the subgraph often fragments into multiple disconnected components, hindering graph traversal. Consequently, the moderate selectivity regime constitutes the primary bottleneck for AFANNS.
}

\stitle{C1: Heterogeneous feature types rule out specialized indexing.}
Specialized ANNS variants use specific feature properties: keyword- or tag-based filtering ($f_i$ is a set of labels, attributes, or tags) relies on discrete set membership \cite{UNG,filteredDiskann,NHQ,TFANNS}, while range-based filtering ($f_i$ is a numeric scalar) uses the order of numerical attributes \cite{iRangeGraph,DBLP:conf/icml/EngelsLYDS24,SeRF,DSG,UNIFY,DIGRA,RangePQ}. In addition, some recent range-filtering indexes also support dynamic feature updates, such as DSG, DIGRA, and RangePQ+ \cite{DSG,DIGRA,RangePQ}. In contrast, AFANNS must handle a mix of data types (e.g., integers, floats, strings, and timestamps), which may also change over time. No single indexing paradigm can take advantage of the different properties of such diverse types at the same time, since the discrete or order-based assumptions behind specialized indexes do not always hold. As a result, general-purpose indexes often index only the vectors $V$, and treat the feature space $F$ as a generic metadata store.

\stitle{C2: Ad-hoc predicates break distribution assumptions.}
Several existing ANNS methods design indexing strategies under the assumption that future queries follow past distributions of both features and their predicates \cite{SIEVE,DBLP:conf/sigir/VecchiatoLNB24, DBLP:journals/jmlr/HyvonenJR24, DBLP:conf/dac/0003LYCWM024}. This assumption does not hold for AFANNS, where predicates are ad-hoc and unknown at index construction time. In particular, even for the same query vector, different predicates can produce very different valid subsets. This unpredictability breaks the assumption of stable query patterns, and makes distribution-based optimizations ineffective for arbitrary filtering workloads.


\subsection{Existing methods and their defects} \label{subsec:exist-methods}
Existing methods related to AFANNS can be broadly grouped into four categories: pre-filtering, post-filtering, predicate-subgraph traversal (represented by \acorn \cite{ACORN}), and specialized filtering methods tailored to specific predicate families. In the following discussion, the time complexity for the first three categories is measured by the number of distance computation operations during search.

\stitle{Pre-filtering.}
This strategy first identifies all elements that satisfy the predicate $P_q$, then performs similarity search (often via brute-force scanning or simple structures like IVF) over this filtered set. It guarantees high recall and is effective for queries with very low selectivity (small filtered sets). However, its time complexity is $O(sn)$, where $s$ is the selectivity of $P_q$, making it non-scalable for moderate or large selectivity on large datasets.

\stitle{Post-filtering.}
This approach performs an unfiltered ANNS to retrieve $O(k/s)$ candidates, then filters out those that do not satisfy $P_q$. Its time complexity is $O(\log n + k/s)$, making it efficient for highly selective queries (large $s$). However, when selectivity is low (i.e., $s$ is small), the number of candidates $O(k/s)$ may become impractically large, leading to poor efficiency.

\stitle{Predicate-subgraph traversal (\acorn).}
\acorn \cite{ACORN} extends the HNSW \cite{HNSW} index to support predicate filtering by enlarging the neighbor candidate pool during graph construction. Specifically, it scales the construction parameter $ef_c$ to $M \cdot \gamma$ (where $\gamma > 1$ is a scaling factor) to ensure sufficient valid neighbors exist for traversal under arbitrary predicates. To control memory overhead, it stores only the first $M_\beta$ candidates explicitly and approximates the remainder via 2-hop neighbors. \acorn achieves time complexity of $O((M \cdot \gamma \cdot s + \gamma) \cdot \log(sn) + \log(\frac{1}{s}))$ and requires $O(n(M_\beta + M + \frac{M\gamma}{\ln{M}}))$ index storage. While more scalable than pure pre- or post-filtering, \acorn still incurs significant storage costs and may suffer from suboptimal recall in complex selectivity scenarios, limiting its applicability in high-performance arbitrary-filtering ANNS \cite{iRangeGraph}.

\stitle{Specific filtering methods.} Several methods are designed for specific predicate types, such as keyword-based filtering \cite{UNG,filteredDiskann,NHQ,TFANNS}, range-based filtering \cite{iRangeGraph,DBLP:conf/icml/EngelsLYDS24, SeRF, DSG, UNIFY, DIGRA, RangePQ}, and categorical/numerical attribute filtering \cite{BeyondVectorSearch}. While they can achieve superior performance within their specialized domains, their underlying assumptions about feature types and predicate logic render them inapplicable to the general AFANNS setting, which must accommodate heterogeneous features and ad-hoc predicates. Table~\ref{tab:methods} summarizes the key characteristics of these methods.

\begin{table}[t!]
	\scriptsize
    \centering
    \vspace*{-0.2cm}
    \caption{Characteristics of representative filtered ANNS methods. ``Upd.'' indicates whether the feature values attached to indexed vectors can be updated after index construction without rebuilding.}
    \vspace*{-0.3cm}
    \label{tab:methods}
	\setlength{\tabcolsep}{2pt}
    \begin{tabular}{l|c|c|c|c}
        \toprule
        \textbf{Method} & \textbf{$F$ type} & \textbf{Upd.} & \textbf{$f_q$} & \textbf{$P_q(f_i,f_q)$} \\
        \midrule
        \multicolumn{5}{c}{\textit{Range-filtering methods}} \\
        \midrule
        SeRF \cite{SeRF}              & scalar (float) & \ding{55} & interval $[l,r]$ & $f_i\in[l,r]$ \\
        iRangeGraph \cite{iRangeGraph} & scalar (float) & \ding{55} & interval $[l,r]$ & $f_i\in[l,r]$ \\
        DSG \cite{DSG}                 & scalar (float) & \ding{51} & interval $[l,r]$ & $f_i\in[l,r]$ \\
        UNIFY \cite{UNIFY}             & scalar (float) & \ding{55} & interval $[l,r]$ & $f_i\in[l,r]$ \\
        Engels et al. \cite{DBLP:conf/icml/EngelsLYDS24} & scalar (float) & \ding{55} & interval $[l,r]$ & $f_i\in[l,r]$ \\
        DIGRA \cite{DIGRA}             & scalar (float) & \ding{51} & interval $[l,r]$ & $f_i\in[l,r]$ \\
        RangePQ+ \cite{RangePQ}        & scalar (float) & \ding{51} & interval $[l,r]$ & $f_i\in[l,r]$ \\
        \midrule
        \multicolumn{5}{c}{\textit{Keyword/label-filtering methods}} \\
        \midrule
        Filtered-DiskANN \cite{filteredDiskann} & label set & \ding{55} & label(s) & $f_q\subseteq f_i$ \\
        NHQ \cite{NHQ}                 & label set  & \ding{55} & label(s) & $f_q= f_i$ \\
			UNG \cite{UNG}                 & label set  & \ding{55} & label(s) & \begin{tabular}[c]{@{}c@{}}$f_q\in f_i$, $f_i\in f_q$,\\$f_i=f_q$, $f_i\cap f_q\neq \emptyset$\end{tabular} \\
        TFANNS \cite{TFANNS}           & label set  & \ding{55} & label(s) & $f_q\in f_i$ \\
        \midrule
		\multicolumn{5}{c}{\textit{Range plus label-filtering methods}} \\
		\midrule
			Xie et al. \cite{BeyondVectorSearch} & \begin{tabular}[c]{@{}c@{}}scalar (float) or label set\end{tabular} & \ding{55} & \begin{tabular}[c]{@{}c@{}}$[l,r]$ or\\label(s)\end{tabular} & $f_i\in [l,r]$ or $f_i\in f_q$ \\
		\midrule
        \multicolumn{5}{c}{\textit{Arbitrary-filtering methods}} \\
        \midrule
        Pre-/Post-filtering            & arbitrary & \ding{51} & arbitrary & arbitrary \\
        \acorn \cite{ACORN}            & arbitrary & \ding{51} & arbitrary & arbitrary \\
        \mci (ours)                    & arbitrary & \ding{51} & arbitrary & arbitrary \\
        \bottomrule
    \end{tabular}
    \vspace*{-0.3cm}
\end{table}

\comment{
\section{A Novel Maximal Clique Index} \label{sec:mc-index}
The inherent limitations of existing approaches suggest that traditional indexing mechanisms are inadequate for the specific requirements of AFANNS. To address these challenges, we propose a fundamentally new index structure, the \textit{Maximal Clique Index} (\mci), which is specifically engineered to handle the complexities of AFANNS.

\subsection{Core Ideas}\label{subsec:coreIdeas}
Traditional proximity graph-based indexes typically control the average out-degree and neighbor distribution of each node to optimize performance. The Relative Neighborhood Graph (RNG) pruning technique \cite{RNG1980} is widely employed in unfiltered vector retrieval indexes \cite{HNSW, NSG}. Formally, for a target node $u$ and a candidate neighbor $v$, the edge $(u, v)$ is retained only if there is no other node $w$ such that $d(u, w) \le d(u, v)$ and $d(w, v) \le d(u, v)$ (assuming the edge is not the longest side of the triangle $uvw$). This technique enforces an omnidirectional neighbor distribution and ensures graph sparsity, thereby accelerating the search process \cite{survey21}.

However, under query-specific filter conditions, the neighbors selected by RNG may fail to satisfy the filter constraints, leading to poor graph connectivity within the subset of valid nodes \cite{ACORN}. \textbf{Consequently, to adapt to varying, ad-hoc filter conditions, a robust index must retain a dense set of $k$-nearest neighbors ($k$-NN).} When selectivity $s$ is low, \textbf{the required $k$ must be large}. Specifically, to ensure an expected out-degree of $m$ valid neighbors under selectivity $s$, each node must maintain $k \approx m/s$ neighbors. This value can be prohibitively large, leading to a substantial increase in both index size and Time-to-Index (TTI).

\begin{example}
	\textit{Consider a simplified example with $k=6$ to illustrate the limitation of RNG pruning in filtered search. Let a node $u$ have candidate neighbors $V=\{v_0, v_1, v_2, v_3, v_4, v_5\}$, partitioned into two dense clusters: $(v_0, v_1, v_2)$ and $(v_3, v_4, v_5)$. RNG pruning might select only $v_0$ and $v_3$ as neighbors to represent each cluster. If a query's filter predicate renders $v_0$ as \textit{False} (invalid) and $v_1$ as \textit{True} (valid), selecting only $v_0$ cuts off the path to $v_1$. The search terminates at $v_0$, failing to retrieve the valid neighbor $v_1$. Thus, preserving all candidates is crucial for AFANNS.}
\end{example}

Clearly, building and storing a massive Nearest Neighbor Graph (NNG) directly is impractical. To address this, we propose the \textit{Maximal Clique Index} (\mci), which approximates a large NNG using two core ideas.
}

\section{A novel Maximal Clique Index} \label{sec:mc-index}


To overcome the limitations of existing methods, we propose the \textit{Maximal Clique Index} (\mci), a novel index designed for efficient and robust AFANNS. Below, we first present the key design principles of \mci, then detail its construction algorithm and analyze its complexity.

\begin{figure}[t!]
		\centering
		\begin{minipage}[b]{0.19\linewidth}
			\centering
			\vfill
			\subfigure[A kNN Graph]{\label{sfig:graph}
				\includegraphics[width=\linewidth]{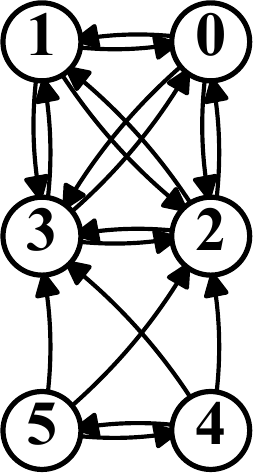}
			}
			\vfill
		\end{minipage}
		\hfill
		\begin{minipage}[b]{0.65\linewidth}
			\centering
			\vfill
			\subfigure[Three different indexes]{\label{sfig:idea}
				\includegraphics[width=\linewidth]{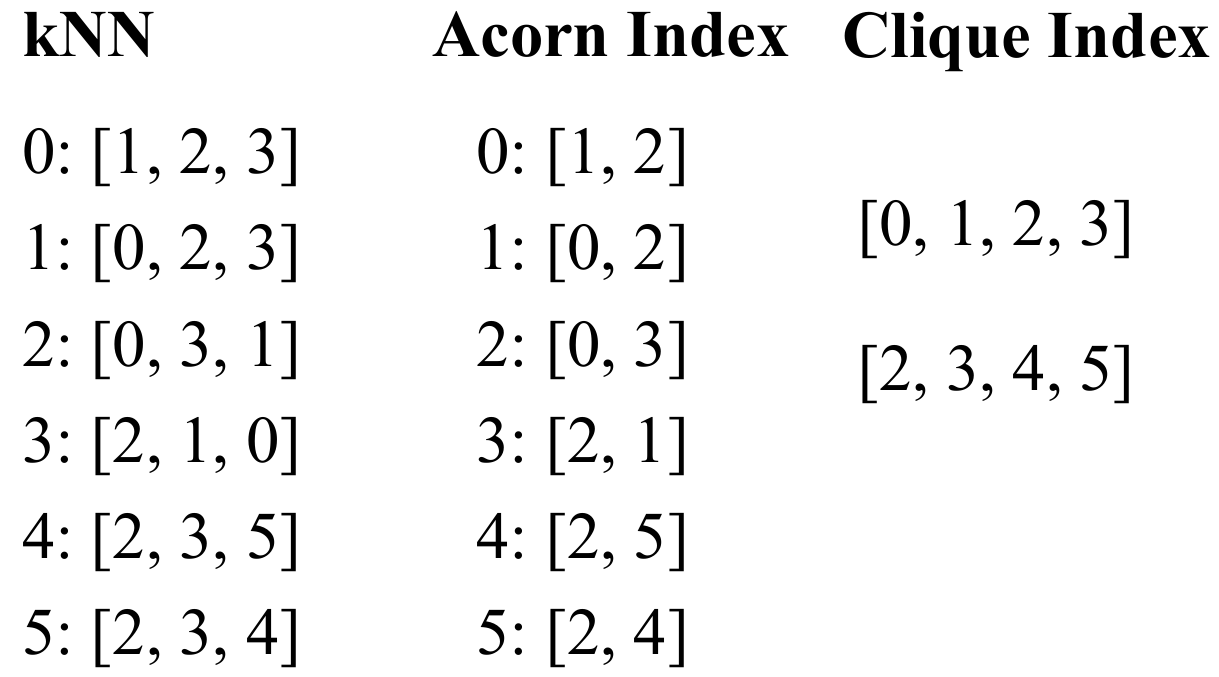}
			}
			\vfill
		\end{minipage}
        \vspace{-0.3cm}
		\caption{Illustration of \knng, \acorn index with two-hop pruning ($M_{\beta}=1$), and our maximal clique index.}  \vspace{-0.3cm}
		\label{fig:idea}
\end{figure}

\subsection{Key ideas of \mci}\label{subsec:coreIdeas}
Traditional proximity graphs optimize for unfiltered search by controlling node degree and neighbor distribution. A common technique is \textit{Relative Neighborhood Graph} (RNG) pruning \cite{RNG1980}, widely used in indexes like HNSW \cite{HNSW} and NSG \cite{NSG}. For a node $u$ and candidate neighbor $v$, the edge $(u, v)$ is retained only if no other node $w$ satisfies both $d(u, w) \le d(u, v)$ and $d(w, v) \le d(u, v)$. This ensures sparse, well-distributed edges, accelerating traversal in standard settings \cite{survey21}.

However, in AFANNS, query-specific predicates can invalidate neighbors selected by RNG, fragmenting the induced subgraph over valid nodes and halting search prematurely \cite{ACORN}. To maintain connectivity under arbitrary filters, a fundamental approach is to increase the \emph{average degree} of nodes in the proximity graph. With a sufficiently large average degree, the index can preserve reachability within the filtered subset \emph{with high probability}, even when many edges are invalidated by ad-hoc predicates. 

Unfortunately, naively increasing the degree by explicitly storing more neighbors incurs prohibitive storage overhead and index construction time. A key challenge, therefore, is to effectively \emph{amplify} the graph's connectivity without explicitly storing a dense adjacency structure.

To this end, we propose a novel solution based on maximal cliques. Our method, the \textbf{M}aximal \textbf{C}lique \textbf{I}ndex (\mci), achieves high effective connectivity through a compact, clique-based representation, which we detail in the following.



\comment{
\begin{example}
    Consider node $u$ with six candidate neighbors $\{v_0,\dots,v_5\}$ forming two dense clusters: $\{v_0,v_1,v_2\}$ and $\{v_3,v_4,v_5\}$. RNG pruning might retain only $v_0$ and $v_3$ as representative neighbors. If a query’s predicate invalidates $v_0$ but validates $v_1$, the path to $v_1$ is lost, causing search failure. Preserving all candidates avoids this issue but is infeasible at scale.
\end{example}
}

\stitle{Maximal Clique Cover (\mcc).} 
Our approach exploits a fundamental property of high-dimensional vector spaces: spatial proximity exhibits strong \emph{transitivity}, meaning that \emph{neighbors of neighbors are highly likely to be neighbors themselves} \cite{NN-Decent}. Formally, if $d(a,b)$ and $d(b,c)$ are small, then $d(a,c)$ is also likely small. Consequently, mutually proximate nodes in a $k$-NNG tend to form \emph{cliques}.

Cliques offer inherent compression benefits. In a standard adjacency-list representation, a clique of size $c$ requires storing $c(c-1)/2$ edges. By representing the clique implicitly via its $c$ member nodes, we achieve a space reduction from $O(c^2)$ to $O(c)$, corresponding to a compression factor of $(c-1)/2$. This property enables a dense graph index to be compactly represented as a collection of cliques.

However, mining all maximal cliques is inefficient because their number can be exponential (up to $O(3^{n/3})$) and they often overlap heavily \cite{DBLP:journals/tcs/TomitaTT06}. Storing every maximal clique would lead to significant redundancy. Therefore, we focus on a subset that captures the essential connectivity: the \textit{Maximal Clique Cover}.

\begin{definition}[Maximal Clique Cover (\mcc)]\label{def:mcc}
	A set of maximal cliques is a Maximal Clique Cover if every node in the graph is contained in at least one maximal clique in the set.
\end{definition}

A \mcc compactly represents all maximal cliques, providing full node coverage without redundancy. However, small maximal cliques (e.g., of size 3) often carry limited structural information and may introduce noise into the graph index by linking weakly associated nodes. To prioritize meaningful connectivity, we introduce a size threshold $\tau$ and define a refined structure:

\begin{definition}[$\tau$-Maximal Clique Cover (\tmcc)]\label{def:tmcc}
	Let $\tau$ be a positive integer. A \mcc is a \tmcc if every maximal clique in the set has a size $|C| \ge \tau$.
\end{definition}

Note that Definition~\ref{def:tmcc} guarantees each node participates in at least one clique of size at least $\tau$, implying a lower bound of $(\tau-1)$ on its degree within that clique. In practice, a node often belongs to multiple cliques, which collectively yield a high effective degree in the compressed index. By filtering out small cliques, the \tmcc retains only cohesive, information-rich subgraphs, thereby enhancing both search precision and traversal efficiency. 

Our \mci is formally a \tmcc, and the neighborhood of a node $u$ is defined as the union of the node sets of all maximal cliques containing $u$. The following example illustrates the concept.

\begin{example}\textit{
     Consider a \knng with $n=6$ and $k=3$, as shown in \textit{Figure~\ref{sfig:graph}}. Storing this graph via an adjacency list requires 18 integers (Figure~\ref{sfig:idea}, left). For comparison, methods like \acorn approximate two-hop neighborhoods; e.g., node 0's effective neighbors include $\{1, 2\}$ plus the neighbors of 1 and 2. In contrast, our \mci encodes the graph using only two maximal cliques: $C_1=\{0, 1, 2, 4\}$ and $C_2=\{2, 3, 4, 5\}$. This representation requires only 8 integers (Figure~\ref{sfig:idea}, right) to capture all adjacency relationships, while achieving a high average effective degree. For instance, node 3 belongs to both $C_1$ and $C_2$, resulting in a high effective degree of $5$.}
\end{example}

While a \tmcc can be computed efficiently (e.g., via a linear-time greedy algorithm that iteratively extracts and removes a maximal clique until all nodes are covered \cite{CoreCliqueRemoval}), a fundamental challenge remains: to guarantee coverage by cliques of size at least $\tau$, the underlying \knng must be sufficiently dense. Constructing such a dense \knng directly is prohibitively expensive. In the following, we address this by showing how to obtain a high-quality \tmcc from a much sparser $k'$-NNG (with $k' \ll k$).

\stitle{Geometric densification of neighborhood graphs.}
We construct a \tmcc from a sparse $k'$-NNG (with $k' \ll k$) by using a geometric densification technique. For each node $u$, let $V' = \mathcal{N}_{k'}(u) \cup \{u\}$ be its local neighborhood. To obtain cliques of size at least $\tau$ within $V'$, we employ a geometric densification strategy: we progressively lower the distance threshold for forming edges, thereby gradually adding connections to $G(V')$.

Formally, let $x$ be the top-1 nearest neighbor of $u$. For a parameter $\alpha > 1$, we define a distance threshold $t = \alpha \cdot d(u, x)$ and construct an edge set $E_\alpha = \{(v_i, v_j) \mid v_i, v_j \in V',\; d(v_i, v_j) \le t\}$. Maximal cliques are then greedily mined from the induced subgraph $G(V', E_\alpha)$.

A single fixed $\alpha$, however, may not yield a subgraph in which every node belongs to a clique of size $\ge \tau$ (i.e., a valid \tmcc for $V'$). To maximize coverage, we iteratively increase $\alpha$: starting from an initial value (e.g., $\alpha_0 = 1.2$), we double $\alpha$ in each round, greedily mine cliques from the corresponding denser subgraph $G(V', E_\alpha)$, and continue until every node in $V'$ is covered by at least one maximal clique of size $\ge \tau$. This process is repeated for all nodes $u$ in the vector dataset to build the global \mci. Below, we provide a formal analysis of the underlying rationale for this geometric densification strategy.

A key property of high-dimensional data is the \textit{Distance Concentration Phenomenon}: pairwise distances tend to concentrate around their mean with relatively small variance \cite{Vershynin_2018}. For analytical clarity, we model this phenomenon using a Gaussian distribution $\mathcal{N}(\mu, \sigma^2)$ with mean $\mu$ and variance $\sigma^2$. Theorem~\ref{the:nnarenn} leverages this model to establish a strong local transitivity property under high distance concentration.

\begin{theorem}[Local Connectivity under Distance Concentration]
	\label{the:nnarenn}
	Let $V = \{x_1, \dots, x_n\} \subset \mathbb{R}^d$ be a vector dataset whose pairwise distances are distributed as $\mathcal{N}(\mu, \sigma^2)$, with $\frac{\mu}{\sigma} > \sqrt{2\ln n}$. For a node $u \in V$, denote its nearest neighbor by $x$. If $v, w \in \mathcal{N}_k(u)$ satisfy $d(v, w) \le \alpha \, d(u, x)$ for a small constant $\alpha$, then, with probability approaching $1$ as $n \to \infty$, $v$ is also a $k$-nearest neighbor of $w$.
\end{theorem}

\begin{proof}
	Let $y \in V \setminus \{v, w\}$ be an arbitrary node. For a value $\delta < \mu$, let $\delta' = \frac{\delta - \mu}{\sigma}$. The probability that $d(y, w) \le \delta$ is given by the cumulative distribution function (CDF) of the standard Gaussian distribution:
	\[
	p(\delta) = \mathbb{P}(d(y, w) \le \delta) = \Phi(\delta').
	\]
	
	First, we bound the distance to the nearest neighbor $d(u, x)$. The probability that the nearest neighbor distance exceeds $\mu - t$ (for some $t > 0$) corresponds to the event where all $n-1$ other nodes are at a distance greater than $\mu - t$ from $u$:
	\[
	\mathbb{P}(d(u, x) \ge \mu - t) = \left( 1 - p(\mu - t) \right)^{n-1} = \left( 1 - \Phi\left( -\frac{t}{\sigma} \right) \right)^{n-1}.
	\]
	Consequently, the probability that $d(u, x) < \mu - t$ is:
	\[
	\mathbb{P}(d(u, x) < \mu - t) = 1 - \left( 1 - \Phi\left( -\frac{t}{\sigma} \right) \right)^{n-1}.
	\]
	
	We select $t$ such that the term $\Phi(-t/\sigma)$ is sufficiently large to make the expression approach 1:
	\[
	\alpha(\mu - t) = \mu - \sigma\sqrt{2\ln n}.
	\]
	Solving for $t$, we get $t = \mu - \frac{\mu - \sigma\sqrt{2\ln n}}{\alpha}$. Note that since $\frac{\mu}{\sigma} > \sqrt{2\ln n}$, we have $\mu - t > 0$.
	
	Consider the bound for $d(u, x)$. As $n$ grows large, the minimum of $n$ Gaussian variables concentrates near the lower tail. Specifically, for our chosen $t$, the probability that $d(u, x) \le \mu - t$ approaches 1. 
	
	Given the condition $d(v, w) \le \alpha d(u, x)$, with high probability we have:
	\[
	d(v, w) \le \alpha(\mu - t) = \mu - \sigma\sqrt{2\ln n}.
	\]
	
	Now, let $p'$ be the probability that a random node $y$ is closer to $w$ than $v$ is (i.e., $d(y, w) \le d(v, w)$). Using the bound derived above:
	\[
	p' = \mathbb{P}(d(y, w) \le d(v, w)) \le \mathbb{P}\left( d(y, w) \le \mu - \sigma\sqrt{2\ln n} \right).
	\]
	This is equivalent to evaluating the CDF at $z = -\sqrt{2\ln n}$:
	\[
	p' \le \Phi\left( -\sqrt{2\ln n} \right).
	\]
	Using Mill's Inequality \cite{10.1214/aoms/1177731603}, which states that $\Phi(z) \le \frac{1}{|z|\sqrt{2\pi}} e^{-z^2/2}$ for $z < 0$, we have:
	\[
	p' \le \frac{1}{\sqrt{2\ln n}\sqrt{2\pi}} e^{-\frac{2\ln n}{2}} = \frac{1}{2n\sqrt{\pi \ln n}}.
	\]
	
	Finally, let $N_{closer}$ be the number of nodes in $V \setminus \{v, w\}$ that are closer to $w$ than $v$ is. The expectation is $\mathbb{E}[N_{closer}] = (n-2)p'$. As $n \to \infty$:
	\[
	(n-2)p' \approx n \cdot \frac{1}{2n\sqrt{\pi \ln n}} = \frac{1}{2\sqrt{\pi \ln n}} \to 0.
	\]
	Since the expected number of nodes closer to $w$ than $v$ approaches 0, the probability that $v$ is among the top-$k$ neighbors (i.e., $N_{closer} \le k-1$) approaches 1:
	\begin{equation}
		\begin{aligned}
		\mathbb{P}(v \in \mathcal{N}_{k}(w)) = \mathbb{P}(N_{closer} \le k-1)  \\= \sum_{i=0}^{i=k-1} { n-2\choose i } {p'}^i (1-p')^{n-2-i}\to 1.
		\end{aligned}\nonumber
	\end{equation}
    This completes the proof.
\end{proof}

Theorem~\ref{the:nnarenn} implies that in high-dimensional spaces, the $k$-nearest neighbors of a node are highly likely to be mutually proximate. Consequently, these neighbors tend to form cliques within the local neighborhood. This theoretical insight directly supports our geometric densification strategy: by progressively increasing the distance threshold (increasing $\alpha$), we systematically induce such cliques, enabling the efficient mining of the \tmcc from a initially sparse neighborhood graph.

\comment{
\begin{figure}[t!]
	\begin{center}
	\end{center}
	\includegraphics[width=0.5\linewidth]{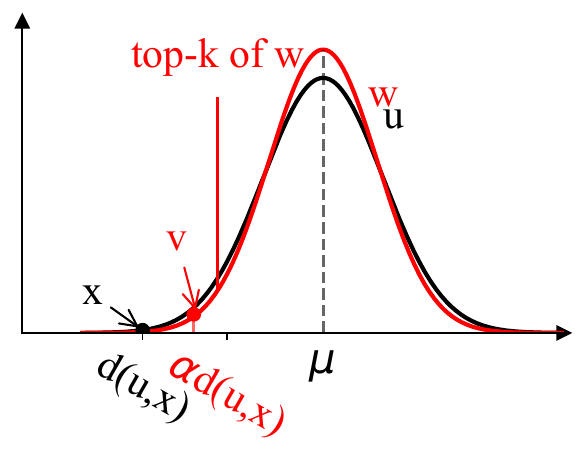}
	\caption{Example for Theorem~\ref{the:nnarenn}}
	\label{fig:distribution} 
\end{figure}
\begin{example}\textit{
	Figure~\ref{fig:distribution} illustrates Theorem~\ref{the:nnarenn}. The plot depicts the density distributions of distances from nodes to $u$ and $w$, respectively, with the x-axis representing distance and the y-axis representing density. Both distributions are concentrated around the expectation $\mu$. Let $x$ be the closest node to $u$. Given the relationship $d(v,w)=\alpha d(u,x)$, the figure demonstrates that $v$ is sufficiently close to $w$ such that $v \in \mathcal{N}_k(w)$.}
\end{example}
}

More importantly, our local \tmcc mining strategy produces an index whose average out-degree far exceeds the initial parameter $k'$. This is due to a \textit{dual coverage mechanism}: a node $u$ is covered not only by cliques mined from its own local neighborhood $V'_u = \mathcal{N}_{k'}(u) \cup \{u\}$, but also by cliques mined from the neighborhoods $V'_v$ of other nodes $v$ for which $u \in \mathcal{N}_{k'}(v)$. This cooperative coverage enables the \mci to capture rich global connectivity from only local $k'$-NN subgraphs.

\stitle{Remark (doubling $\alpha$).} Doubling $\alpha$ at each step is designed to manage the size of the densified subgraph efficiently. In a growth-restricted metric space, the number of points within a ball of radius $2t$ is at most a constant factor $c$ times the number within radius $t$, i.e., $|B(u, 2t)| \le c \cdot |B(u, t)|$ \cite{NN-Decent}. By doubling the distance threshold $\alpha \cdot d(u, x)$ in each iteration, we ensure that the edge set $E_\alpha$—and hence the graph we mine for cliques—grows in a controlled manner, thereby preserving the efficiency of each clique-mining iteration.

\begin{figure}[t!]
	\vspace*{-0.4cm}
	\begin{center}
		\subfigure[Isolated nodes]{\label{sfig:exa_e1}\includegraphics[width=0.46\linewidth]{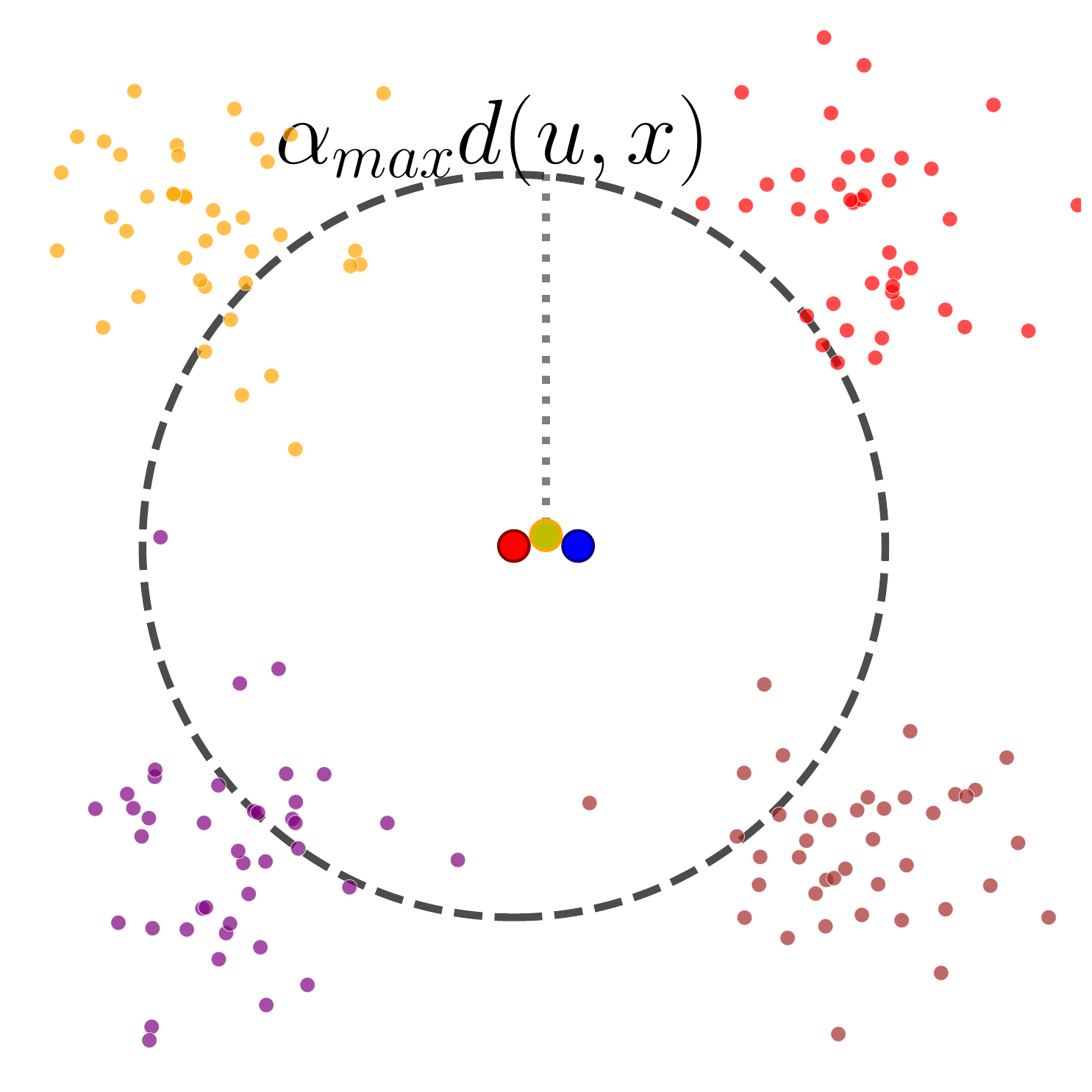}}
		\subfigure[Super center nodes]{\label{sfig:exa_e2}\includegraphics[width=0.46\linewidth]{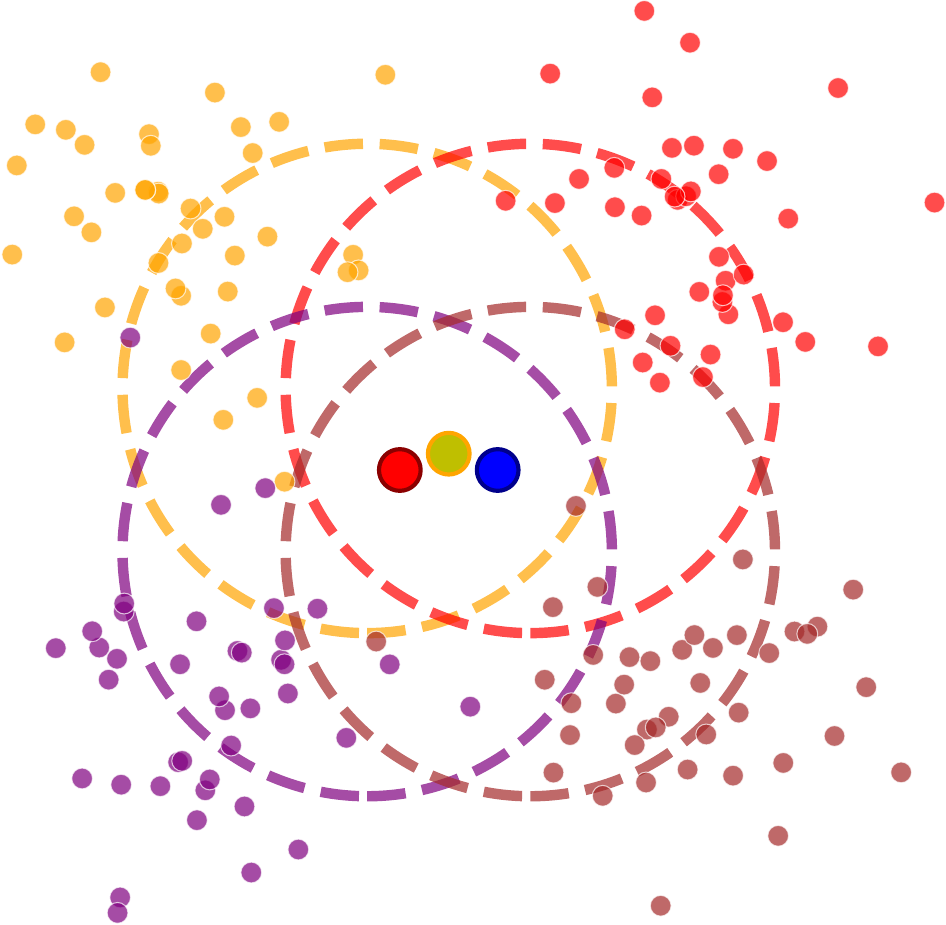}}
	\end{center}
	\vspace*{-0.4cm}
	\caption{Illustration of handling two special cases.}
	\label{fig:exa_early} 
	\vspace*{-0.4cm}
\end{figure}

\begin{figure*}[t]
\vspace*{-0.4cm}
	\centering  
	\begin{minipage}[b]{0.24\linewidth}
		\centering  
		\vfill  
		\subfigure[The distance matrix of example vectors; the $k'$-NNs in each row are red. ]{\label{sfig:matrix}
			\includegraphics[width=\linewidth]{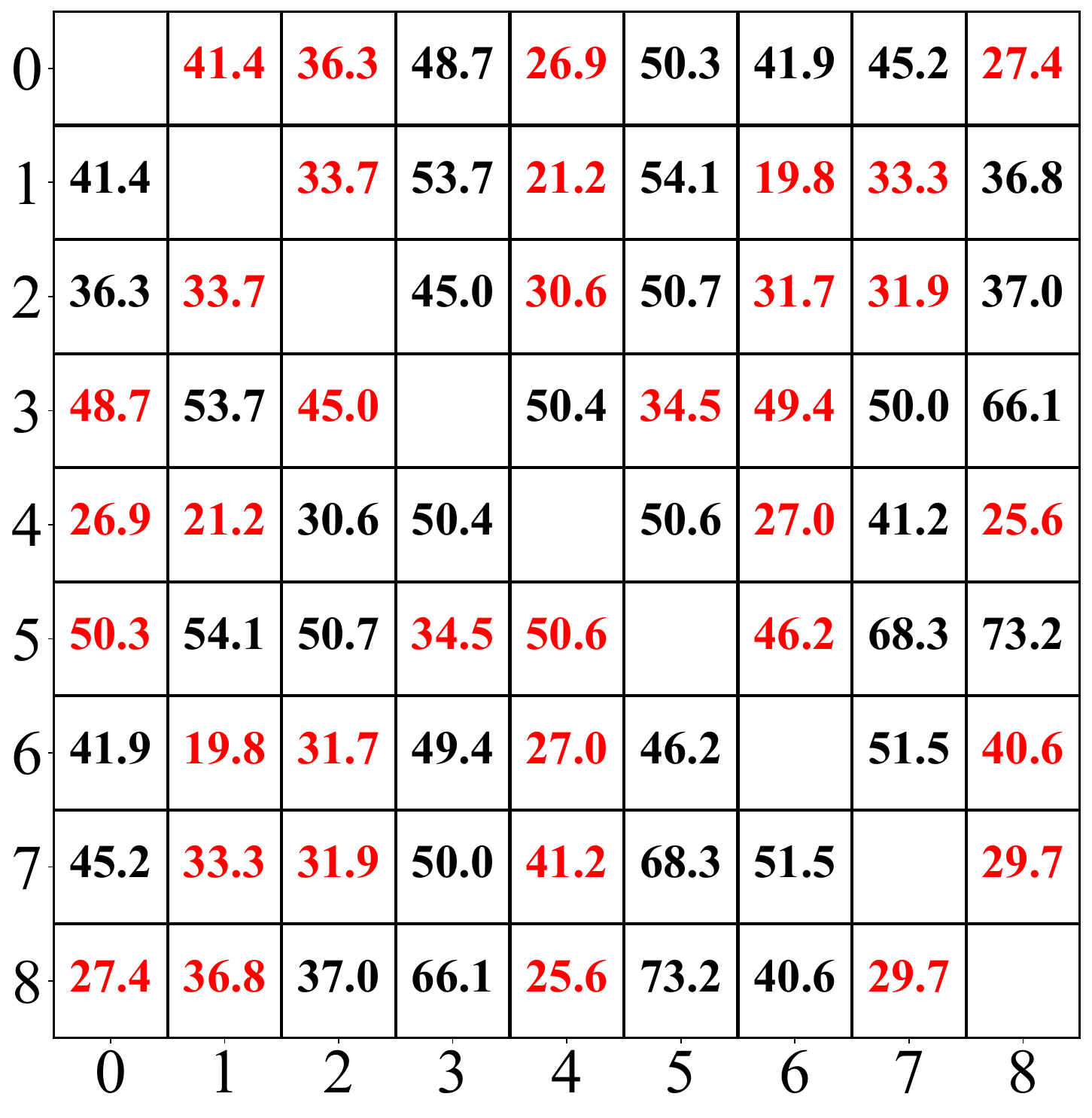}
		}
		\vfill  
	\end{minipage}
	\begin{minipage}[b]{0.49\linewidth}
		\centering
		\vfill
		\subfigure[Local graph of node $0$ with $\alpha=1.2$]{\label{sfig:cliques0}
			\includegraphics[width=0.25\linewidth]{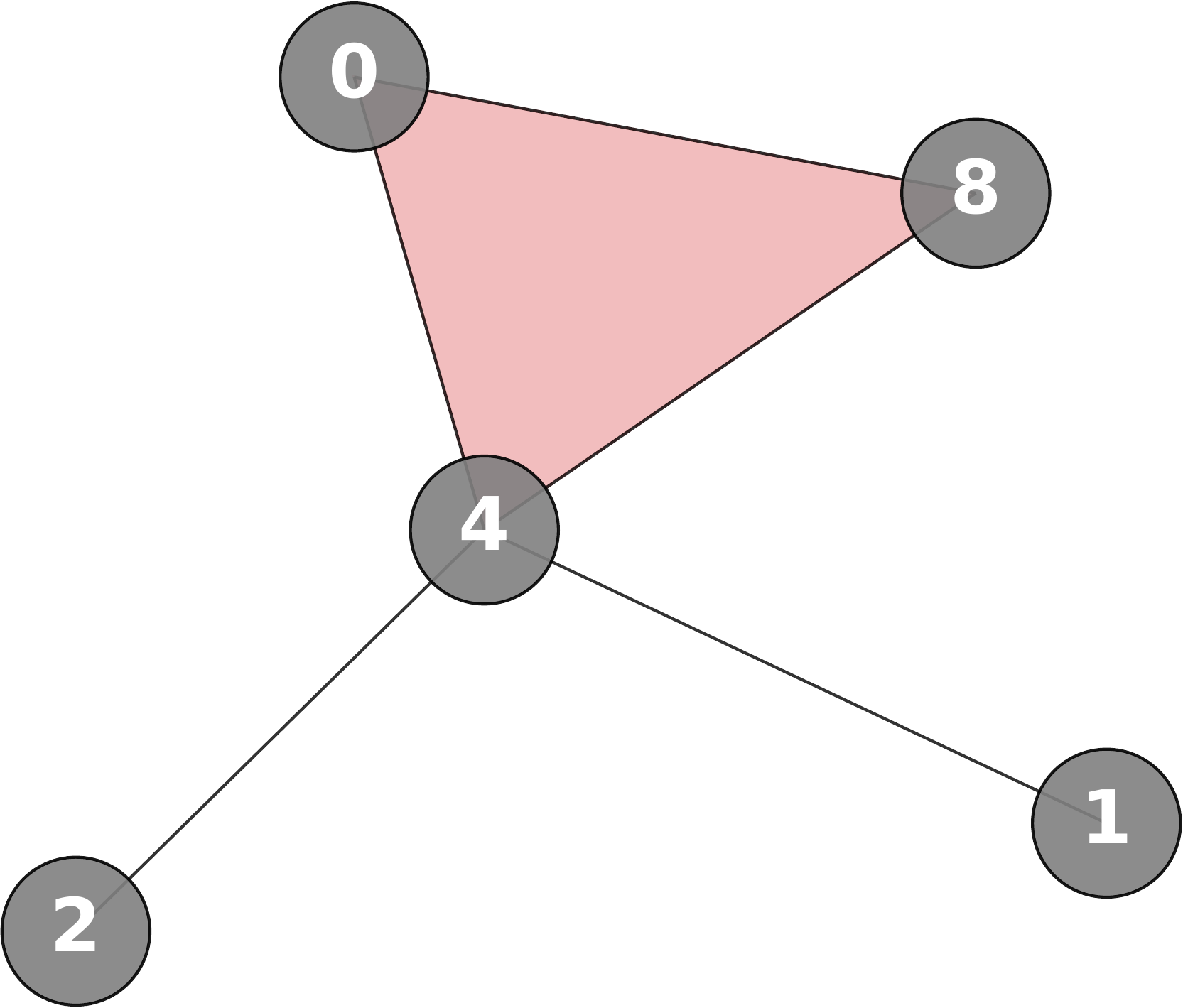}
		}
		\subfigure[Local graph of node $1$ with $\alpha=1.2$]{\label{sfig:cliques1}
			\includegraphics[width=0.25\linewidth]{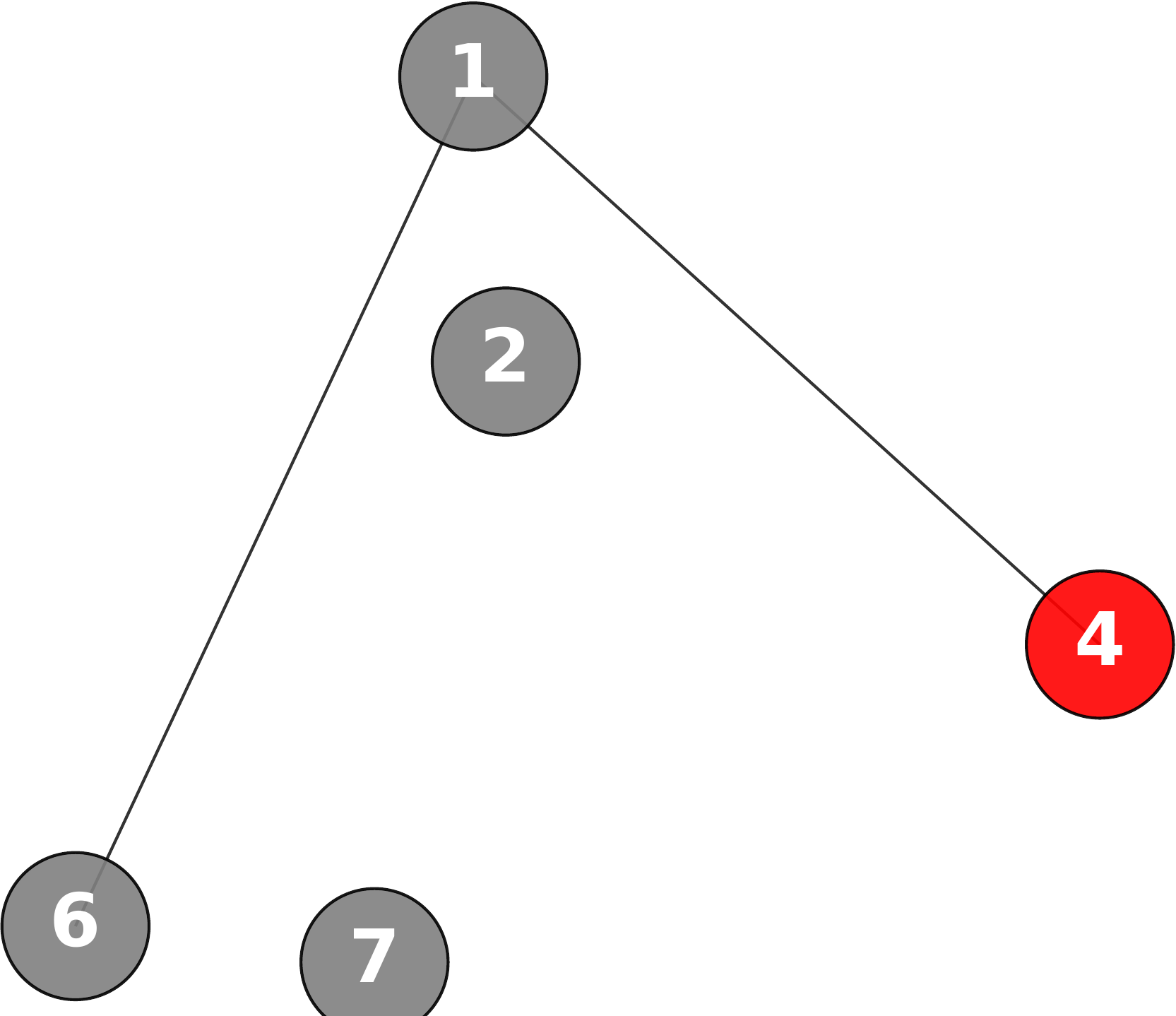}
		}
		\subfigure[Local graph of node $2$ with $ \alpha=1.2$]{\label{sfig:cliques2}
			\includegraphics[width=0.25\linewidth]{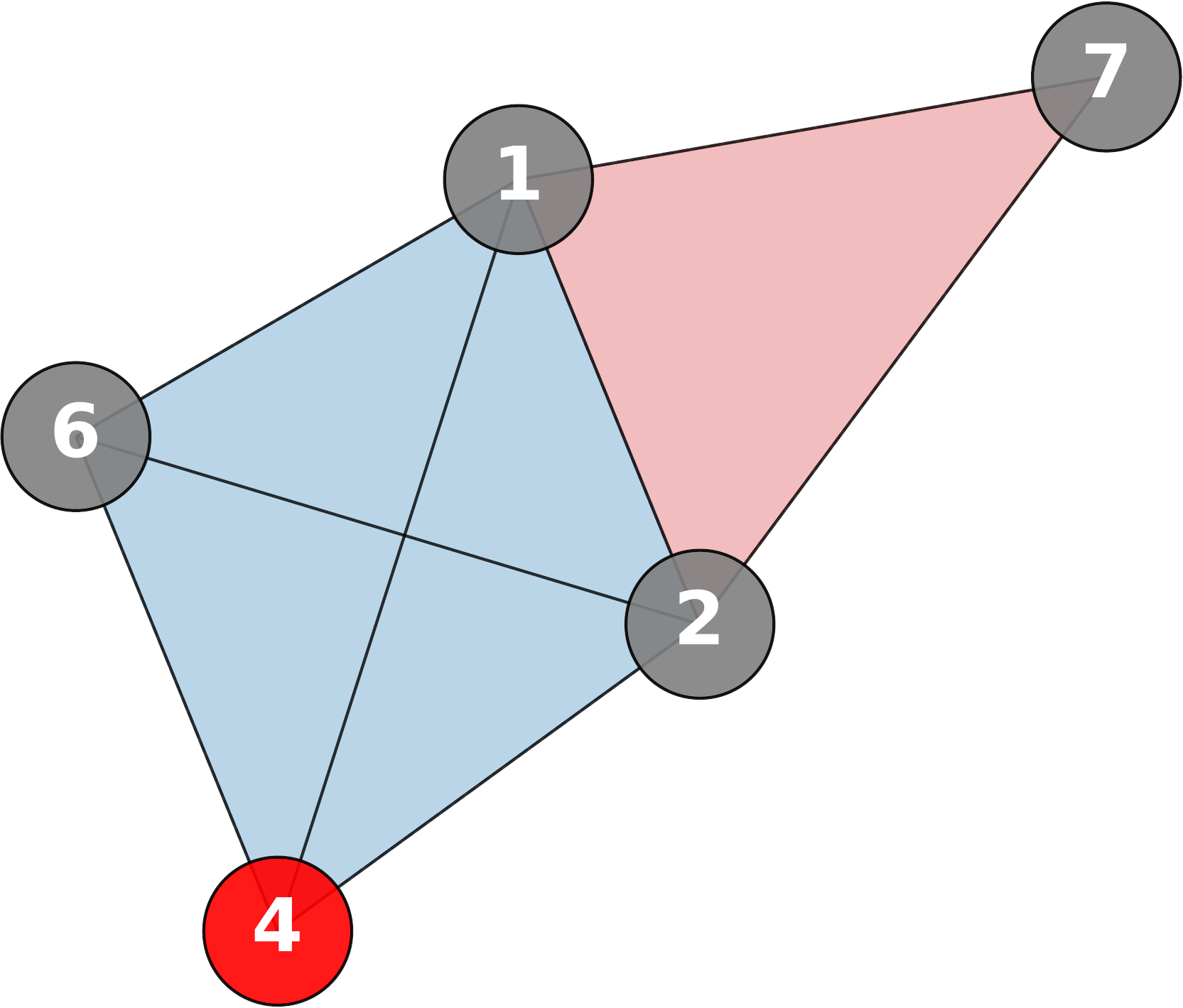}
		}
		\subfigure[Local graph of node $3$ with $ \alpha=1.2$]{\label{sfig:cliques3}
			\includegraphics[width=0.25\linewidth]{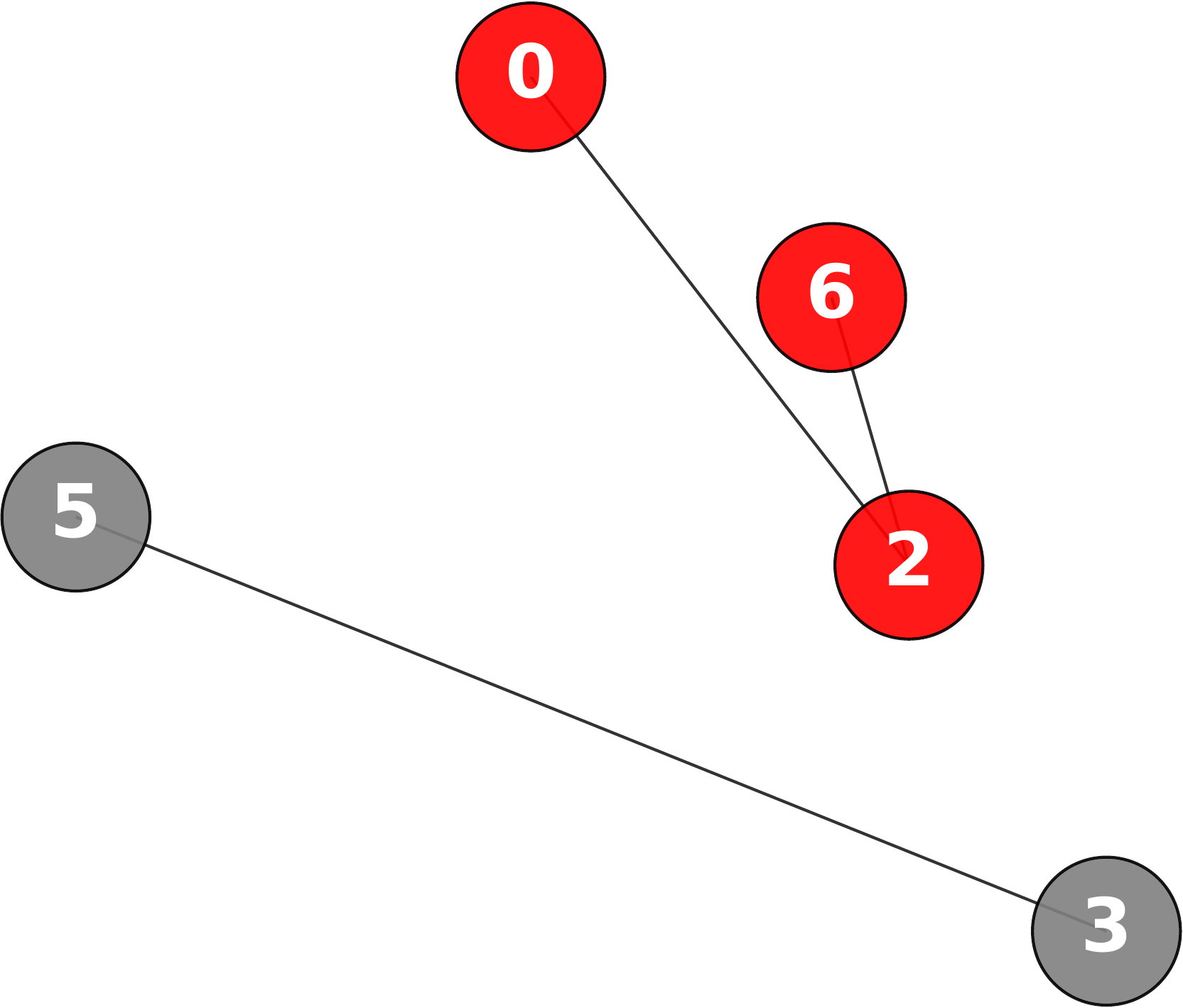}
		}
		\subfigure[Local graph of node $5$ with $ \alpha=1.2$]{\label{sfig:cliques4}
			\includegraphics[width=0.25\linewidth]{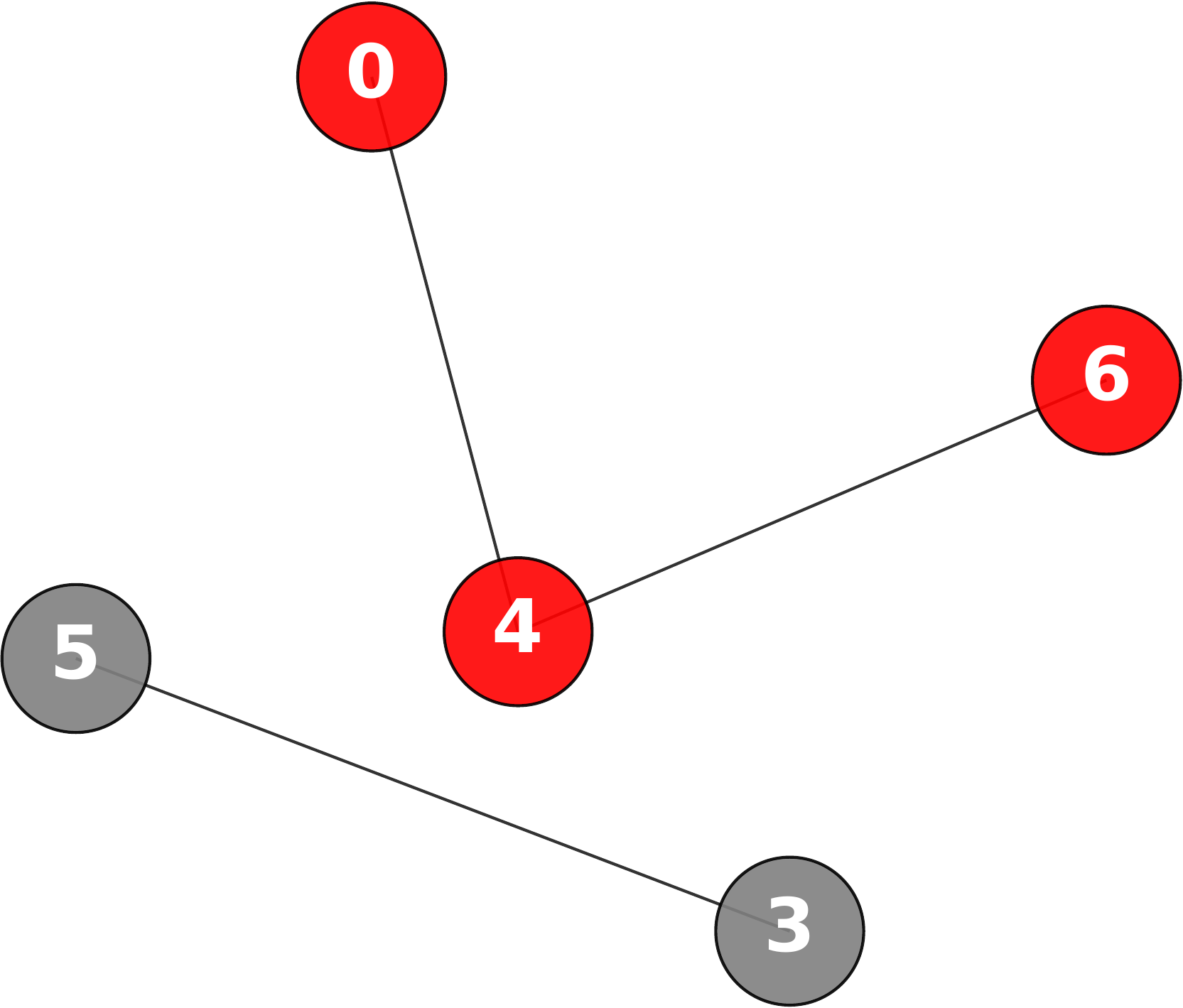}
		}
		\subfigure[Local graph of node $3$ with $ \alpha=2.4$]{\label{sfig:cliques5}
			\includegraphics[width=0.25\linewidth]{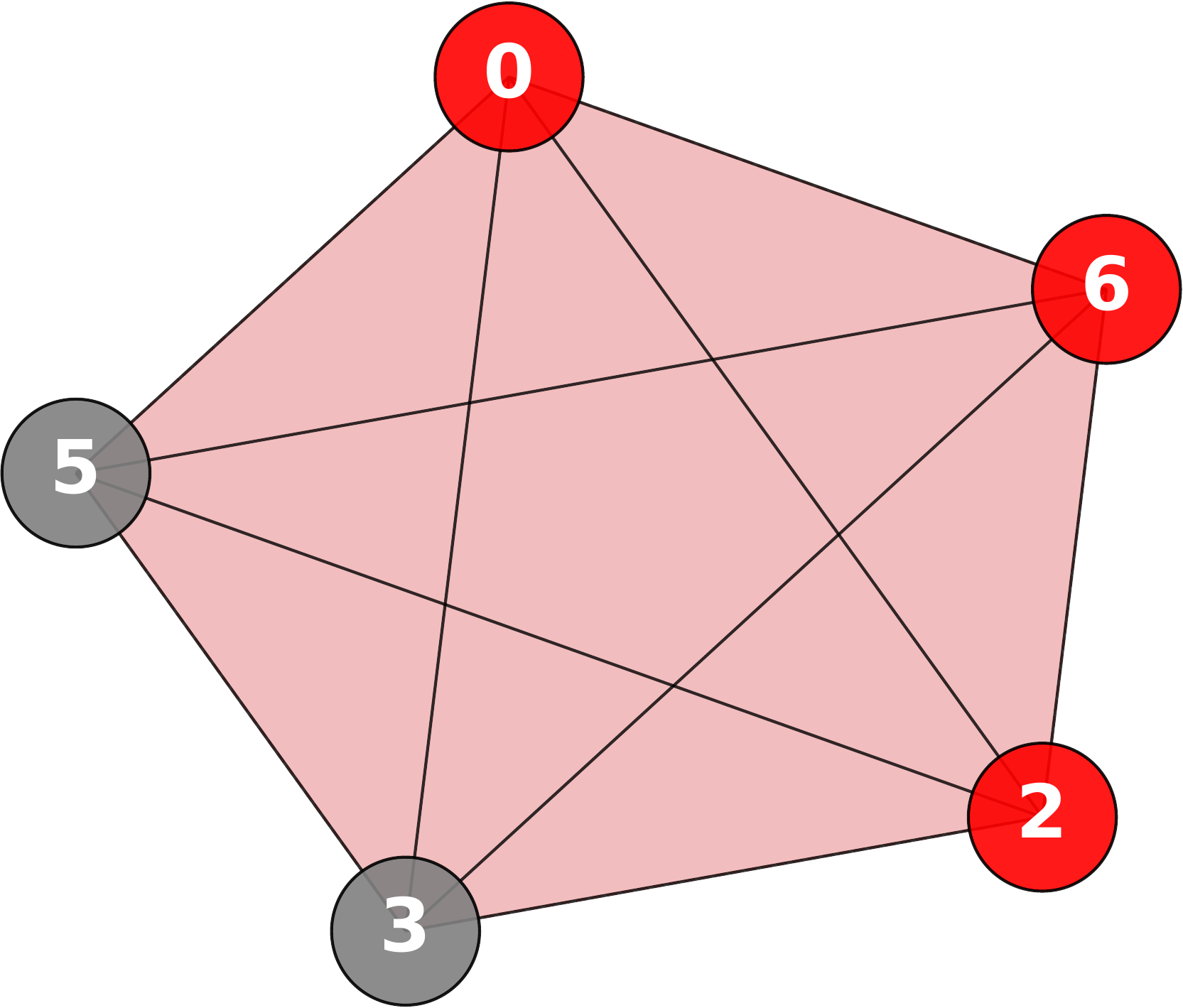}
		}
		
		\vfill
	\end{minipage}
	\begin{minipage}[b]{0.13\linewidth}
		\centering
		\vfill
		\subfigure[The final \mci]{\label{sfig:cliquestotal}
			\includegraphics[width=0.9\linewidth]{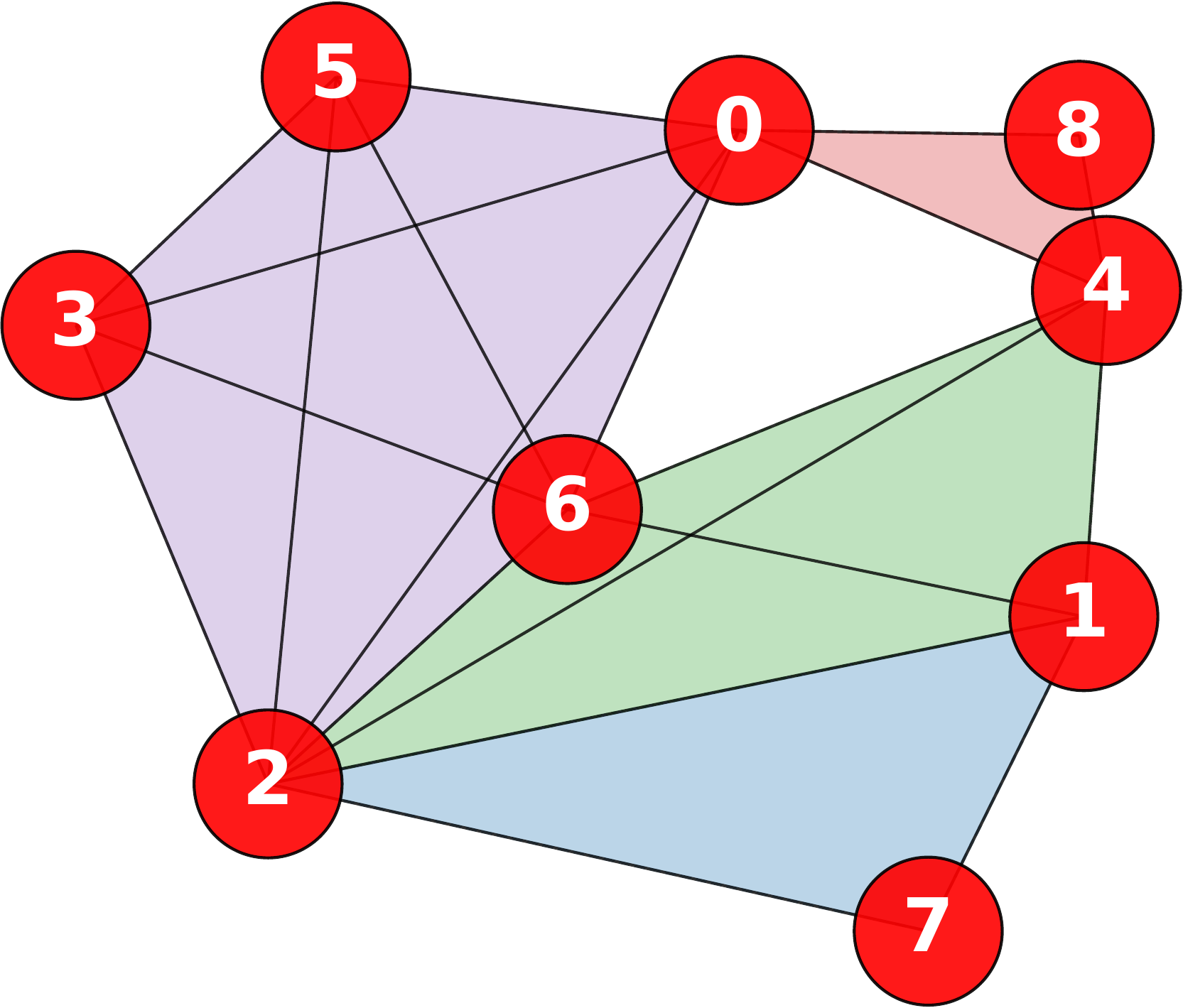}
		}
		\subfigure[NN (all numbers) / approximated NN by \mci (red numbers)]{\label{sfig:knnexample}
			\includegraphics[width=0.75\linewidth]{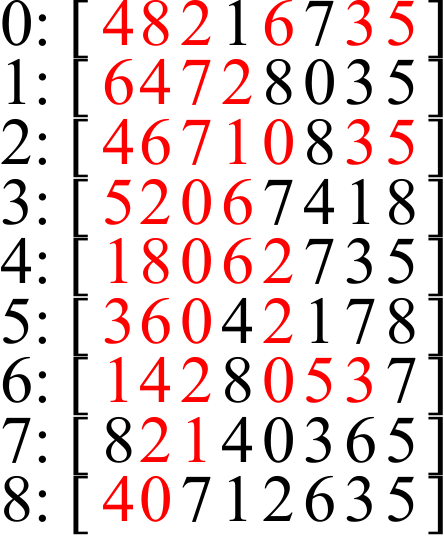}
		}
		\vfill
	\end{minipage}	
    \vspace{-0.5cm}
	\caption{Illustration of the \mci construction algorithm with $n=9, k'=4, \tau=3$}
	 \vspace{-0.3cm}
    \label{fig:exa_construction}
\end{figure*}

\subsection{The \mci construction algorithm} \label{subsec:index-construct-alg}

\begin{algorithm}[t]
	\caption{ The \mci Construction Algorithm}
	\label{alg:construction}
	\small
	\KwIn{The set of vectors $V=\{v_1,v_2,...,v_n\}$, three integers $k', \tau, \alpha_{max}$  }
	\KwOut{The maximal clique index $\mathcal{M}$}
	\SetKwProg{Fn}{Procedure}{}{}
	
	Construct a $k'$-NNG with a small $k'$ by invoking NN-Descent \cite{NN-Decent}\; 
	$I \gets \emptyset$;
	$\mathcal{M}\leftarrow \{\}$\;
	Initialize $\alpha \leftarrow 1.2$\;
	
	\While{$|I|<n$}{	
		\For{$i\in [1, n]$ s.t. $v_i\notin I$}{
			$\mathcal{M}', I' \gets  \minCliques(i, \alpha, I)$\;
			
			$\mathcal{M} \gets \mathcal{M}_j\cup \mathcal{M}'$; $I\gets I'$\;
		}
		
		$\alpha \gets \alpha \times 2$\;
	} 
	
	\Return{$\mathcal{M}$;}

	\Fn{$\minCliques(i,\alpha,  I)$}{
		$V'\gets \mathcal{N}_{k'}(v_i)\cup\{ v_i\}$ \;
		
		$g\gets $ a graph where each edge represents a pair of vectors $ (v_a\in V',v_b\in V')$ such that $d(v_a,v_b)\le \alpha \min_{v\in V'}{d(v, v_i)}$\;
		$\mathcal{M}'\gets \emptyset$\;
		\For{$v_j \in V'$ s.t. $v_j\notin I$} {
			Find a maximal clique $C$ in $g$ with $v_j \in C$ using the linear-time greedy algorithm \cite{CoreCliqueRemoval}\; 
			\If{$ |C|\ge \tau$}{
				$\mathcal{M}'\gets \mathcal{M}'\cup \{C\}$;
				$I\gets I\cup C$\;
			}

			%
		}
		
		\If{$\mathcal{M'} = \emptyset$ and $\alpha\ge \alpha_{max}$ }{
			$\mathcal{M}'\gets \mathcal{M}'\cup \{V'\}$; $I\gets I \cup V'$\;
		}
		
		\Return{$\mathcal{M}',I$;}
	}
	
\end{algorithm}

Based on the key ideas presented previously, we now describe the \mci construction algorithm, outlined in Algorithm~\ref{alg:construction}. The algorithm takes four parameters: the vector dataset $V$, the neighbor count $k'$ for building an initial $k'$-NNG (used only during construction and then discarded), the minimum clique size $\tau$, and the maximum expansion factor $\alpha_{\max}$. Its output is a $\tau$-MCC, which constitutes the final \mci index. A key feature of the algorithm is that it relies solely on the vectors $V$ for index construction, making it independent of the feature sets $F$ and thus fully adaptable to arbitrary filtering queries.

The algorithm begins by building an initial $k'$-NNG for all vectors using the NN-Descent algorithm \cite{NN-Decent} (Line~1). It initializes an empty set $I$ for covered nodes and an empty set $\mathcal{M}$ for the index (Line~2). As described in Section~\ref{subsec:coreIdeas}, the algorithm employs a geometric expansion strategy controlled by a factor $\alpha$ (Line~3). The main loop continues while uncovered nodes remain (Line~4). In each iteration, the algorithm processes every uncovered node $v_i$ (Line~5), invoking the $\minCliques(i, \alpha, I)$ procedure to greedily mine maximal cliques in its local neighborhood (Line~6). The resulting cliques $\mathcal{M}'$ are added to the index $\mathcal{M}$, and the covered-node set $I$ is updated accordingly (Line~7). After processing all nodes in the current round, $\alpha$ is doubled (Line~8), expanding the connection radius for the next iteration and increasing the chance of covering remaining nodes.

The $\minCliques(i, \alpha, I)$ procedure (Lines~10--19) implements the local mining step, grounded in Theorem~\ref{the:nnarenn}. It first retrieves the $k'$ nearest neighbors of $v_i$ and includes $v_i$ itself to form the local node set $V'$ (Line~11). Let $d_{\min}$ be the distance from $v_i$ to its nearest neighbor. A local graph $g$ is then induced on $V'$, where an edge exists between $v_a$ and $v_b$ iff $d(v_a, v_b) \le \alpha \cdot d_{\min}$ (Line~12). Maximal cliques are mined from $g$ to cover the currently uncovered nodes in $V'$ (Lines~13--17). For each uncovered node $v_j \in V'$ (Line~14), the procedure attempts to find a maximal clique $C$ in $g$ that contains $v_j$ and satisfies $|C| \ge \tau$ (Lines~15--16). If such a clique is found, it is added to $\mathcal{M}'$ and the nodes in $C$ are marked as covered in $I$ (Line~17). The clique mining itself is efficient; we employ a linear-time greedy algorithm \cite{CoreCliqueRemoval} to extract a maximal clique for each candidate node.

Finally, if after processing $V'$ no cliques have been mined and $\alpha$ has reached $\alpha_{\max}$ (Line~18), we add the entire set $V'$ as a single pseudo‑clique to $\mathcal{M}'$ (Line~19). This case handles isolated nodes that remain uncovered even at the maximum radius. By treating $V'$ as a bridge, we ensure that nodes in sparse regions are still included in the index, preserving overall connectivity.

\begin{example}[Handling Isolated Nodes]
\textit{
Figure~\ref{sfig:exa_e1} illustrates a scenario where a small central cluster is isolated from larger surrounding clusters. The outer circle represents the search boundary at $\alpha=\alpha_{max}$. The nodes captured within this boundary (the set $V'$ in Line~11 of Algorithm~\ref{alg:construction}) are distributed across four distinct clusters. Crucially, the node count within any individual cluster is below the threshold $\tau$, and the pairwise distances between nodes in different clusters exceed the edge creation threshold (i.e., $\ge \alpha_{max}d(u,x)$). Consequently, the local graph contains no valid maximal cliques. To preserve global connectivity, the early‑termination mechanism (Line~18) treats the entire set $V'$ as a single pseudo‑clique, effectively bridging these isolated nodes.
	}
\end{example}

Conversely, we encounter ``super‑center'' nodes (or hubs) that belong to an excessively large number of maximal cliques. Such a hub acts as a high‑degree connector, linking all nodes across the many cliques it participates in, as detailed in Example~\ref{ex:super-center}. To mitigate the resulting query‑time overhead, we implement a simple pruning rule: if a node is already covered by more than $n/100$ cliques, it is excluded from the local candidate set $V'$ (Line~11). This cap limits the maximum connectivity influence of any single node without harming overall index quality.

\begin{example}[Super‑Center Nodes]
    \label{ex:super-center}
    \textit{
        Figure~\ref{sfig:exa_e2} depicts super‑center nodes that connect to four distinct clusters (shown as four cycles). Without pruning, these hubs would bridge every node across all four clusters. Consequently, any search traversal passing through a super‑center would require distance computations against the entire union of the connected clusters, incurring prohibitive computational cost.
    }
\end{example}

\stitle{Lock-free parallel implementation.} 
Algorithm~\ref{alg:construction} can be easily designed for efficient lock‑free parallel execution. Each thread maintains a private, independent set to store the maximal cliques it discovers (Lines~2, 7). This thread‑local storage eliminates write conflicts during parallel processing; the cliques from all threads are merged into a unified index only at the final stage (Line~9). The coverage mask $I$ (implemented as a bit‑array of length $n$) is updated via atomic operations rather than locks (Lines~7, 17, 19), which enhances scalability by avoiding lock contention. Our experimental evaluation confirms the high parallel efficiency of this design.

\stitle{Robustness to construction parameters.} 
Regardless of the specific values assigned to \(k'\) and \(\tau\), the algorithm guarantees the generation of a valid \mcc, thereby ensuring consistent structural integrity. Consequently, the search performance remains relatively stable across different parameter settings. While \(k'\) and \(\tau\) naturally govern the trade-off between construction overhead and index quality, \mci exhibits significant robustness to their variations. Experimental results in Section~\ref{sec:exp} confirm that although larger values of \(k'\) and \(\tau\) correlate with marginal gains in precision, the overall system performance is generally insensitive to fine-grained parameter tuning.

To facilitate the understanding of Algorithm~\ref{alg:construction}, we provide an illustrative example using a toy dataset below.

\begin{example}\label{exa:construction}
	\textit{
		Figure~\ref{fig:exa_construction} illustrates the \mci construction process on a randomly generated dataset with $n=9$, using parameters $k'=4$ and $\tau=3$. The pairwise distance matrix is presented in Figure~\ref{sfig:matrix}, where cells highlighted in red indicate the $k'$-nearest neighbors for each node. Figures~\ref{sfig:cliques0}--\ref{sfig:cliques3} depict the local graphs and mined cliques during the first iteration ($\alpha=1.2$). In each subfigure, the depicted nodes correspond to the local candidate set $V'$ (Line~11 of Algorithm~\ref{alg:construction}), with edges established based on the distance threshold defined in Line~12. Nodes highlighted in red indicate that they have been successfully covered by maximal cliques. For instance, the identification of the clique $\{0,4,8\}$ in Figure~\ref{sfig:cliques0} results in these nodes being marked as covered (red nodes) in subsequent steps. Maximal cliques are distinguished by fill colors; for example, Figure~\ref{sfig:cliques2} identifies two overlapping maximal cliques: $\{1,2,7\}$ and $\{1,2,4,6\}$. The final constructed \mci is visualized in Figure~\ref{sfig:cliquestotal}.
		Figure~\ref{sfig:knnexample} details the connectivity: the full neighbor lists are ordered by distance for each node, with the red numbers denoting the neighbors approximated by the \mci. Storing this connectivity explicitly as an adjacency list would require 40 integers; in contrast, our \mci representation requires only 15 integers. This yields an effective average out-degree of $40/9 \approx 4.44$, which is larger than the construction parameter $k'=4$. This degree amplification is attributable to the dual coverage mechanism discussed in Section~\ref{subsec:coreIdeas}. On larger datasets with higher $k'$ and $\tau$, the effective degree typically surpasses $k'$ by a considerably wider margin.
	}
\end{example}

\subsection{Complexity analysis} \label{subsec:time-complexity}
This subsection analyzes the complexity of  index construction.

\begin{theorem}\label{the:construction_time}
	The worst-case time complexity of Algorithm~\ref{alg:construction} is $O(T+r \cdot n \cdot (k')^2 \cdot D)$, where $O(T)$ is the time to construct the initial $k^\prime$-NN graph, $r$ denotes the number of iterations in the main loop (line~4), $n$ is the dataset size, and $D$ is the vector dimension. The number of rounds $r$ is bounded by $O(\min(\log_2{\alpha_{max}}, \log(\Delta)))$, where $\Delta$ is the ratio of the maximum to minimum pairwise distance within the local neighborhoods.
\end{theorem}

\begin{proof}
    The algorithm consists of two phases. First, building the initial $k'$-NNG costs $O(T)$ (Line~1); for NN‑Descent \cite{NN-Decent}, $T$ is near $O(n^{1.14})$ in practice. Second, the geometric expansion loop (Lines~4--9) runs for $r$ iterations. In each iteration, for every uncovered node $v_i$, we construct a local induced subgraph on its $O(k')$ neighbors (Line~12), which requires $O((k')^2 )$ distance computations. In the worst case, all $n$ nodes remain uncovered in every iteration, giving a per‑round cost of $O(n \cdot (k')^2 \cdot D)$. Since the expansion factor $\alpha$ doubles each round and is capped by $\alpha_{\max}$, the number of rounds $r$ is $O(\log_2 \alpha_{\max})$. In practice, $r$ is also limited by the local distance ratio $\Delta$, yielding the combined bound $O(\min(\log_2 \alpha_{\max}, \log \Delta))$. Summing the two phases gives the total complexity $O(T + r \cdot n \cdot (k')^2 \cdot D)$.
\end{proof}

\stitle{Empirical performance.} In practice, the first one or two rounds dominate the total runtime. This occurs because dense local structures are quickly identified with small $\alpha$, covering the majority of nodes. Subsequent rounds handle only a few remaining isolated nodes, making the observed construction time significantly lower than the worst‑case bound.

\stitle{Space complexity.}
The construction algorithm requires $O(nk')$ space to store the initial $k'$‑NNG, which is used only during index building and discarded afterward. The size of the final \mci index is bounded as follows:

\begin{theorem}\label{the:indexSize}
    The worst‑case size of the \mci index is $O(nk')$.
\end{theorem}
\begin{proof}
    Each maximal clique added to the index covers at least one previously uncovered node, so the total number of cliques is at most $n$. Every clique is mined from a local neighborhood of size $O(k')$, hence its size is also bounded by $O(k')$. Therefore, the total index size is at most $n \cdot O(k') = O(nk')$.
\end{proof}



The worst-case bound of $O(nk')$ arises in the extreme scenario where there are $n$ maximal cliques, each of size $k'$. However, such a extreme case is rarely encountered in practice. Typically, the number of maximal cliques is far fewer than $n$, and their sizes are well below $k'$. Our experiments demonstrate that the actual index size is significantly superior to the worst-case theoretical bound (see Exp-2 in Section~\ref{sec:exp}).

\begin{algorithm}[t]
	\caption{ Query Processing Algorithm with \mci}
	\label{alg:search}
	\small
	\KwIn{$D=(V,F)$, $Q=(v_q,f_q,P_q)$, Maximal Clique Index $\mathcal{M}$, an integer $k$, an integer  $l_s$, and a parameter $\epsilon$}
	\KwOut{ $k$ ANNs of $Q$}
	\SetKwProg{Fn}{Procedure}{}{}
	
	$R\gets \emptyset, E\gets \emptyset, I\gets \emptyset$; \textit{/*Results, Explored, Inserted*/} \\
	$seeds\gets $ sample $\epsilon \sqrt{n}$ vectors that $\forall v_i\in seeds, P_q(f_i, f_q)=True$\;
	\For{$v\in seeds$}{
		$I\gets I\cup \{v\}$;
		$\inset(R,v)$\;
	}
	$\mathcal{M'}\gets \emptyset$; \textit{/*visted maximal cliques*/} \\ 
	\While{$R\setminus E \neq \emptyset$}{
		Let $p$ be the vector in $R\setminus E$ closest to $v_q$\;
		$E\gets E\cup \{p\}$\;
		\For{$C\in \mathcal{M}$ s.t. $p\in C$, $C\notin \mathcal{M'}$}{
			$\mathcal{M'}\gets \mathcal{M'}\cup \{C\}$\;
			\For{$v_i \in C$}{
				\If{$P_q(f_i,f_q)$ and $v_i \notin I$}{
					$I\gets I\cup \{v_i\}$;
					$\inset(R,v_i)$\;
				}
			}
		}
	}
	
	\Return{the closest $k$ vector of $v_q$ in $R$;}
	
	\Fn{$\inset(R, v)$}{
		$R\gets R\cup \{v\}$\;
		\If{$|R|> l_s$}{
			Update $R$ with the closest $l_s$ vectors to $v_q$\;
		}
	}
	
\end{algorithm}

\section{Query processing with \mci} \label{sec:query-process}
This section describes query processing with \mci for AFANNS. We first introduce a seed‑node selection strategy, then present the query processing algorithm and analyze its complexity.

\subsection{Seed nodes selection strategy}
Proximity‑graph‑based search typically starts from a fixed entry node, assuming the graph is well‑connected. This assumption fails in AFANNS because arbitrary predicates can invalidate many edges, fragmenting the induced subgraph into disconnected components. A single entry point may trap the search in one component, missing valid candidates in others.

To overcome this, we employ a multi‑seed initialization. We sample $\epsilon \sqrt{n}$ nodes uniformly at random as candidate seeds. From these, we keep the $l_s$ closest to the query vector as the actual starting points. This strategy increases the chance of hitting different connected components of the filtered subgraph.

The parameter $\epsilon$ allows adaptation to query selectivity. Our experiments show that even a small $\epsilon$ (e.g., $0.01$) works well for high‑selectivity queries, as the giant component dominates. For low‑selectivity queries, a larger $\epsilon$ (e.g., $1$) helps locate seeds in scattered components, thereby improving recall.

\stitle{Why $\epsilon \sqrt{n}$ Seeds?}
The $\sqrt{n}$ term is motivated by the observation that in most practical scenarios, query selectivity is typically at least $O(1/\sqrt{n})$. Under this condition, $\sqrt{n}$ seeds provide a good balance: enough to cover the main components, yet sub‑linear to avoid excessive overhead. While the seed sampling adds $\epsilon \sqrt{n}$ distance computations, this cost becomes negligible in high‑recall regimes where the total number of distance evaluations is much larger. The law of diminishing returns in ANNS means that achieving the last few percentage points of recall often dominates the overall cost \cite{survey21,HNSW,NSG,NHQ,SSG}; the modest seed‑selection overhead is therefore justified by the substantial gains in robustness and final recall.

\comment{\section{Query Processing with \mci} \label{sec:query-process}
This section presents the query processing technique to handle the AFANNS problem with \mci. Below, we first introduce a seed nodes selection strategy, followed by query processing algorithm ant its complexity analysis.

\subsection{Seed Nodes Selection Strategy}
Previous works leveraging proximity graphs typically adopt a fixed entry-node strategy for search initialization. This approach relies on the assumption that the graph is connected, ensuring that a single entry point suffices to traverse the structure and reach all other nodes.
However, a fixed entry-node strategy is infeasible for AFANNS, where the predicate-induced subgraph is not guaranteed to be connected. Note that \mci has no guarantees for strict connectivity for arbitrary filters \cite{survey21}. The predicate prunes \textit{False} elements, naturally fragmenting the index into isolated connected components. A fixed entry point may trap the search within a local component, preventing the retrieval of valid neighbors located in other components.

To address this connectivity  challenge, we propose a seeds sampling strategy that we sample $\epsilon \sqrt{n}$ seeds as the initial start nodes, where $\epsilon$ denotes a tunable scaling factor. Specifically, the algorithm maintain a priority queue with a maximum capacity of $l_s$, and the top-$l_s$ nearest neighbors in the $\epsilon \sqrt{n}$ seeds serve as the initial start nodes. This multi-seed initialization ensures coverage of multiple potential connected components in the predicate-subgraph.

The parameter $\epsilon$ allows for flexible adjustment based on query selectivity. Experimental results validate that the algorithm's performance is robust to variations in $\epsilon$ in high-selectivity regimes. In such cases, even a small $\epsilon$ (e.g., $\epsilon=0.01$) ensures that the sampled seeds land in the main giant component of the subgraph. Conversely, when selectivity is low, a larger $\epsilon$ (e.g., $\epsilon=1$) is necessary to locate seeds in sparse, scattered components, thereby improving a  recall.

\stitle{Why $\epsilon \sqrt{n}$ Seeds?}
The choice of $\sqrt{n}$ is grounded in the assumption that selectivity is typically at least $O(1/\sqrt{n})$, a condition holding for most practical search scenarios (if selectivity is lower, pre-filtering methods are generally superior). Mathematically, the $\sqrt{n}$ term ensures moderate growth with dataset size, avoiding  oversampling, which hurts latency.

While this technique introduces a fixed overhead of $\epsilon \sqrt{n}$ distance calculations, potentially impacting latency in low-recall or strictly real-time scenarios, it delivers a performance advantage in high-recall and low selectivity regimes. It is well-established in ANNS literature that recall improvement follows a law of diminishing returns; for instance, the computational effort required to boost recall from $98\%$ to $99\%$ may exceeds that required to improve from $80\%$ to $90\%$. In such high-precision contexts, the aggregate number of distance calculations necessary for accurate retrieval inherently exceeds $O(\epsilon \sqrt{n})$. Therefore, the initialization overhead renders itself negligible relative to the substantial gains in search robustness.
}

\subsection{The query processing algorithm} \label{subsec:query-process-alg}
Nodes belonging to the same maximal clique as a reference node $u$ exhibit strong semantic or structural correlations, rendering them high-potential candidates for the nearest neighbor search. Accordingly, for any node $u$ visited during the traversal, we identify $\mathcal{M}_u$ as the set of all maximal cliques containing $u$. The union of nodes across these cliques, defined as $\bigcup_{C \in \mathcal{M}_u} C$, constitutes the candidate set for the subsequent expansion step.

Algorithm~\ref{alg:search} outlines the pseudo-code for the \mci search procedure. It employs a greedy strategy like beam search \cite{NSG}, customized for the clique-based structure. The algorithm takes as input the dataset $D=(V,F)$, a query $Q=(v_q, f_q, P_q)$, the index $\mathcal{M}$, and three parameters: the beam width $l_s$ (where $l_s \ge k$), the target result count $k$, and the sampling factor $\epsilon$.

The process begins by initializing three sets: $R$,  the top-$l_s$ valid ($True$) candidates found so far; $E$, the set of nodes within $R$ that have already been expanded; and $I$, the set of visited nodes for which the distance to $v_q$ has been computed (Line~1).
The algorithm first samples $\epsilon \sqrt{n}$ random seeds (Line~2) and attempts to insert them into $R$ (Lines~3--4). The core loop proceeds by iteratively selecting the closest unexplored node $p \in R \setminus E$ (Lines~6--8). For each $p$, the algorithm retrieves its associated maximal cliques and expands the search to all unvisited neighbors that satisfy the predicate $P_q$ (lines~9--13). The priority queue $R$ is dynamically updated to retain only the top-$l_s$ closest valid vectors (lines~15--18). Finally, the algorithm terminates when no further candidates can be explored, returning the top-$k$ results (line~14).

The search parameters \(l_s\) and \(\epsilon\) enable a flexible trade-off between search quality and efficiency: larger values of \(\epsilon\) and \(l_{\text{search}}\) yield more accurate retrieval results, but at the cost of increased computational latency.

\stitle{Implementation details.}
To ensure efficient query performance, we adopt the optimized data structure framework established in Vamana \cite{filteredDiskann}. Specifically, the candidate sets $R$ and $E$ are maintained as ordered vectors, where each element is stored as a tuple $(v_i, d(v_i, v_q), \text{flag}_i)$ containing the node identifier, its distance to the query, and a boolean validity indicator. The visited set $I$ is implemented as a bit‑vector of length $n$ supporting $O(1)$ membership tests.

We utilize the reservoir sampling algorithm to sample the $\epsilon \sqrt{n}$ seeds. This sampling algorithm needs a computation of predicates $P_q(f_i, f_q)$ for all nodes. Generally, the predicate computation cost is quite smaller than the search process. Thus, like previous works \cite{ACORN, MochengSurvey}, we pre-compute the predicates as a boolean vector for sampling (Line~2) and searching (Line~12).  Additionally, when the predicates cost can not be ignored, we can keeps sampling without replacement until  $\epsilon \sqrt{n}$ seeds are reached, where the total sample size is expected to be ${\epsilon \sqrt{n}}/{s}$ ($s$ is the query selectivity).

\stitle{Complexity analysis.}
The search procedure comprises two distinct phases: seed initialization and greedy traversal. First, the algorithm samples $\epsilon \sqrt{n}$ seeds to ensure coverage across potentially disconnected components within the predicate-induced subgraph. This initialization incurs a computational cost of $O(\epsilon \sqrt{n})$ distance comparisons. Second, the greedy traversal operates within the subgraph of valid elements (of size $sn$). In high-dimensional vector spaces, the total search overhead is dominated by distance computations. Prior literature on $k$-NNG-based ANNS characterizes the empirical search complexity as ranging from $O(N^{0.52})$ to $O(N^{0.55})$ \cite{NSW14,survey21,DBLP:journals/corr/FuC16,DBLP:journals/tcyb/JinZHLCH14}. Given that \mci approximates a $k$-NNG structure, we estimate the expected traversal complexity to be $O((sn)^{0.55})$. Consequently, the total average-case complexity is dominated by the combination of seed sampling and graph traversal, yielding $O(\epsilon \sqrt{n} + (sn)^{0.55})$. 

\section{Experiments}\label{sec:exp}
We present a comprehensive evaluation of our method and several state‑of‑the‑art baselines for AFANNS.

\subsection{Experimental setup}

\begin{figure*}[t!]
	\vspace*{-0.5cm}
	\begin{center}
		\includegraphics[width=0.7\linewidth]{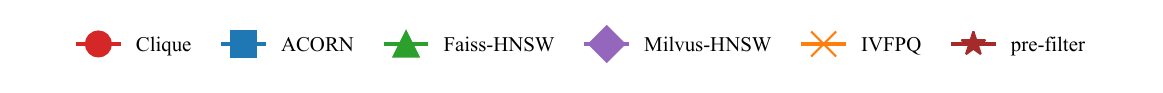}\vspace*{-0.5cm}\\
		\subfigure[sift1M]{\label{sfig:time_sift1M}\includegraphics[width=0.195\linewidth]{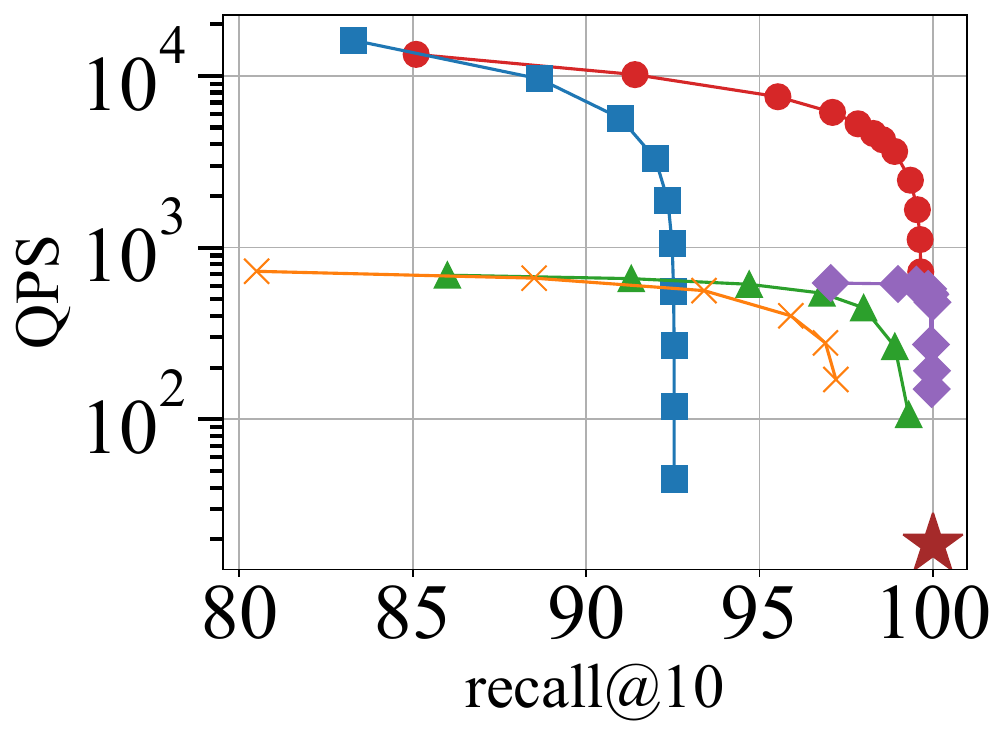}}
		\subfigure[movielens]{\label{sfig:time_movielens}\includegraphics[width=0.195\linewidth]{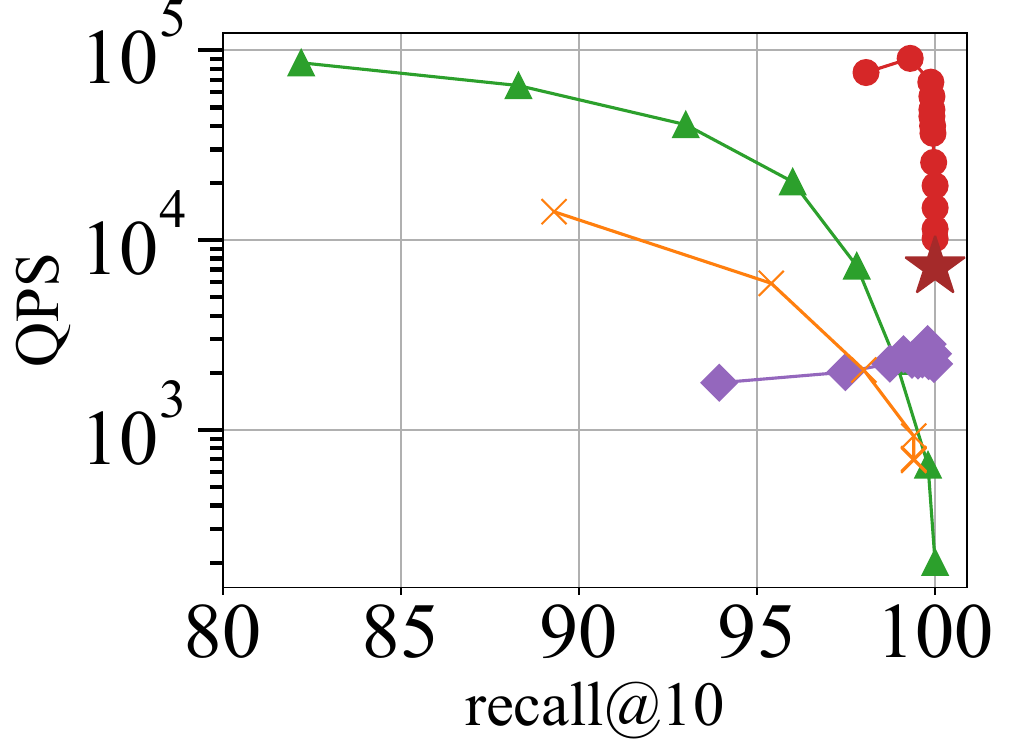}}
		\subfigure[audio]{\label{sfig:time_audio}\includegraphics[width=0.195\linewidth]{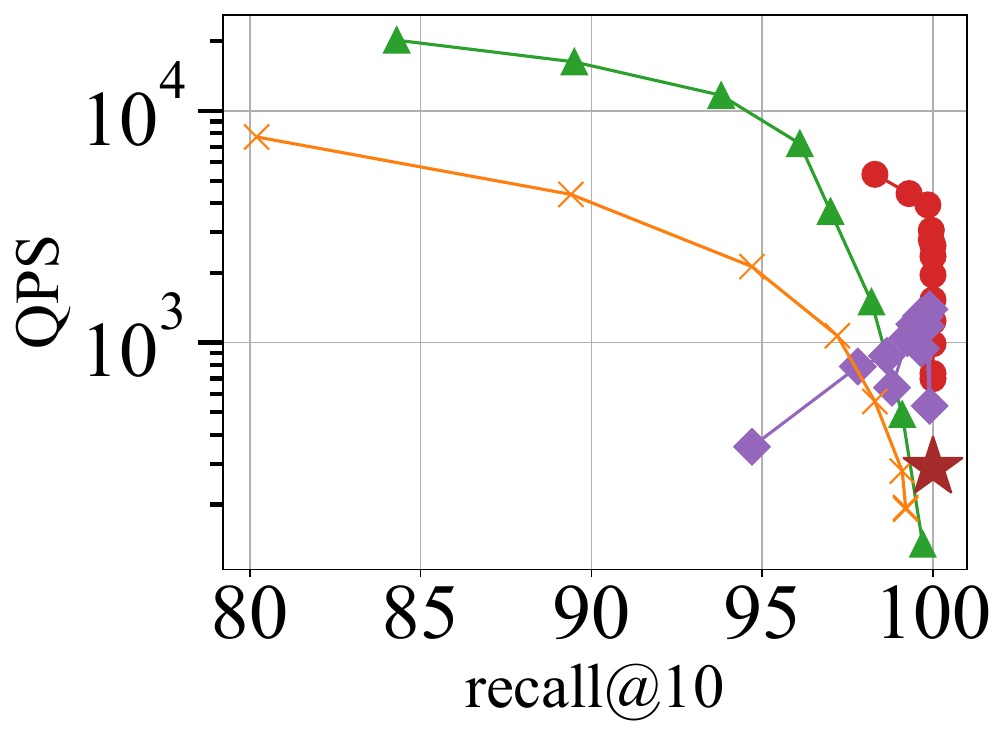}}
		\subfigure[deep1M]{\label{sfig:time_deep1M}\includegraphics[width=0.195\linewidth]{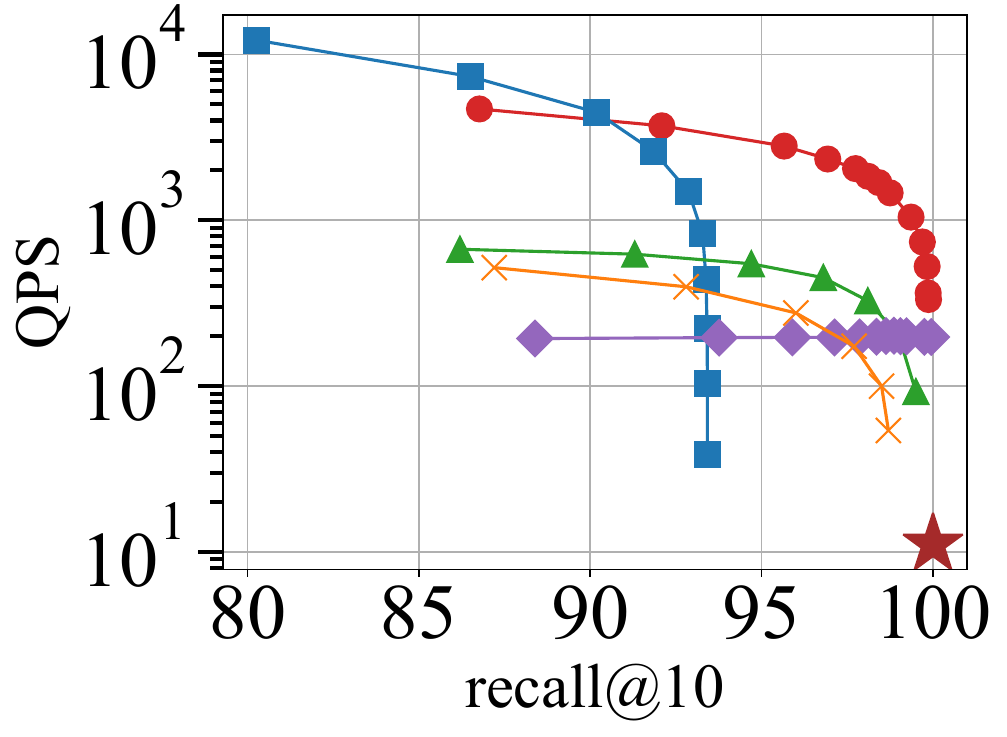}}
		\subfigure[glove2.2m]{\label{sfig:time_glove2.2m}\includegraphics[width=0.195\linewidth]{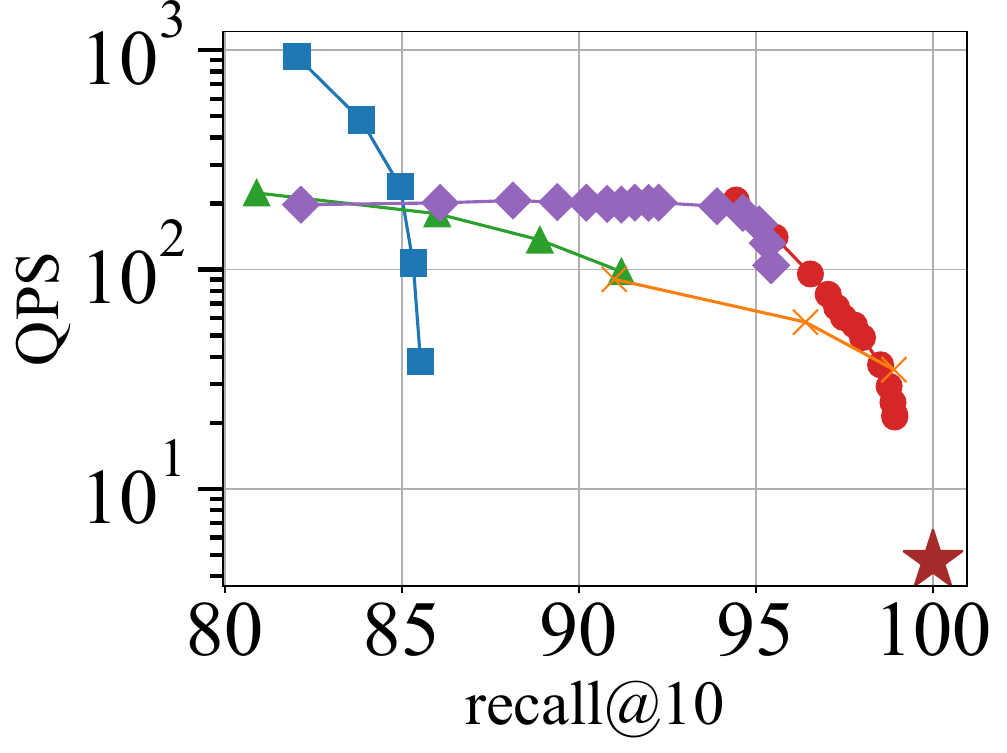}}\vspace*{-0.2cm}\\
		\subfigure[nuswide]{\label{sfig:time_nuswide}\includegraphics[width=0.195\linewidth]{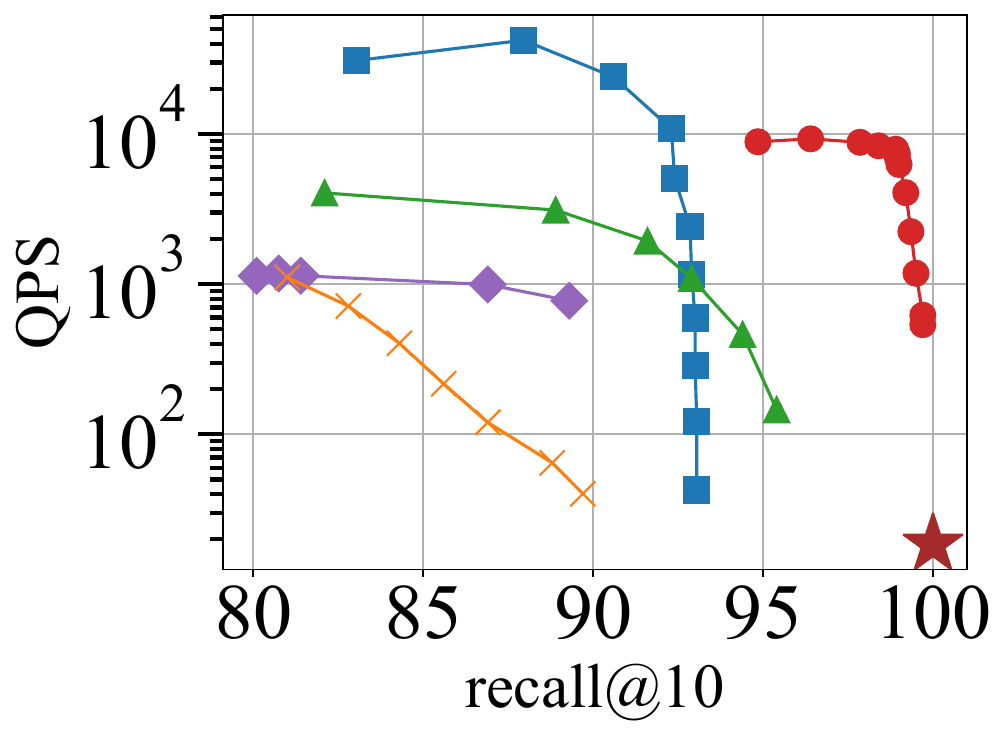}}
		\subfigure[tripclick]{\label{sfig:time_tripclick}\includegraphics[width=0.195\linewidth]{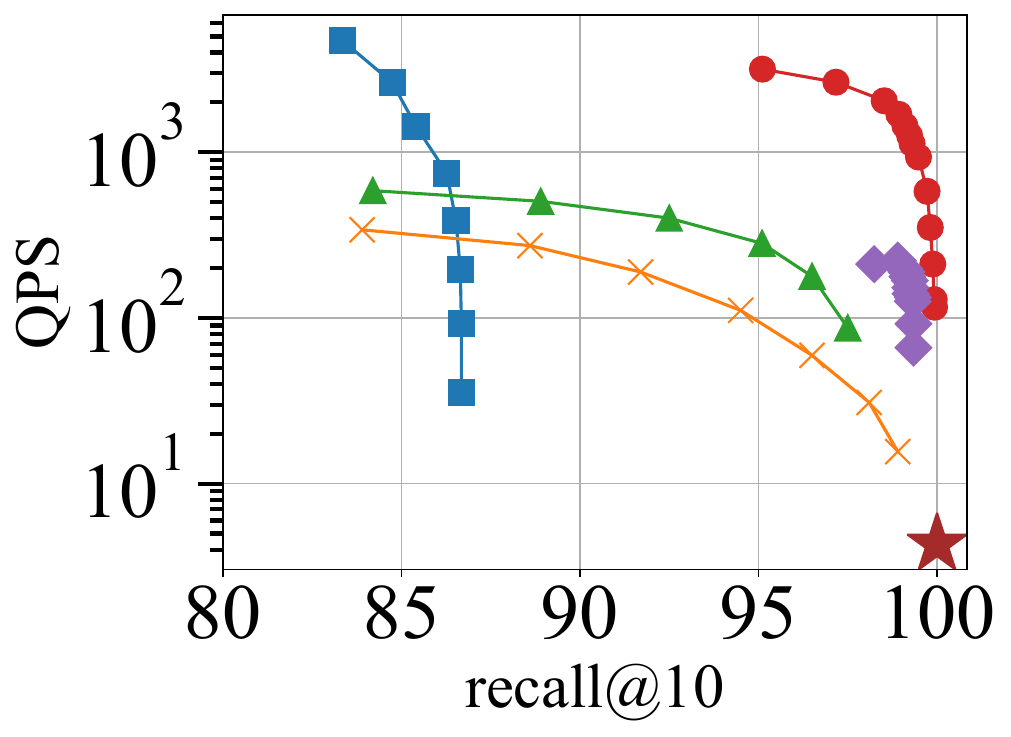}}
		\subfigure[gist1M]{\label{sfig:time_gist1M}\includegraphics[width=0.195\linewidth]{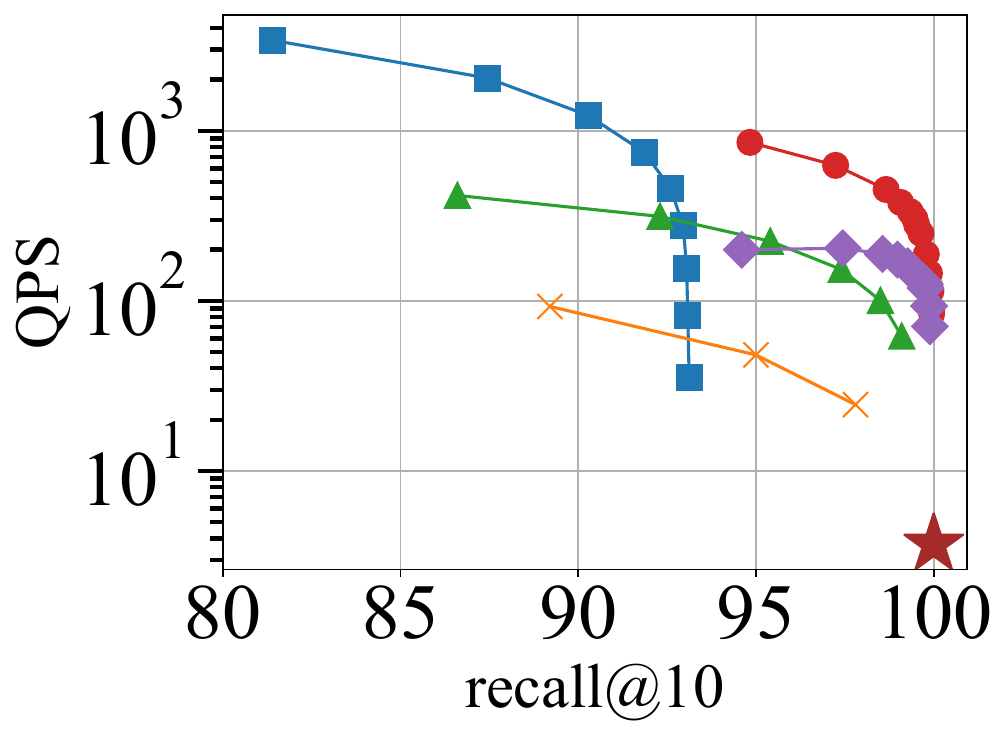}}
		\subfigure[wit]{\label{sfig:time_wit}\includegraphics[width=0.195\linewidth]{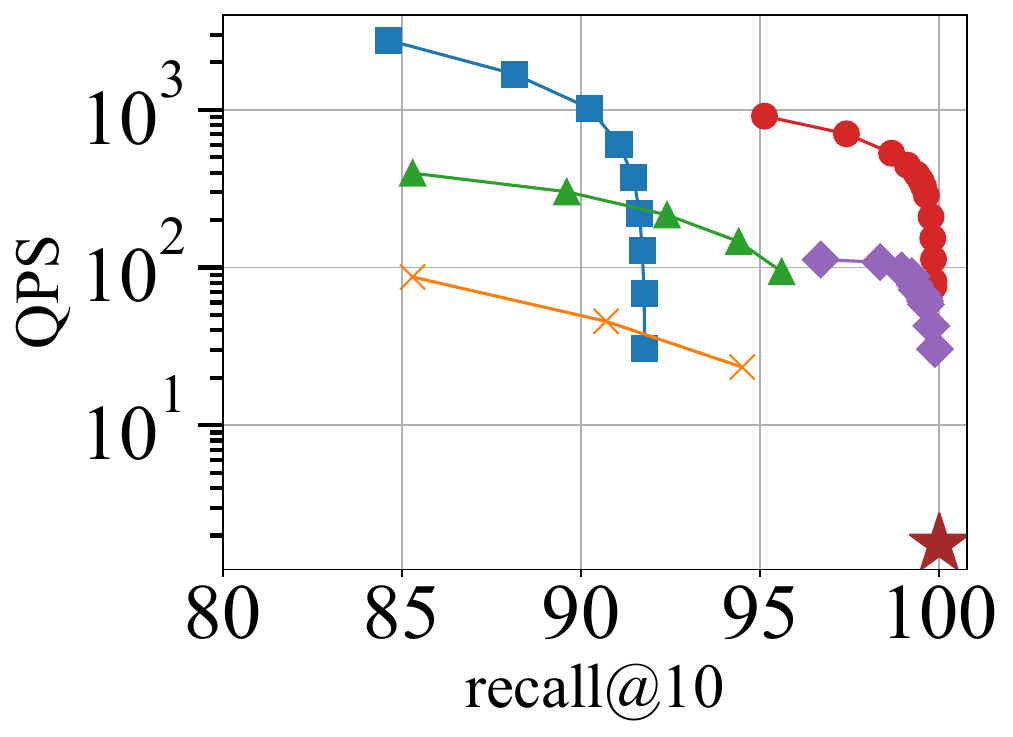}}
		\subfigure[spacev10M]{\label{sfig:time_spacev10M}\includegraphics[width=0.195\linewidth]{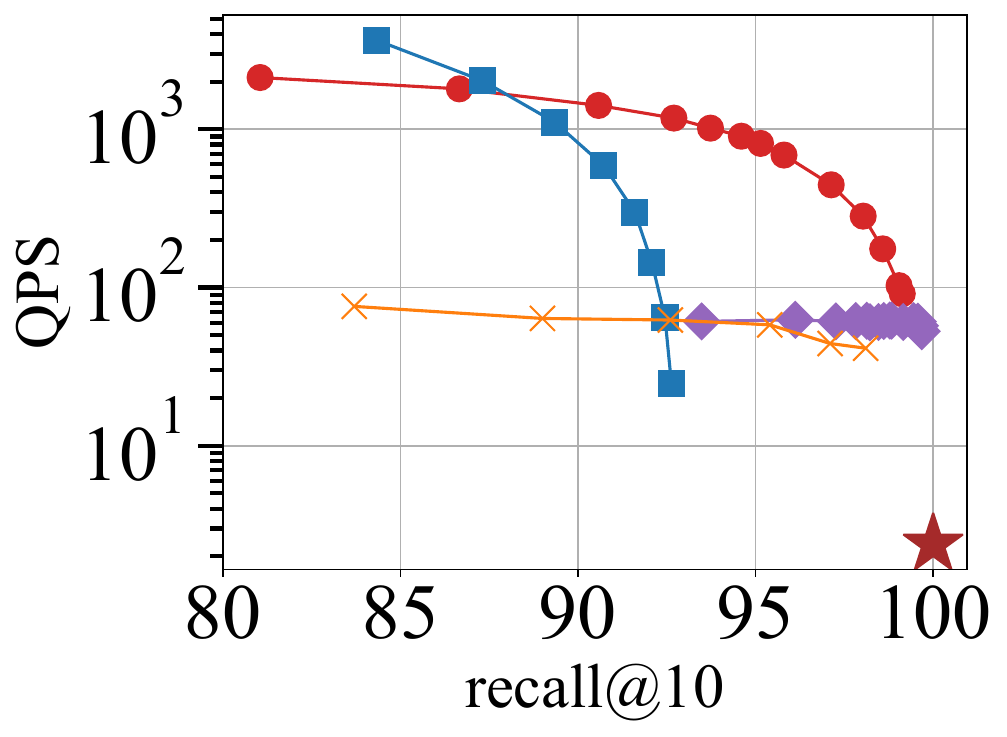}}
	\end{center}
	\vspace*{-0.3cm}
	\caption{Recall@10 vs QPS of different methods under mixed selectivity}
	\label{fig:recall_qps} 
		\vspace*{-0.2cm}
\end{figure*}

\begin{table}[t!]
		\small
	\centering
		\vspace*{-0.2cm}
	\caption{Summary of datasets ($\mathbf{\sum_{C\in\mathcal{M}}|C|}$ is the total number of nodes contained in \mci)}
		\vspace*{-0.3cm}
	\begin{tabular}{c|c|c|c|c|c}
		\toprule
		\textbf{Datasets} & \textbf{$n$} & \textbf{Dim} &  \textbf{\# Q} & \textbf{Type} & $\mathbf{\frac{\sum_{C\in\mathcal{M}}|C|}{n}}$  \\
		\midrule

	sift1M &  1,000,000  & 128 &  10,000  & Image&15.87 \\
	movielens &  10,677  & 150 &  1,000  & Latent Factor&18.57 \\
	audio &  53,387  & 192 &  200  & audio&25.75 \\
	deep1M &  1,000,000  & 256 &  1,000  & Image&17.09 \\
	glove2.2m &  2,196,017  & 300 &  1,000  & Text&45.90 \\
	nuswide &  268,643  & 500 &  200  & Image&32.01 \\
	tripclick &  1,000,000  & 768 &  1,000  & passages&46.29 \\
	gist1M &  1,000,000  & 960 &  1,000  & Scenes&35.66 \\
	wit &  1,000,000  & 2,048 &  1,000  & image&29.60 \\	
    spacev10M &  10,000,000  & 100 &  1,000  & Document&25.43 \\
		\bottomrule
	\end{tabular}
		\vspace*{-0.3cm}
	\label{tab:datasets}
\end{table}

\begin{figure}[t!]
	\begin{center}
	\end{center}
	\vspace*{-0.2cm}
	\includegraphics[width=0.98\linewidth]{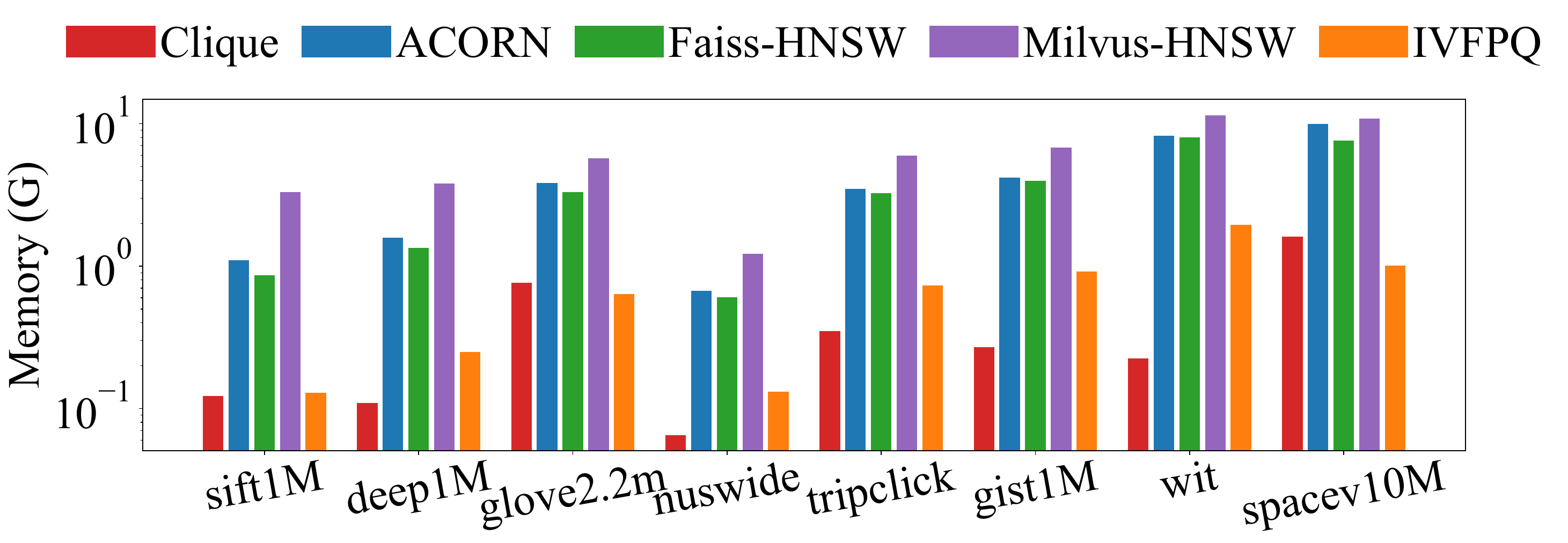}
    \vspace*{-0.3cm}
	\caption{Index size of different methods}
	\label{fig:index_size} 
	\vspace*{-0.3cm}
\end{figure}

\stitle{Datasets.} We use ten publicly available datasets from diverse domains, widely adopted in prior ANNS and AFANNS studies \cite{ACORN,survey21,MochengSurvey}. Detailed characteristics of these datasets are summarized in Table~\ref{tab:datasets}.  
For experiments involving mixed selectivity, we adopt the label (feature) generation method proposed in \cite{UNG}. The selectivity of the generated labels follows a Zipf distribution, where a higher selectivity indicates a larger number of nodes included in the query result set \cite{UNG}. For experiments that require precise control of selectivity, we use the selectivity control method in \cite{MochengSurvey}.

\stitle{Algorithms.} We implement our construction (Algorithm~\ref{alg:construction}) and search (Algorithm~\ref{alg:search}) algorithms in C++. Default parameters are $\tau = 50$, $k' = 200$ (construction) and $\epsilon = 1$ (search). For the large‑scale \texttt{spacev10M} dataset, we increase $k'$ to $300$ to ensure sufficient connectivity.
We compare \mci against four general AFANNS baselines, following the evaluation protocol of \cite{MochengSurvey}:

\begin{itemize}
	\item \textbf{ACORN} \cite{ACORN}: configured with \(M=40\), \(M_{\beta}=80\), \(\gamma = 24\), and \(ef_{\text{construction}} = 500\), which ensures sufficient graph density and convergence for most datasets \cite{ACORN, MochengSurvey}.
	\item \textbf{HNSW}: two implementations (Faiss \cite{FAISS} and Milvus \cite{milvus}), both with $M=50$, $ef_{\text{construction}}=200$.
	\item \textbf{IVFPQ}: Faiss version with 4000 coarse clusters and 8‑bit product quantization.
\end{itemize}

All methods use Euclidean ($L_2$) distance. For specialized filtering tasks, we also include domain‑specific baselines:
\begin{itemize}
    \item \textit{Range filtering}: iRange \cite{iRangeGraph}, SuperPF \cite{DBLP:conf/icml/EngelsLYDS24}, Unify \cite{UNIFY}, and SeRF \cite{SeRF}.
    \item \textit{Keyword filtering}: NHQ (NHQ‑Kgraph variant \cite{MochengSurvey}) and UNG \cite{UNG} (with the number of cross edges $\delta=6$ and entry vectors $\sigma=16$ as recommended \cite{UNG}).
\end{itemize}
All sequential baselines are parallelized for fair latency comparison. General parameter settings follow \cite{MochengSurvey}.

\comment{
To evaluate performance on specialized filtering tasks, we included domain-specific baselines categorized by filter type:
\begin{itemize}
	\item \textit{Range filtering:} iRange \cite{iRangeGraph}, SuperPF \cite{DBLP:conf/icml/EngelsLYDS24}, Unify \cite{UNIFY}, and SeRF \cite{SeRF}.
	\item \textit{Keyword filtering:} NHQ \cite{NHQ} (specifically the high-performance NHQ-Kgraph variant \cite{MochengSurvey}) and UNG \cite{UNG}.
\end{itemize}
All sequential baseline algorithms were parallelized to ensure a fair comparison of query latency. General parameter settings followed the benchmarks in \cite{MochengSurvey}. Specifically for UNG, we set the number of cross edges $\delta=6$ and entry vectors $\sigma=16$, consistent with the original authors' recommendations \cite{UNG}.
}

\stitle{Implementation.}
All experiments run on a Linux server with an AMD Ryzen Threadripper 3990X CPU and 256GB RAM (CentOS 7). Our code is written in C++ and compiled with \texttt{-O3}. Index construction and query processing use 16 threads for all algorithms.

\comment{
\stitle{Implementation.}
We conducted all evaluations on a Linux server equipped with an AMD Ryzen Threadripper 3990X CPU and 256GB of RAM, running CentOS 7. Our proposed algorithm was implemented in C++ and compiled with the \texttt{-O3} optimization flag. All algorithms, including baselines, were executed using 16 threads for both index construction and query processing.
}

\subsection{Performance studies}

\stitle{Exp-1: Recall@10 vs. QPS under mixed selectivity.}
Figure~\ref{fig:recall_qps} shows the Recall@10–QPS trade‑off across datasets under mixed selectivity. \acorn cannot be evaluated on \texttt{movielens} and \texttt{audio} due to segmentation faults in its original implementation.

As can be seen, \mci consistently dominates the high‑recall regime. On the \texttt{wit} dataset (Figure~\ref{sfig:time_wit}), at recall at $95\%$, our \mci outperforms the baselines by an order of magnitude in QPS. On \texttt{deep1M} (Figure~\ref{sfig:time_deep1M}), it sustains $>99.5\%$ recall at $739.7$ QPS, whereas the best baseline (Milvus‑HNSW) reaches only $197.8$ QPS at the same recall—a $3.7\times$ speedup. At lower recall targets, \acorn occasionally achieves higher QPS. For example, on \texttt{spacev10M} (Figure~\ref{sfig:time_spacev10M}) at 85\% recall, \acorn attains $2221$ QPS vs. $1799$ QPS for \mci. This gap stems from our fixed seed‑sampling overhead ($\epsilon\sqrt{n}$ distance computations), which becomes noticeable when the total search budget is small (i.e., a small $l_s$). 

Notably, \acorn fails to reach 99\% recall on several datasets, a result that differs from findings in prior studies \cite{ACORN, MochengSurvey}. This discrepancy arises from our adoption of the query generation method from \cite{UNG}, which includes queries with extremely low selectivity (e.g., $s \approx 0.001$). \acorn struggles to navigate the sparse subgraphs induced by these strict predicates. In contrast, \mci successfully attains over $>98\%$ recall across all tested datasets, demonstrating superior robustness to selectivity variations.

The \texttt{glove2.2m} dataset poses a particular challenge due to its complex topology (many small clusters, hub nodes, and imbalanced distribution) \cite{VISTA}. In such settings, graph‑based indices usually degrade, allowing IVFPQ to lead in recall (Figure~\ref{sfig:time_glove2.2m}). Nevertheless, \mci still outperforms IVFPQ even under these adverse scenarios.

\stitle{Impact of beam width ($l_s$).}
As the search parameter $l_s$ increases, recall improves with diminishing returns. Analyzing \texttt{deep1M} specifically (Figure~\ref{sfig:time_deep1M}): at $l_s = 10$, the system achieves 88\% recall at 4557 QPS. Increasing $l_s$ to 20 boosts recall to 93\% (3570 QPS), and further to 160 yields 99.08\% recall (1391 QPS). However, bridging the final gap to 99.9\% recall requires increasing $l_s$ drastically to 2640, causing QPS to drop to 348. This non-linear trade-off highlights that queries with very low selectivity require extensive graph traversal to ensure near-perfect retrieval.

\begin{figure}[t]
	\vspace*{-0.3cm}
	\begin{center}
	\end{center}
	\includegraphics[width=0.95\linewidth]{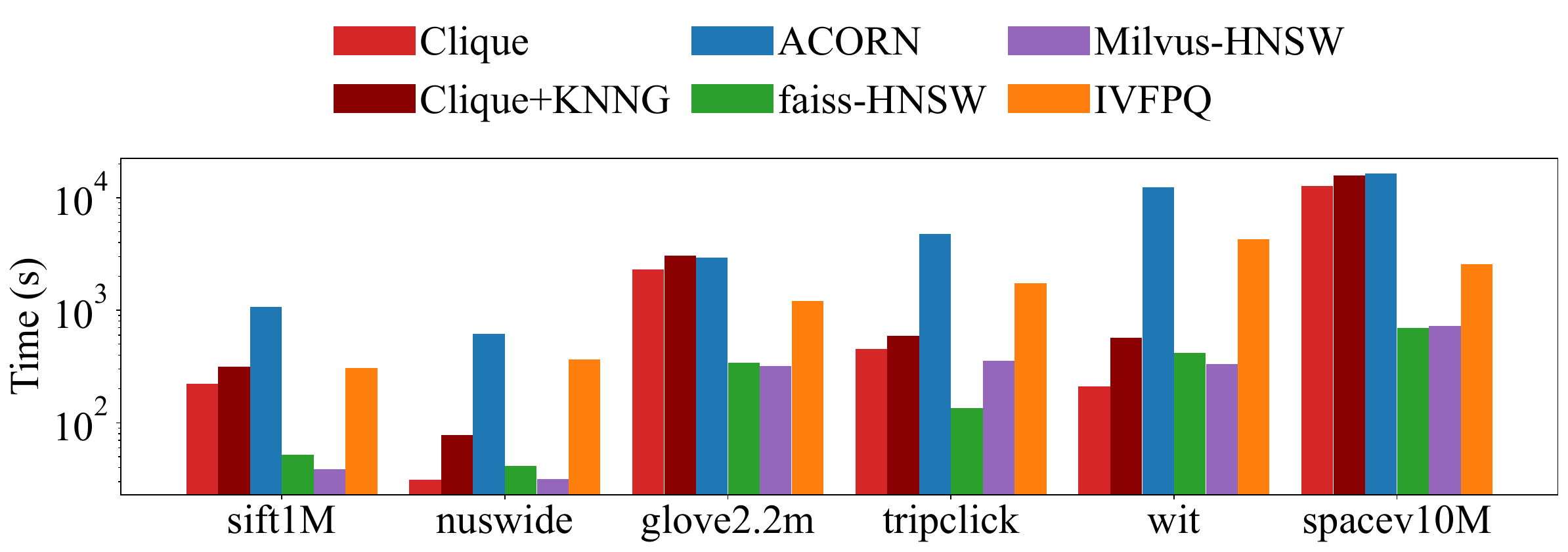}
    \vspace*{-0.3cm}
	\caption{Index construction time of different methods}
	\label{fig:TTI} 
	\vspace*{-0.3cm}
\end{figure}

\begin{figure*}[t!]
	\vspace*{-0.3cm}
	\begin{center}
		\includegraphics[width=0.6\linewidth]{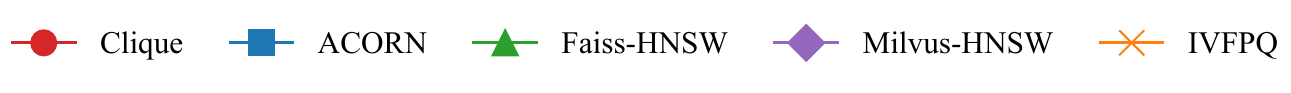}\vspace*{-0.4cm}\\
		\subfigure[glove2.2m, s=0.5]{\label{sfig:time_glove2.2m0.5}\includegraphics[width=0.195\linewidth]{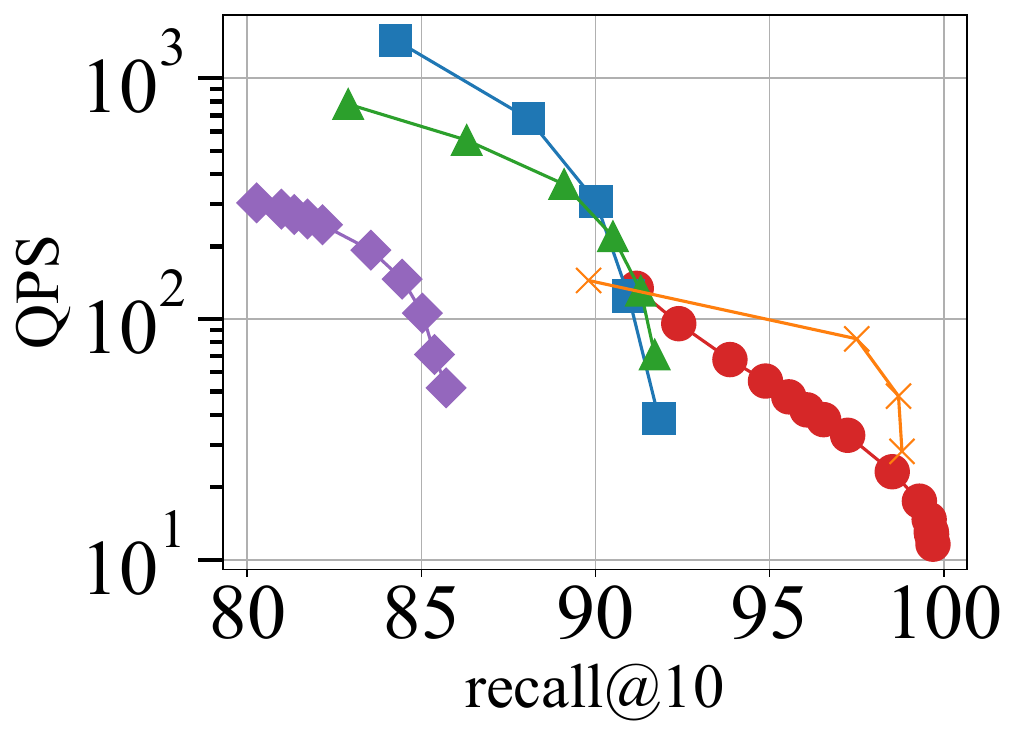}} 
		\subfigure[glove2.2m, s=0.1]{\label{sfig:time_glove2.2m0.1}\includegraphics[width=0.195\linewidth]{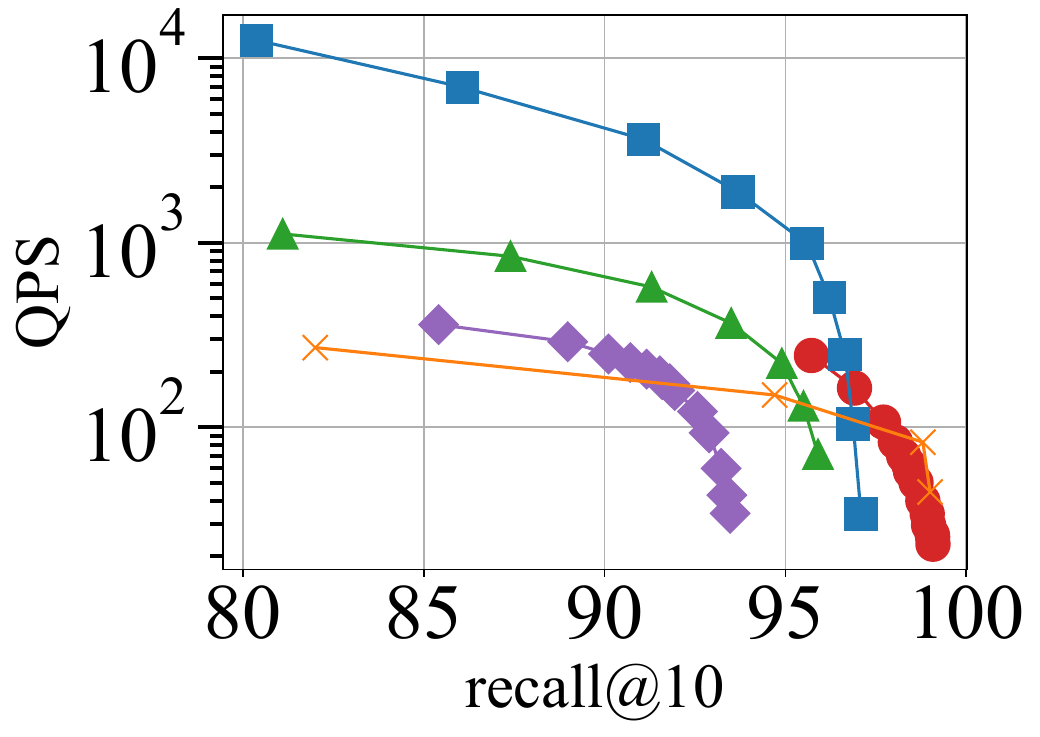}} 
		\subfigure[glove2.2m, s=0.05]{\label{sfig:time_glove2.2m0.05}\includegraphics[width=0.195\linewidth]{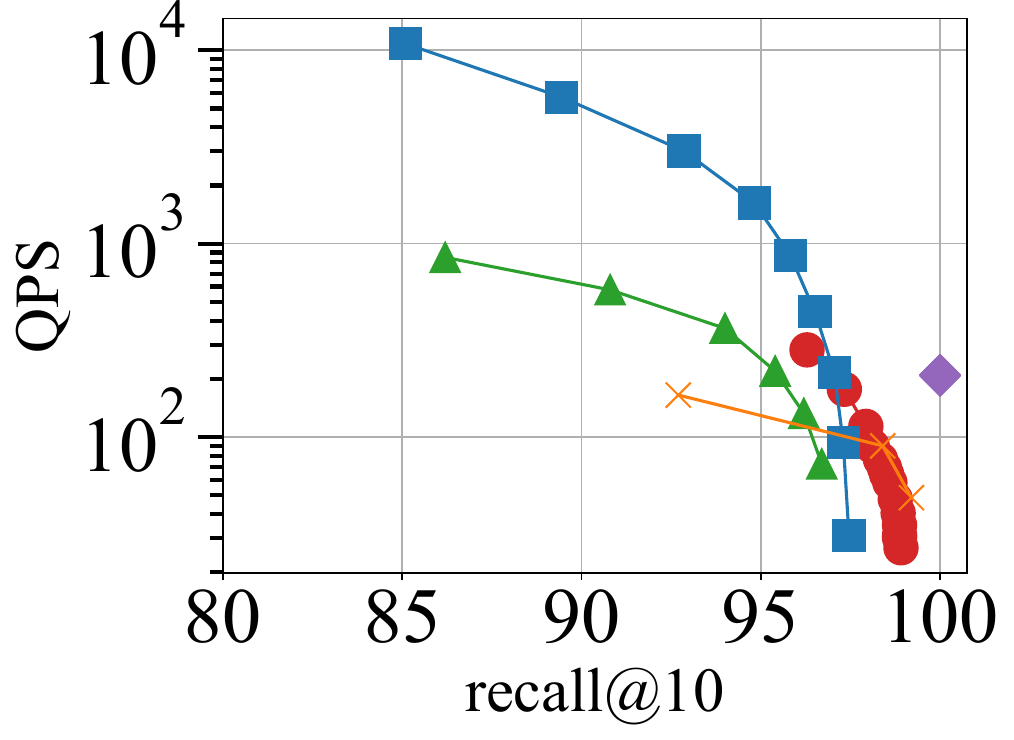}}
		\subfigure[glove2.2m, s=0.01]{\label{sfig:time_glove2.2m0.01}\includegraphics[width=0.195\linewidth]{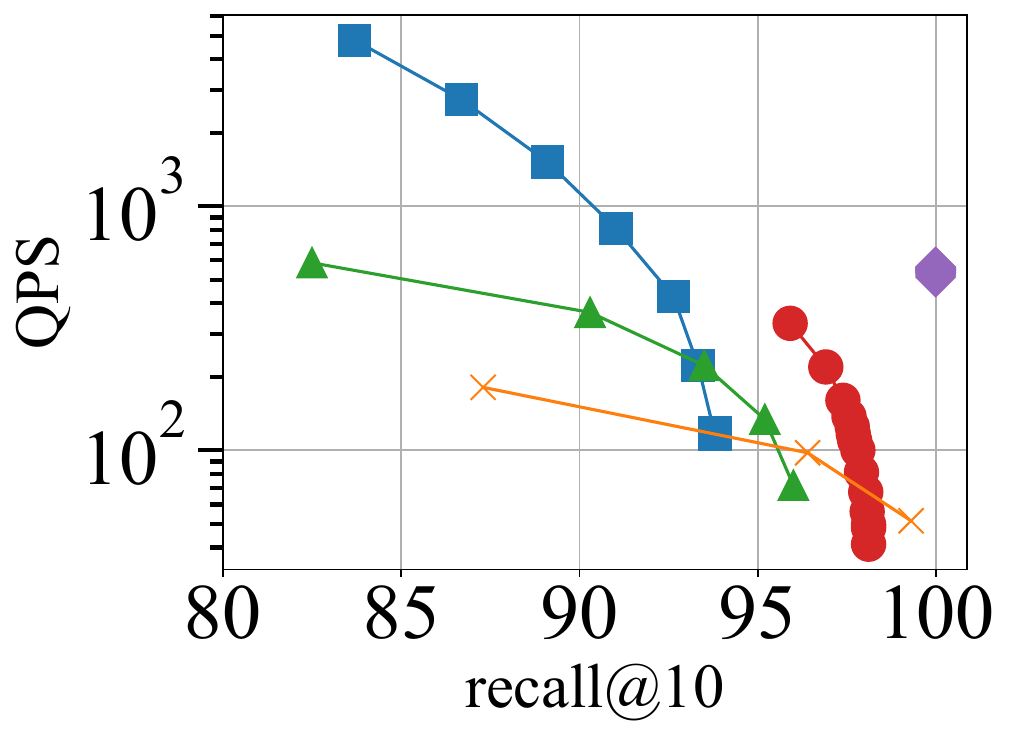}}
		\subfigure[glove2.2m, s=0.005]{\label{sfig:time_glove2.2m0.005}\includegraphics[width=0.195\linewidth]{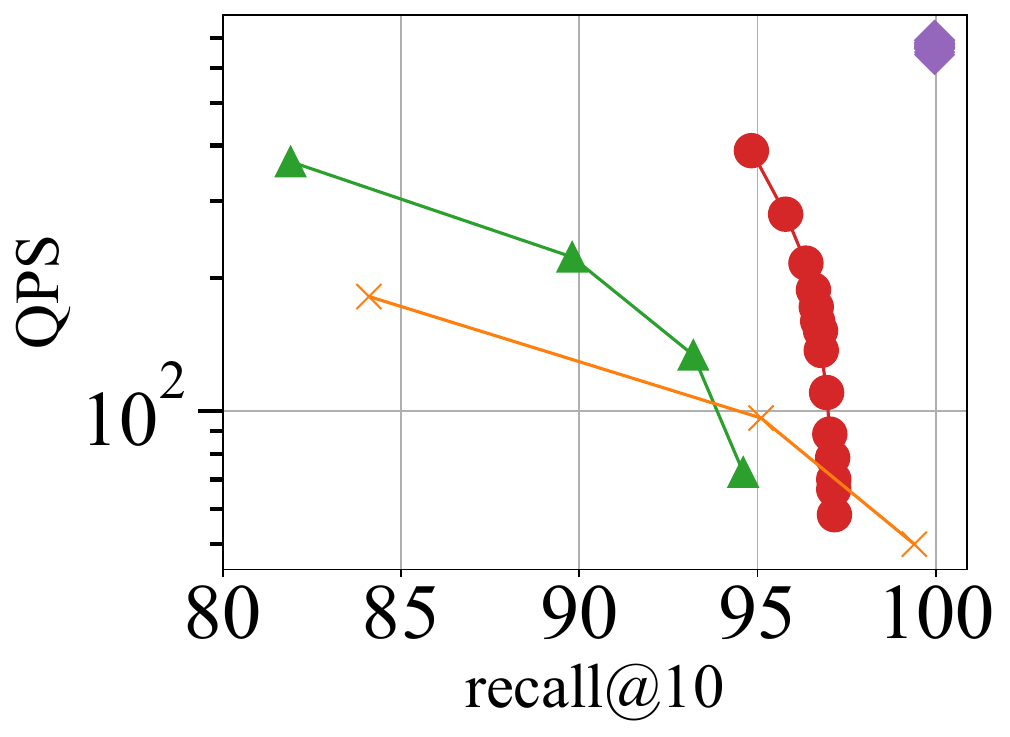}}
		\vspace*{-0.3cm}\\
        \subfigure[tripclick, s=0.5]{\label{sfig:time_tripclick0.5}\includegraphics[width=0.195\linewidth]{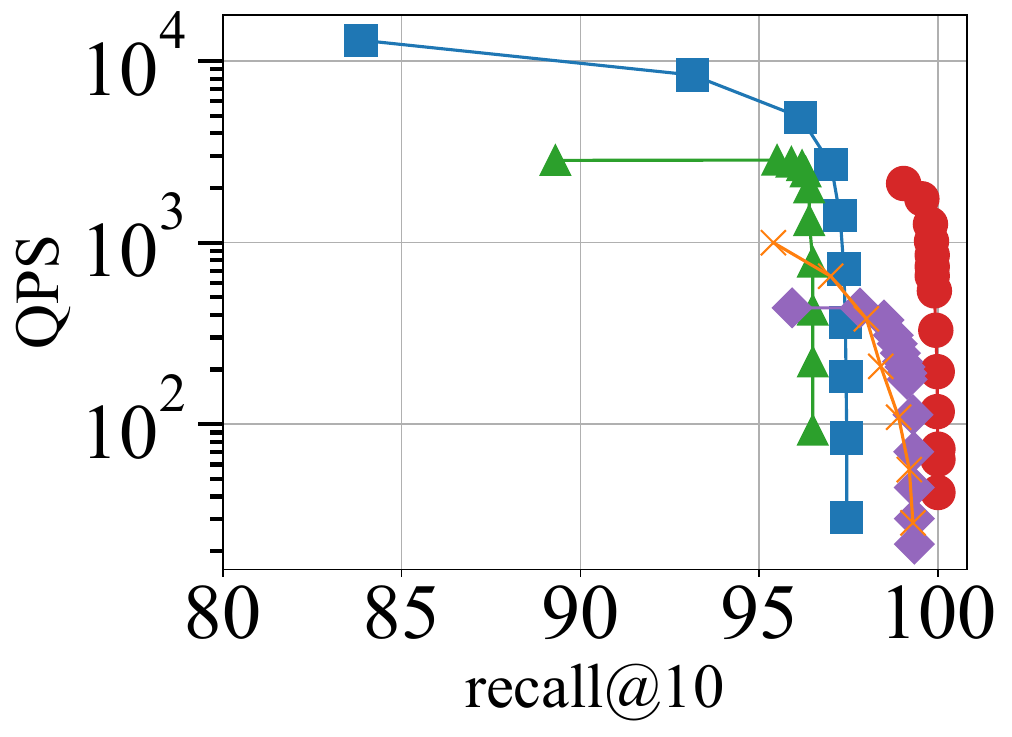}} 
		\subfigure[tripclick, s=0.1]{\label{sfig:time_tripclick0.1}\includegraphics[width=0.195\linewidth]{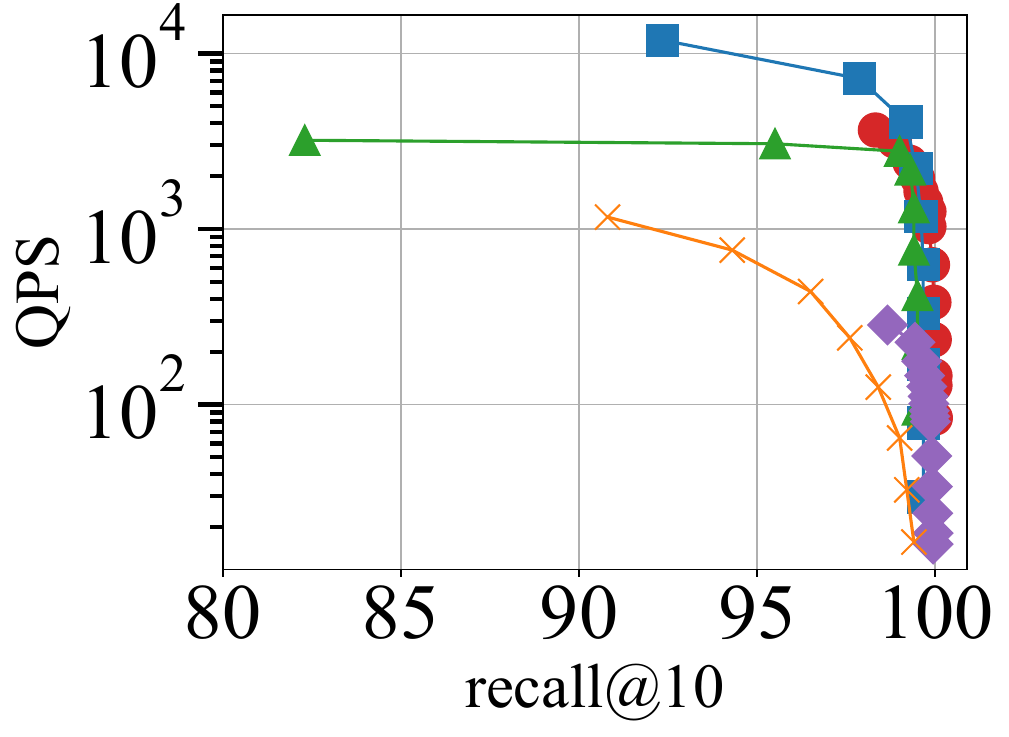}} 
		\subfigure[tripclick, s=0.05]{\label{sfig:time_tripclick0.05}\includegraphics[width=0.195\linewidth]{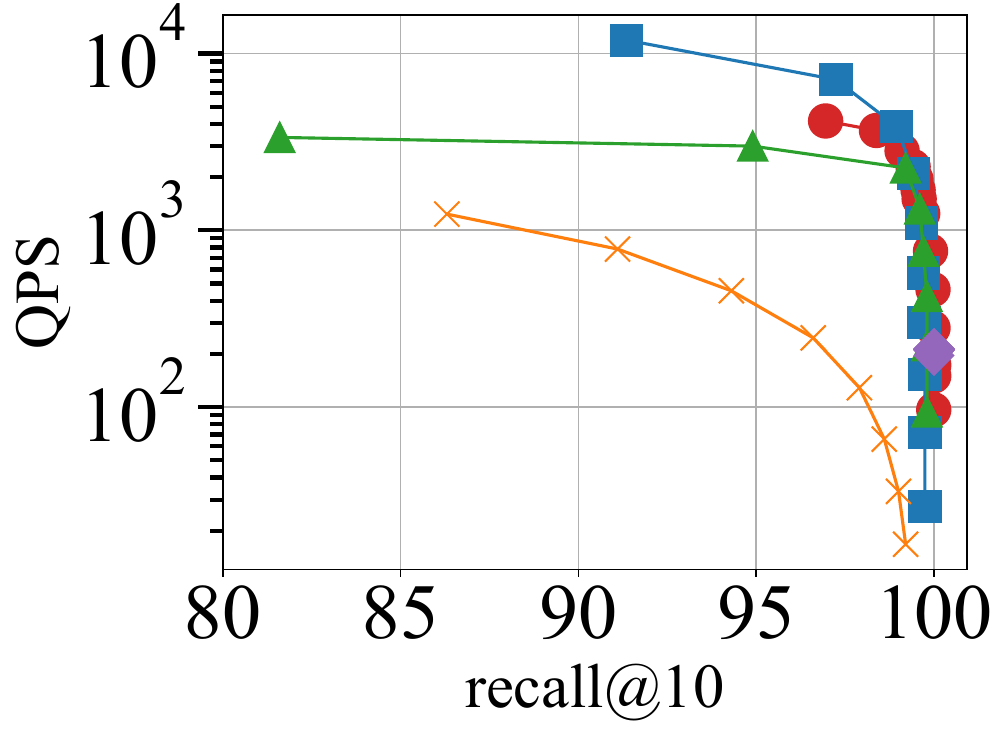}}
		\subfigure[tripclick, s=0.01]{\label{sfig:time_tripclick0.01}\includegraphics[width=0.195\linewidth]{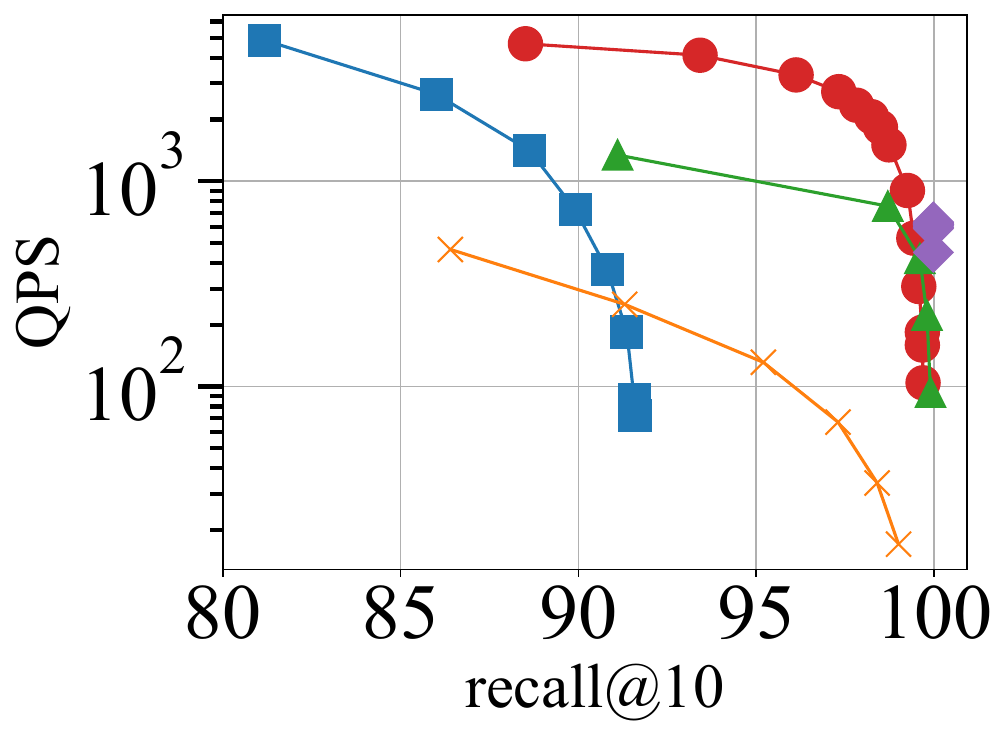}}
		\subfigure[tripclick, s=0.005]{\label{sfig:time_tripclick0.005}\includegraphics[width=0.195\linewidth]{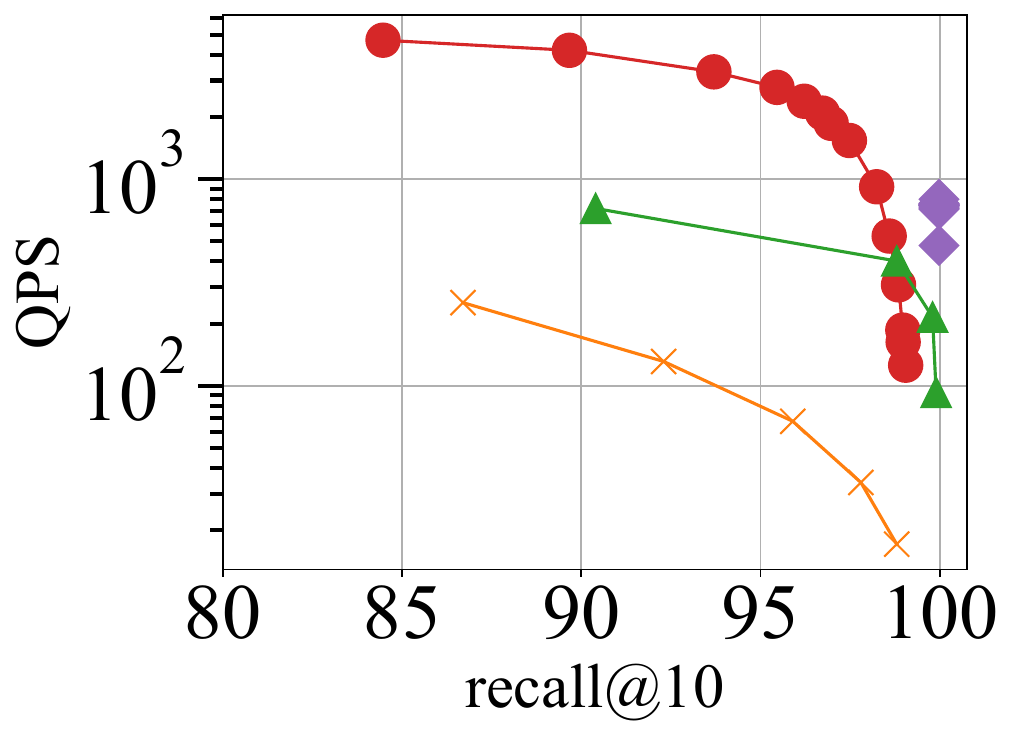}}
		\vspace*{-0.3cm}	\\
		\subfigure[gist1M, s=0.5]{\label{sfig:time_gist1M0.5}\includegraphics[width=0.195\linewidth]{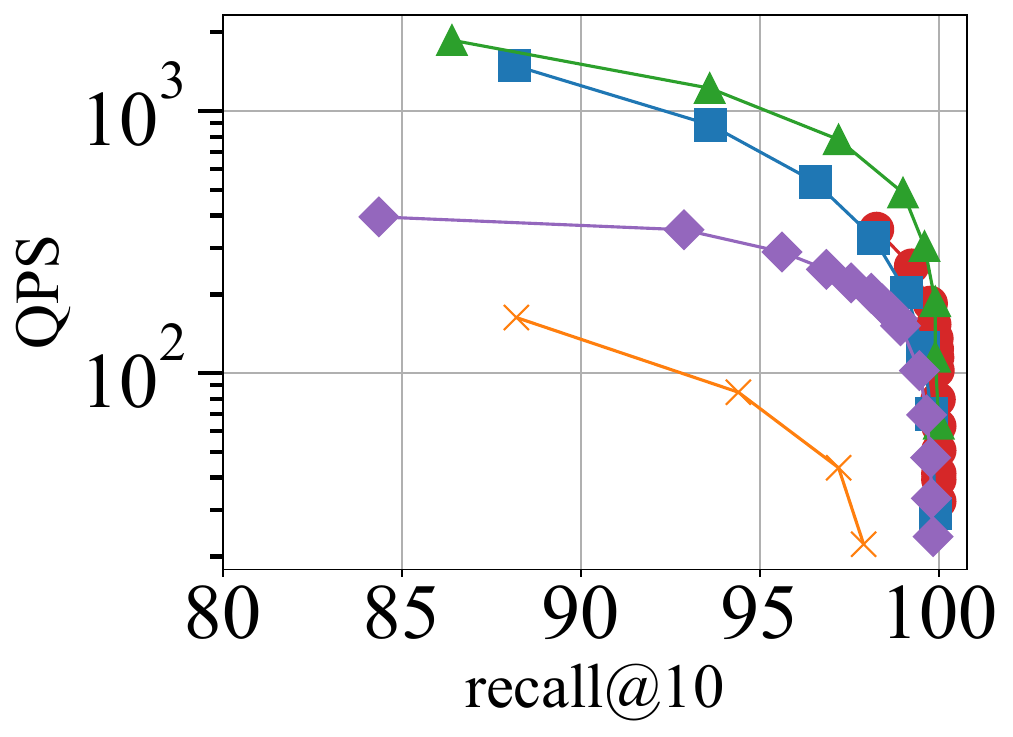}}
		\subfigure[gist1M, s=0.1]{\label{sfig:time_gist1M0.1}\includegraphics[width=0.195\linewidth]{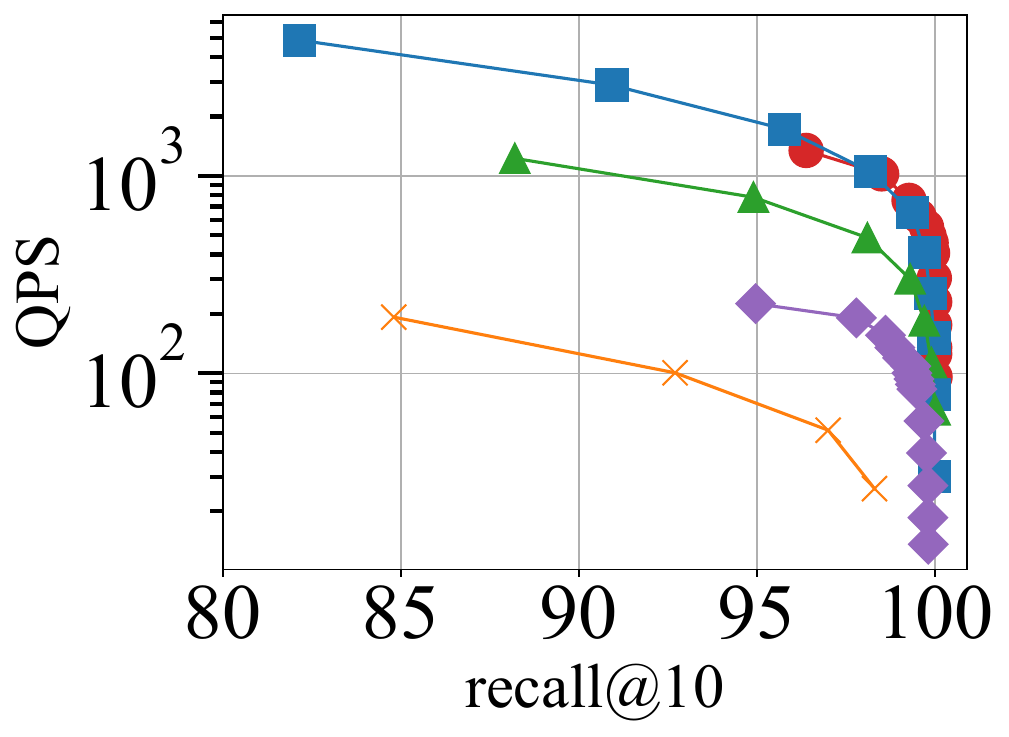}}
		\subfigure[gist1M, s=0.05]{\label{sfig:time_gist1M0.05}\includegraphics[width=0.195\linewidth]{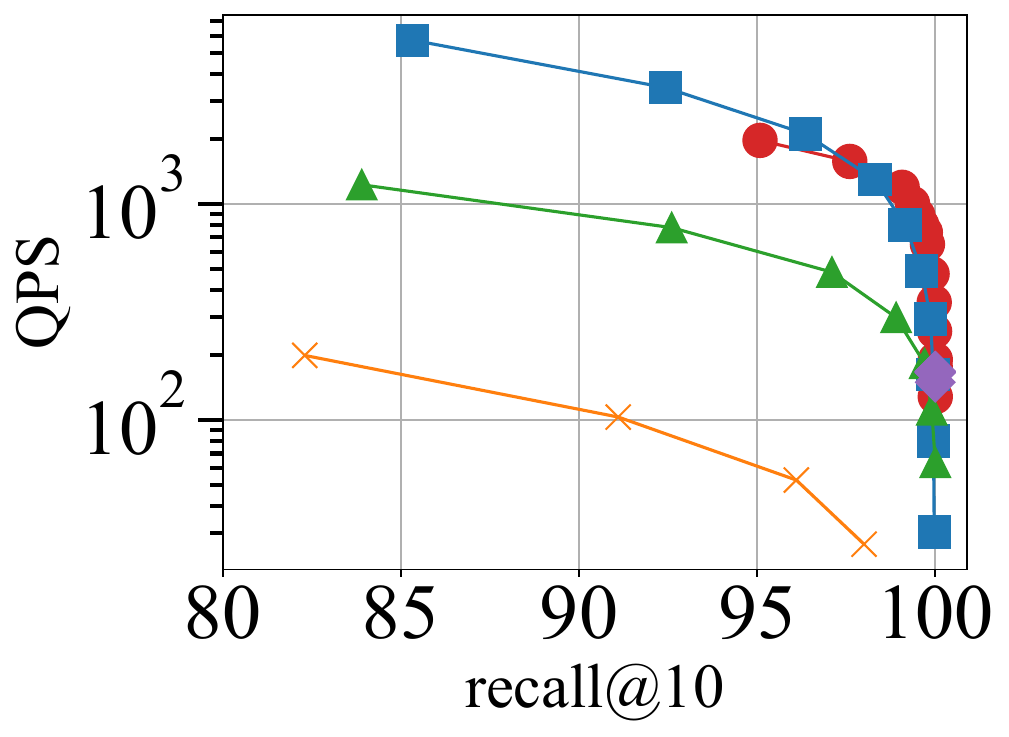}}       
		\subfigure[gist1M, s=0.01]{\label{sfig:time_gist1M0.01}\includegraphics[width=0.195\linewidth]{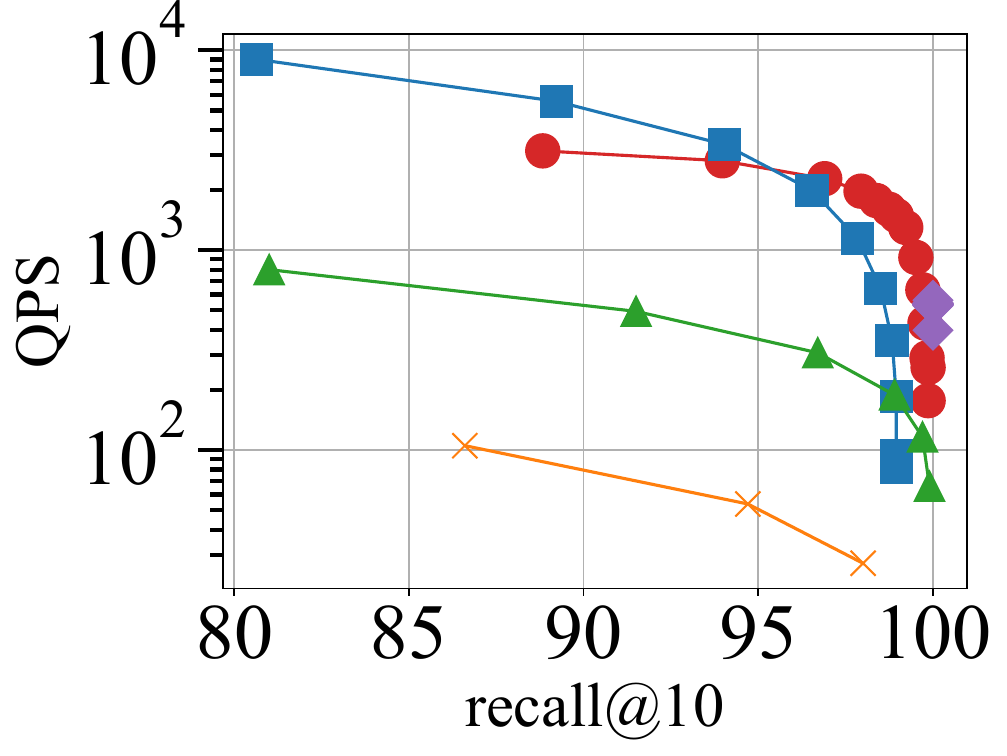}}       
		\subfigure[gist1M, s=0.005]{\label{sfig:time_gist1M0.005}\includegraphics[width=0.195\linewidth]{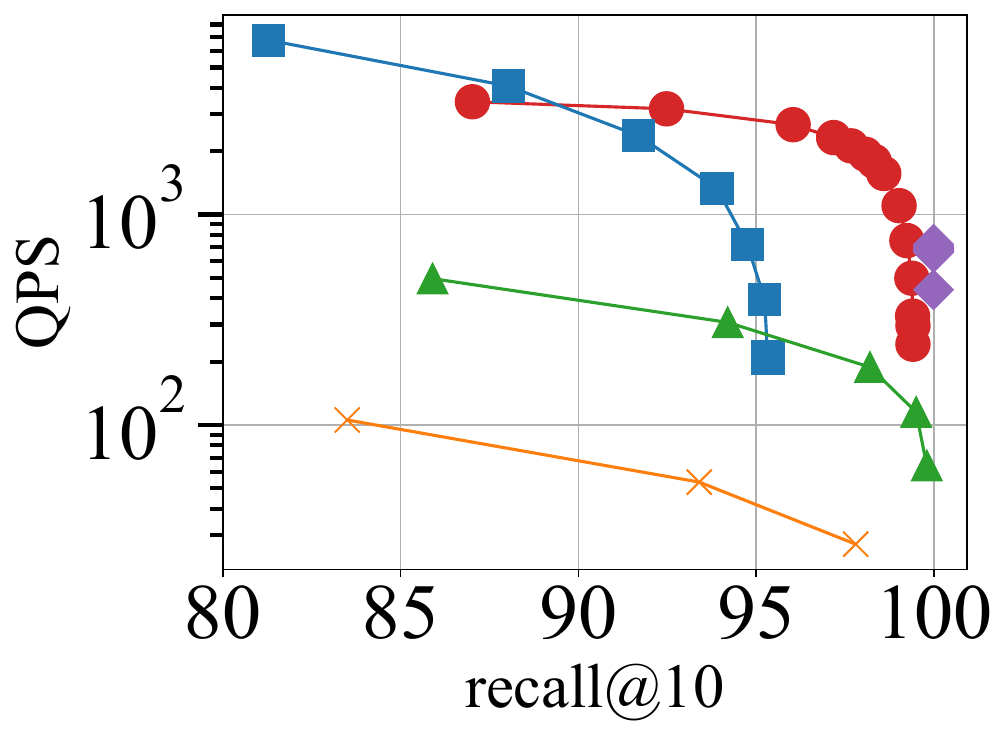}}   	\\

	\end{center}
	\vspace*{-0.3cm}
	\caption{Recall@10 vs QPS of different methods under various selectivity}
	\label{fig:recall_qps_varselectivity} 
	\vspace*{-0.2cm}
\end{figure*}

\stitle{Exp-2: Comparison of index size.}
Figure~\ref{fig:index_size} compares the index storage requirements of our proposed \mci against the baselines. Note that the index size for the Milvus implementation is not directly measurable \cite{MochengSurvey} and is thus approximated by the official sizing tool \url{https://milvus.io/zh-hant/tools/sizing}.

Our \mci demonstrates significant storage efficiency, requiring at most 23\% of the space consumed by \acorn or HNSW across all datasets. When compared to the quantization-based IVFPQ, \mci achieves the smallest index size on 6 out of 8 datasets. On the remaining two, its size is within a factor of $1.6\times$ that of IVFPQ. This reduction stems from the fundamental difference in structure: while proximity graph-based indices (like HNSW and \acorn) must store a neighbor list for every node, \mci stores only the $\tau$-Maximal Clique Cover (\tmcc), which inherently compresses the topological information. Table~\ref{tab:datasets} presents the ratio of the number of integers in \mci and $n$, which not surpass 46.29 for all datasets. These results highlight the storage superiority of \mci, reinforcing its practical viability for large-scale AFANNS deployments.

\stitle{Exp‑3: Comparison of index construction time (ICT).}
Figure~\ref{fig:TTI} compares index construction times of different methods. For \mci, we report two indexing times: “Clique” (building \mci from a pre‑built $k$-NNG) and “Clique+KNNG” (end‑to‑end time, including $k$-NNG construction via NN‑Descent \cite{NN-Decent}). The end‑to‑end construction of \mci is faster than \acorn, requiring only 43\% of \acorn’s time on average. Although HNSW builds more quickly, \mci’s one‑time indexing overhead is justified by its substantially better query performance—especially in challenging low‑selectivity and high‑recall regimes—delivering a favorable trade‑off between construction cost and query efficiency.

\stitle{Exp-4: Results under varying query selectivity.}
To evaluate robustness across selectivity, we tested three representative datasets (Figure~\ref{fig:recall_qps_varselectivity}); trends on other datasets are similar.

Compared to \acorn, \mci is particularly stronger in high‑recall, low‑selectivity regimes. \acorn remains competitive when $s \ge 0.05$ and target recall $< 90\%$. However, for Recall@10 $\ge 95\%$ and $s \le 0.01$, \mci consistently outperforms \acorn on every dataset and achieves a higher overall recall ceiling. Against Milvus‑HNSW, \mci delivers higher throughput at comparable recall, especially at higher selectivities. On \texttt{glove2.2m} with $s=0.5$ (Figure~\ref{sfig:time_glove2.2m0.5}), \mci reaches $>95\%$ recall, a level Milvus‑HNSW cannot attain. At low selectivity ($s \le 0.01$), Milvus‑HNSW can achieve extremely high recall ($\ge 99.9\%$) but at a high latency cost. For example, on \texttt{tripclick} with $s=0.01$, Milvus‑HNSW reaches 99.98\% recall at 614 QPS, while \mci offers more practical trade‑offs: 96\% recall at 3147 QPS ($5\times$ faster) and 98.8\% recall at 1499 QPS ($2.4\times$ faster). 

In summary, \mci adapts robustly across a wide selectivity range, offering an efficient balance between high recall and high throughput, especially in challenging low‑selectivity, high‑recall scenarios.


\begin{figure}[t!]
	\vspace*{-0.3cm}
	\begin{center}
	\includegraphics[width=0.99\linewidth]{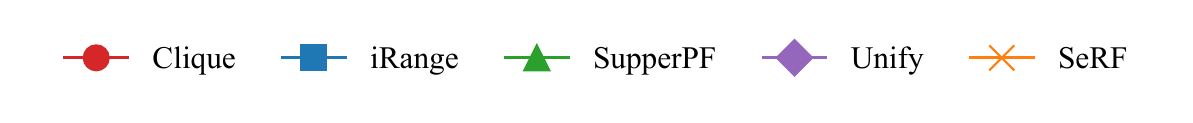}\vspace*{-0.5cm}\\
\subfigure[sift1M]{\label{sfig:time_sift1M_range}\includegraphics[width=0.32\linewidth]{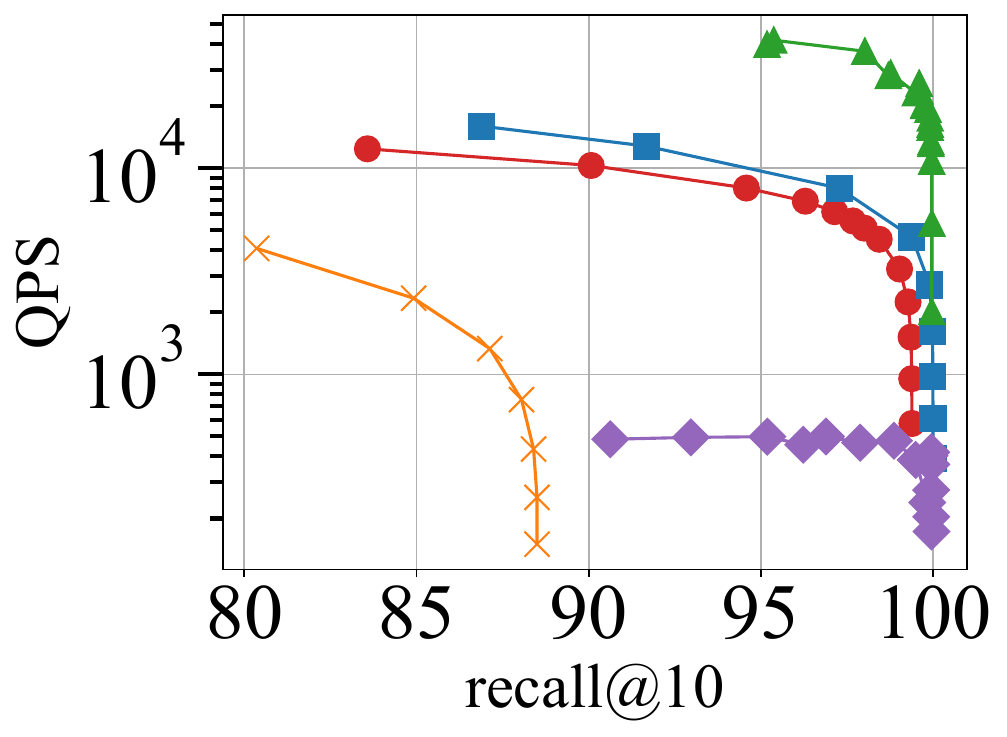}}
\subfigure[glove2.2m]{\label{sfig:time_glove2.2m_range}\includegraphics[width=0.32\linewidth]{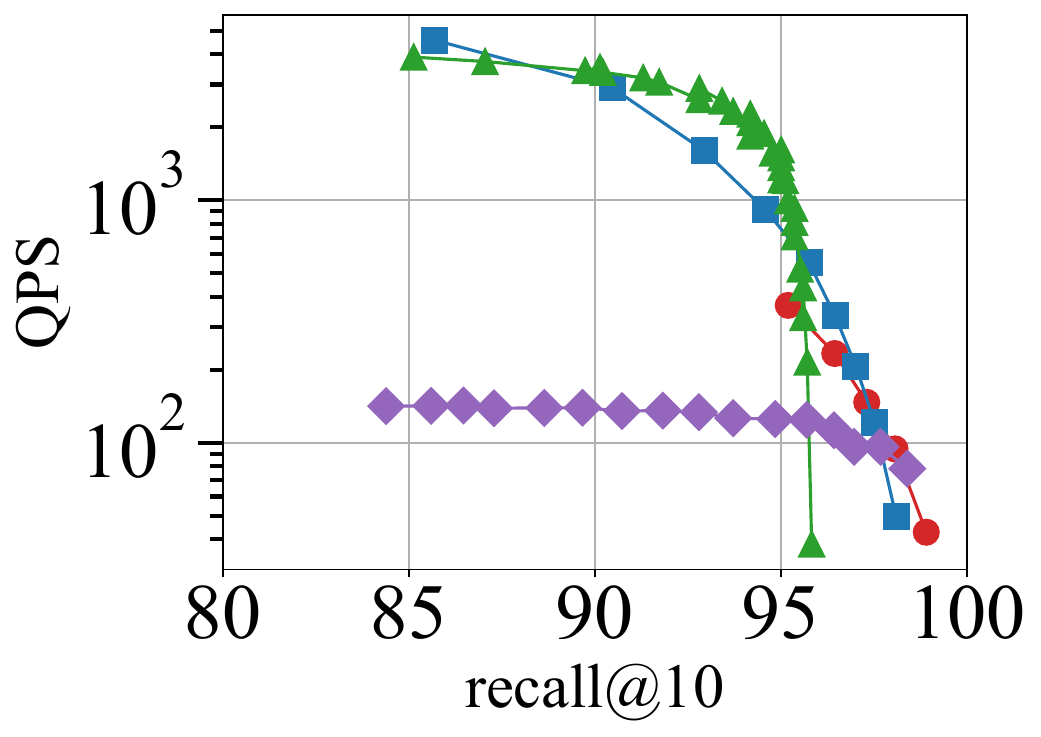}}
\subfigure[nuswide]{\label{sfig:time_nuswide_range}\includegraphics[width=0.32\linewidth]{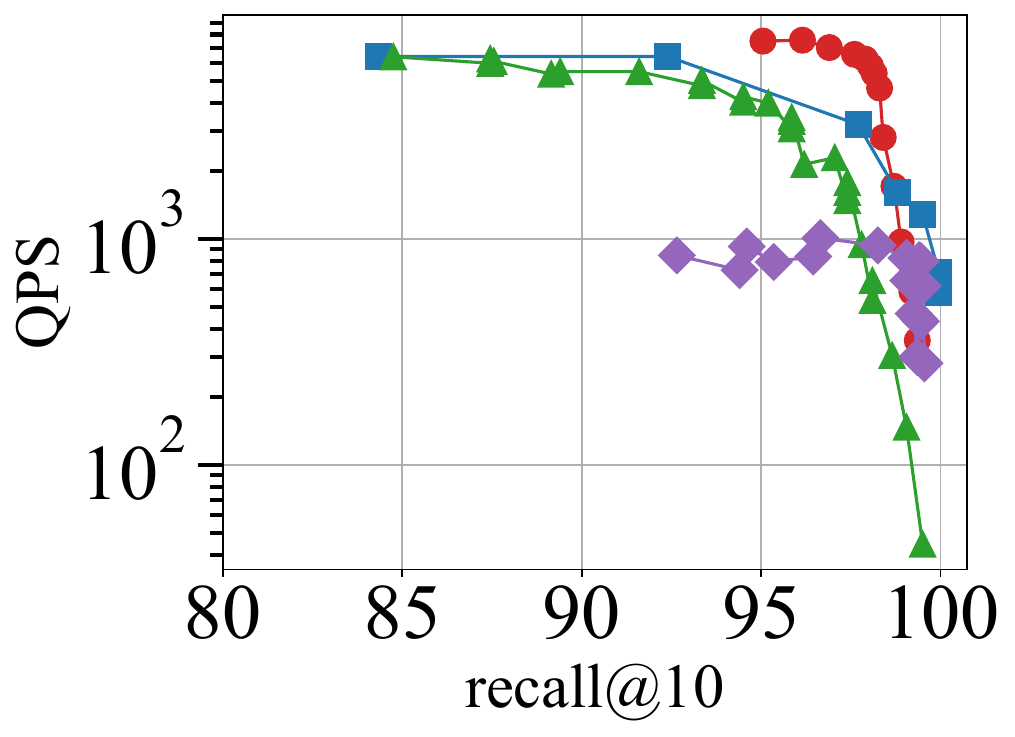}}
	\end{center}
\vspace*{-0.2cm}
\caption{Results of different methods on range filtering}
\label{fig:recall_qps_range} 
\vspace*{-0.4cm}
\end{figure}

\stitle{Exp‑5: Results on range filtering.}
Figure~\ref{fig:recall_qps_range} shows the range‑filtering performance. Each data vector has a numeric feature, and each query specifies a range; query selectivities are controlled at values \( \{0.3, 0.15, 0.07, 0.03, 0.015, 0.007, 0.003, 0.001\} \) using the method in \cite{MochengSurvey}. We report results on three representative datasets.
On \texttt{sift1M} (Figure~\ref{sfig:time_sift1M_range}), \mci delivers competitive performance, closely matching the state‑of‑the‑art specialized baselines (iRange and SuperPF) in the high‑recall regime—a trend observed across most datasets. On more challenging datasets like \texttt{glove2.2m} (Figure~\ref{sfig:time_glove2.2m_range}) and \texttt{nuswide} (Figure~\ref{sfig:time_nuswide_range}), \mci not only remains competitive but also shows a slight advantage at high recall. Despite the fact that the baseline methods build indexes specifically optimized for numerical range filtering, our general‑purpose \mci achieves comparable or better efficiency, especially when high recall is required. This underscores the robustness and wide applicability of our approach.

\comment{
\stitle{Exp-4: Results of range filtering.}
The performance of range filtering is in Figure~\ref{fig:recall_qps_range}. Each data vector has a number as feature and each query has a range. Query selectivities are controlled across a set of values: $ \{0.3, 0.15, 0.07, 0.03, 0.015, 0.007, 0.003, 0.001\} $ by the methods in \cite{MochengSurvey}. We report results on three representative datasets.
On the \texttt{sift1M} dataset (Figure~\ref{sfig:time_sift1M_range}), \mci does not outperform the baseline methods iRange and SupperPF. A similar trend is observed across most other datasets. However, on \texttt{glove2.2m} (Figure~\ref{sfig:time_glove2.2m_range}) and \texttt{nuswide} (Figure~\ref{sfig:time_nuswide_range}), \mci achieves comparable performance and demonstrates an advantage at high recall.

Even though the three baseline methods build specialized indexes tailored to the numerical features, our method still maintains competitive or superior efficiency under high-recall conditions, showcasing its robustness and general applicability.
}

\begin{figure}[t!]
	\vspace*{-0.25cm}
	\begin{center}
\includegraphics[width=0.99\linewidth]{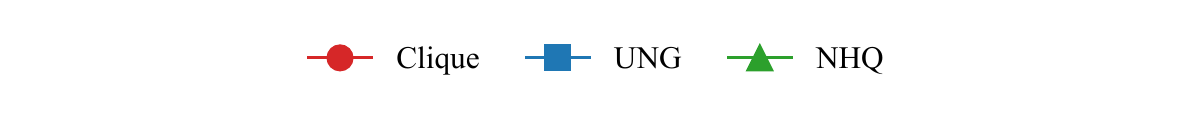}\vspace*{-0.5cm}\\
\subfigure[deep1M]{\label{sfig:time_deep1M_leq}\includegraphics[width=0.32\linewidth]{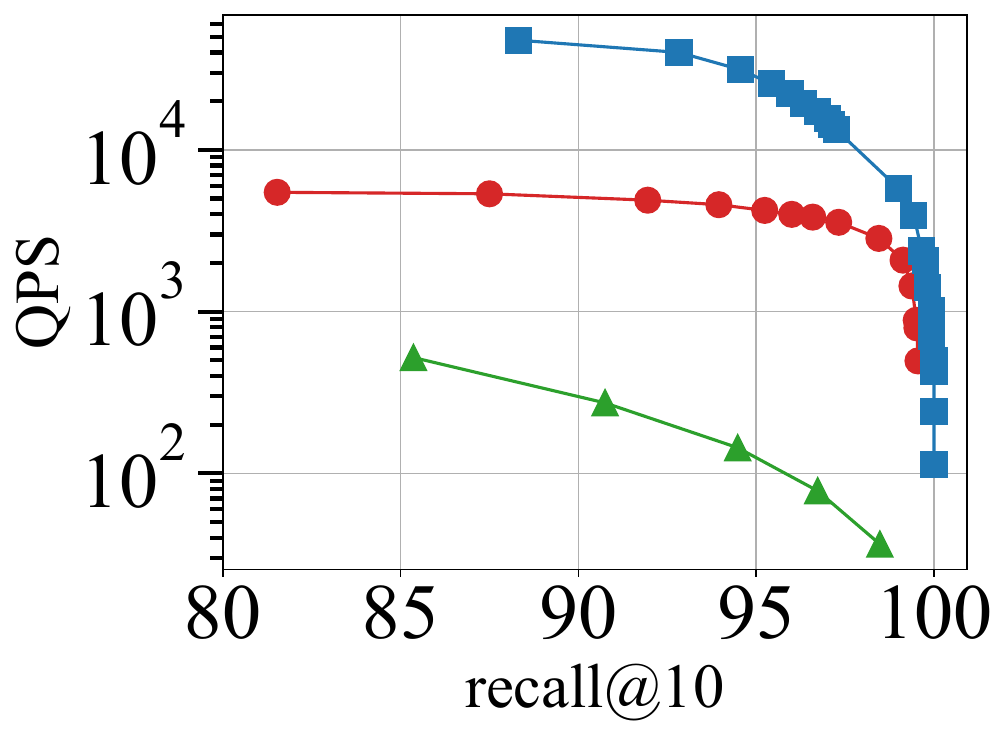}}
\subfigure[glove2.2m]{\label{sfig:time_glove2.2m_leq}\includegraphics[width=0.32\linewidth]{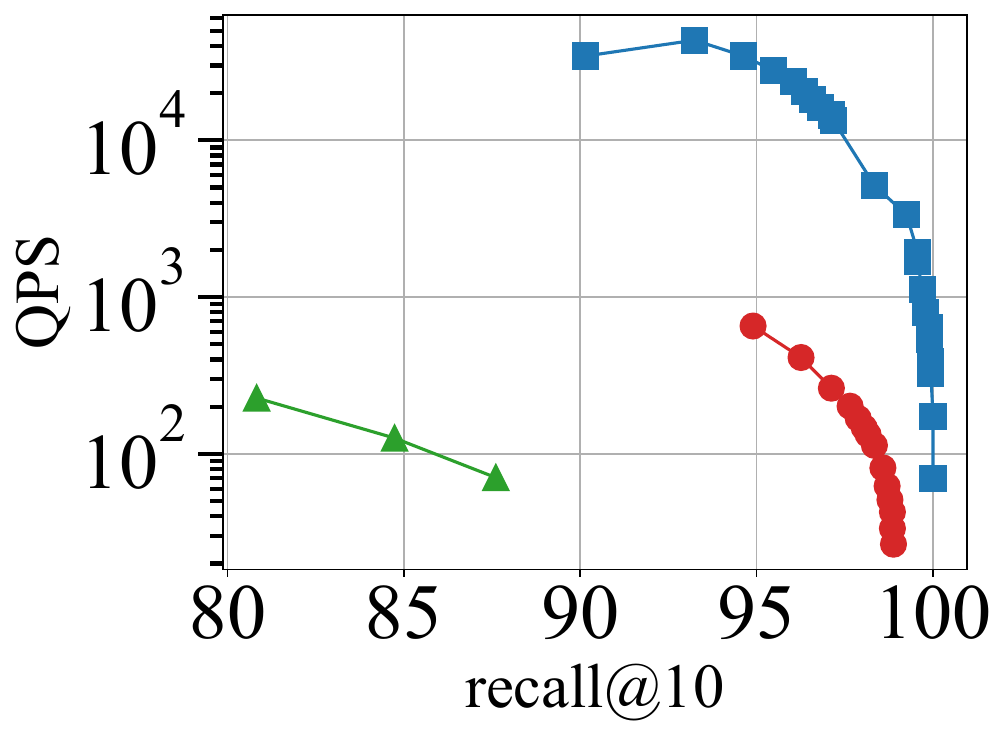}}
\subfigure[wit]{\label{sfig:time_wit_leq}\includegraphics[width=0.32\linewidth]{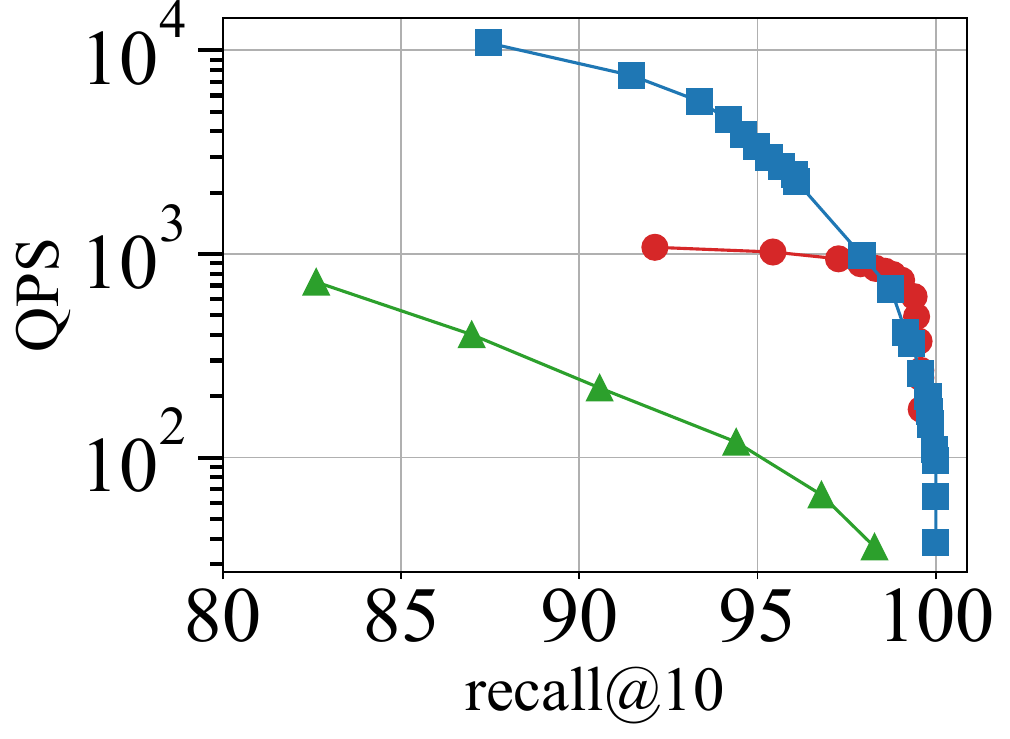}}
	\end{center}
	\vspace*{-0.3cm}
\caption{Results of different methods for keywords filtering}
\label{fig:recall_qpsleq} 
\vspace*{-0.4cm}
\end{figure}

\stitle{Exp-6: Results on keyword  filtering.}
Figure~\ref{fig:recall_qpsleq} shows the keyword filtering performance of different methods. In this experiment, each data vector has a set of keywords, and queries require exactly matching keywords. Among the baselines, UNG achieves the best overall performance, as its index is specifically designed for the sparsity patterns of keyword filtering \cite{UNG}. Despite this specialization, our general‑purpose \mci remains highly competitive, especially in the high‑recall regime on datasets like \texttt{deep1M} and \texttt{wit}, demonstrating its effectiveness even without domain‑specific optimizations.


\begin{figure}[t!]
	\begin{center}
		\includegraphics[width=0.99\linewidth]{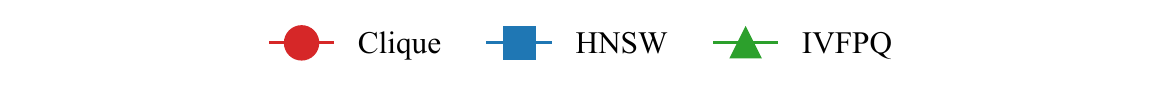}\\
		\subfigure[deep1M]{\label{sfig:nofilter_deep1M1.0}\includegraphics[width=0.32\linewidth]{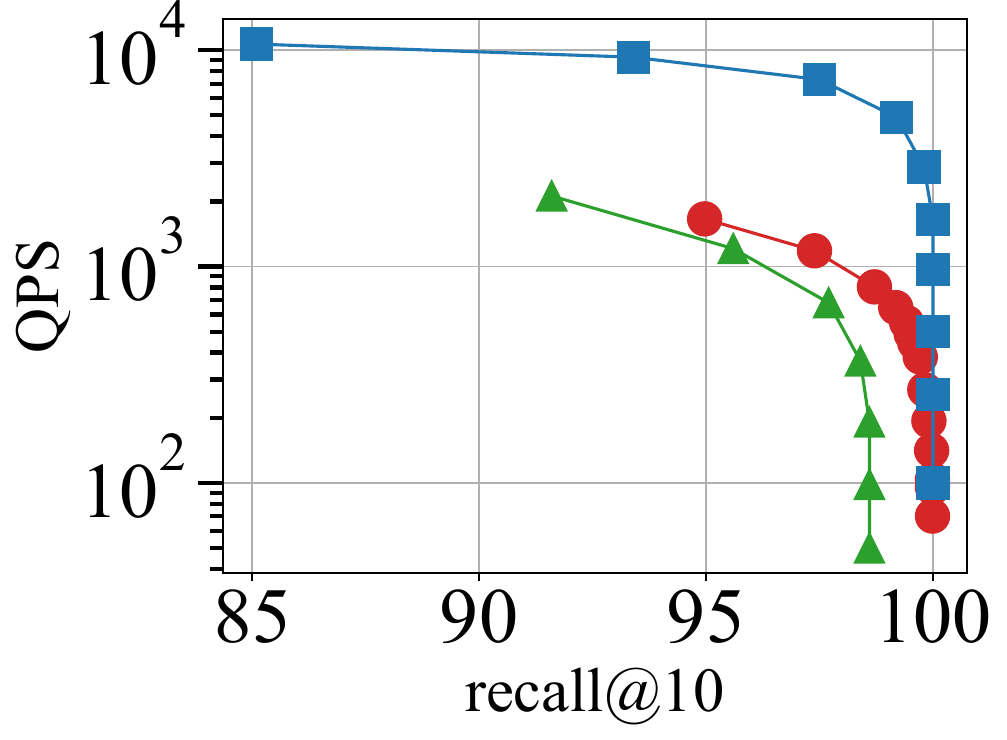}}
		\subfigure[gist1M]{\label{sfig:nofilter_gist1M1.0}\includegraphics[width=0.32\linewidth]{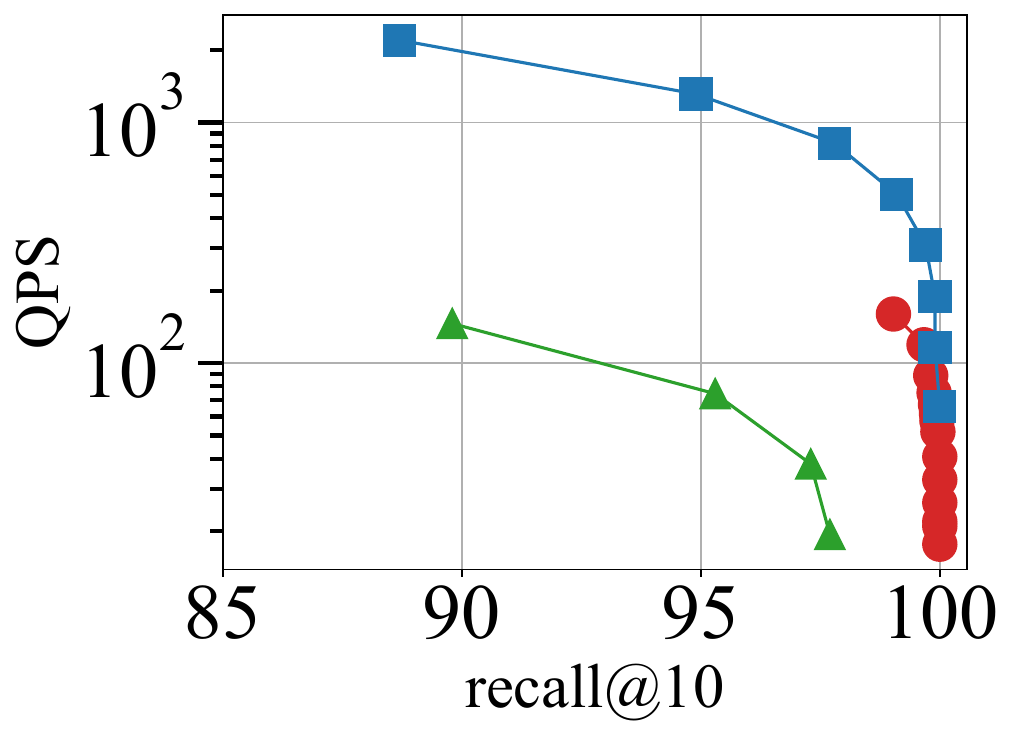}}
		\subfigure[glove2.2m]{\label{sfig:nofilter_glove2.2m1.0}\includegraphics[width=0.32\linewidth]{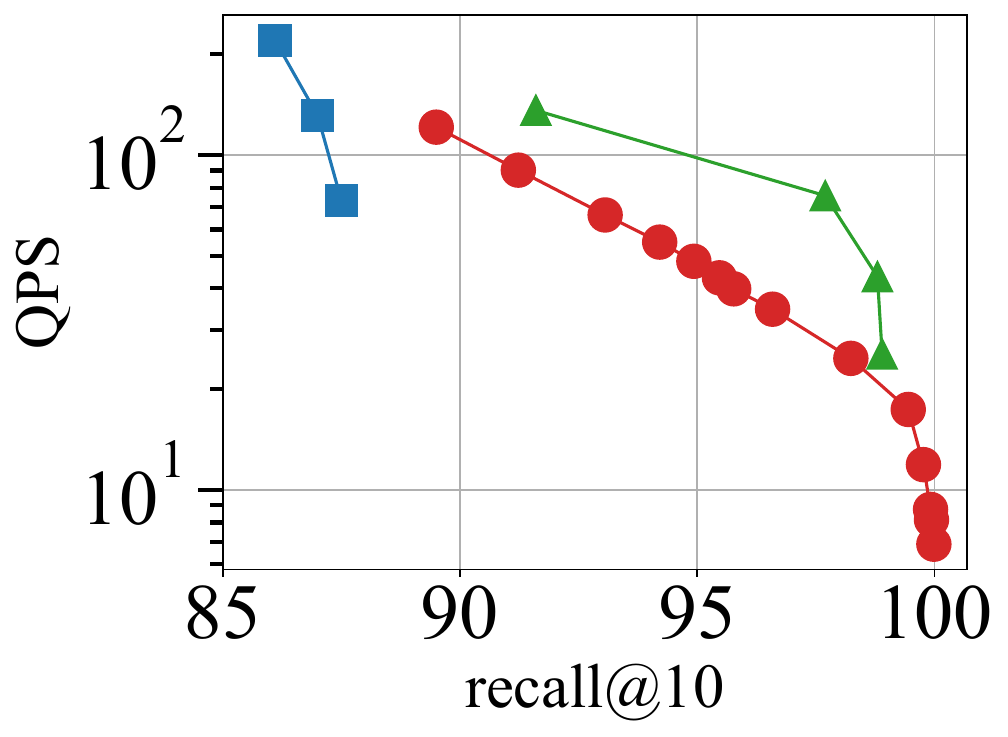}}
	\end{center}
	\caption{No filtering}
	\label{fig:nofilter} 
	
\end{figure}

\stitle{Exp-6: no filtering.} We compare the methods in a standalone ANNS setting. The HNSW and IVFPQ are implemented in the Faiss library \cite{FAISS}. The results are summarized in Figure~\ref{fig:nofilter}. On the deep1M dataset (Figure~\ref{sfig:nofilter_deep1M1.0}), HNSW achieves the highest query throughput, followed by our method \mci, with IVFPQ exhibiting the lowest performance. A similar ranking is observed on gist1M.  The trend differs on glove2.2m (Figure~\ref{sfig:nofilter_glove2.2m1.0}), where IVFPQ demonstrates superior efficiency at moderate recall rates, outperforming both \mci and HNSW. In the high-recall regime (e.g., >98\%), \mci surpasses IVFPQ, while HNSW lags significantly. These results highlight that  \mci provides a more robust and balanced performance, consistently remaining competitive across diverse datasets. This demonstrates the general effectiveness of our approach.

\begin{figure}[t!]
	\begin{center}
\subfigure[deep1M]{\label{sfig:cmp_deep1M}\includegraphics[width=0.32\linewidth]{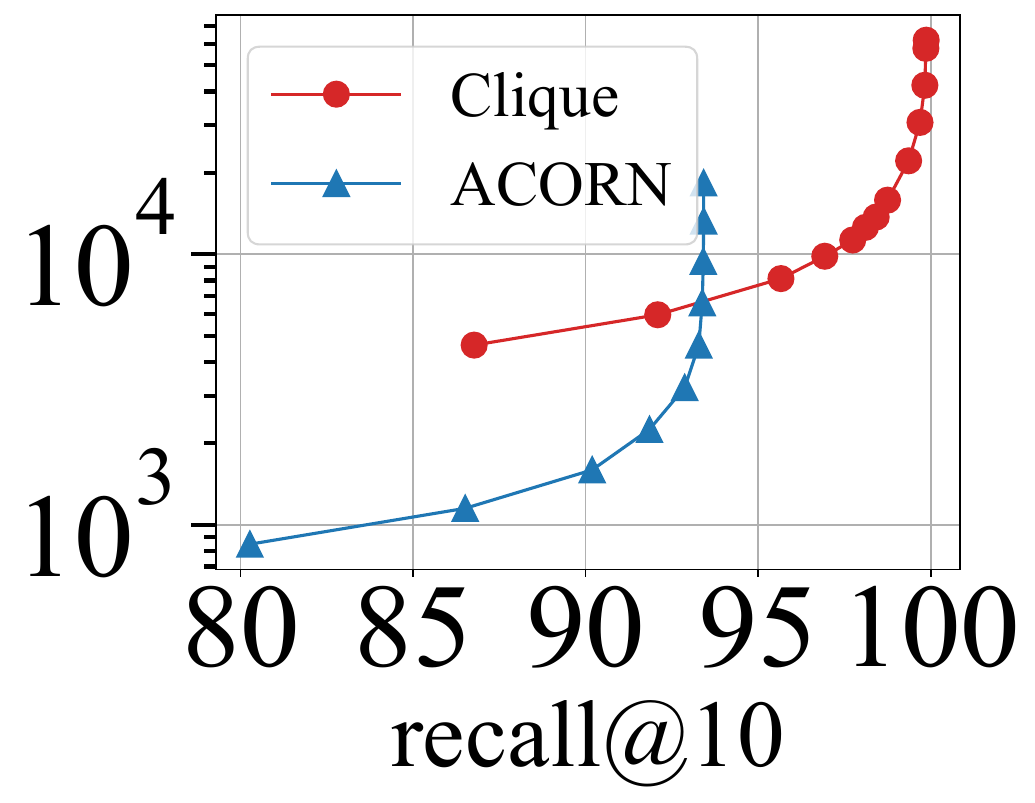}}
\subfigure[glove2.2m]{\label{sfig:cmp_glove2.2m}\includegraphics[width=0.32\linewidth]{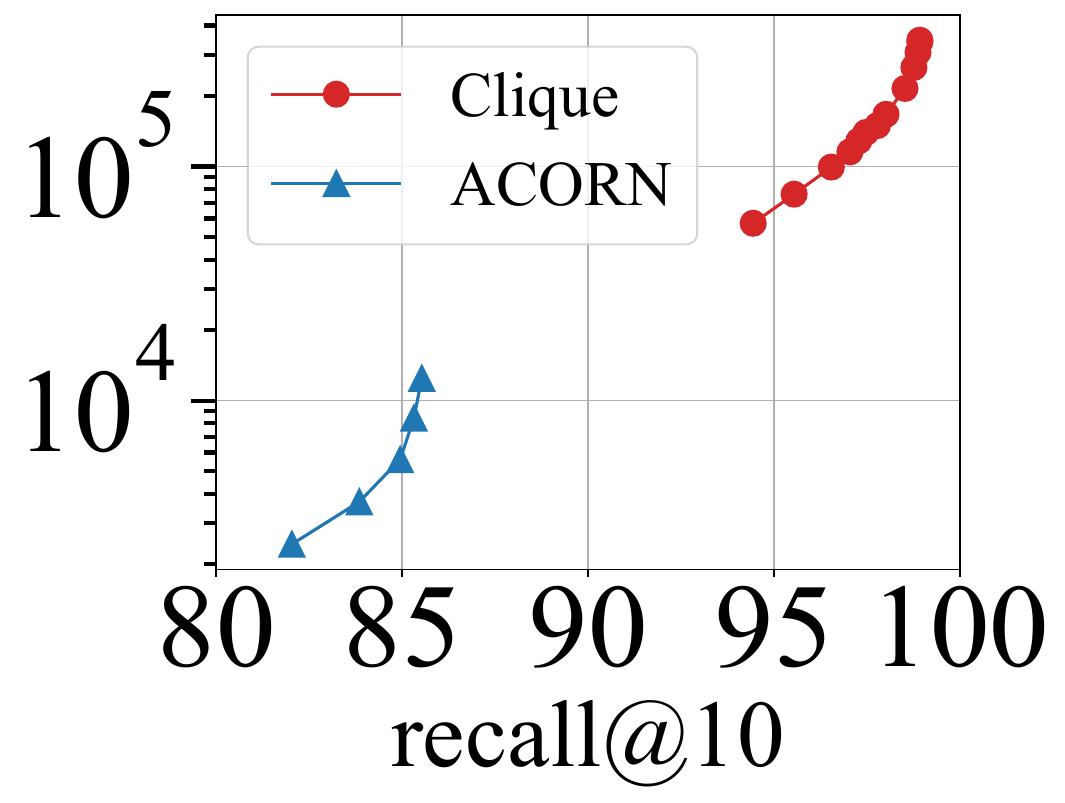}}
\subfigure[tripclick]{\label{sfig:cmp_tripclick}\includegraphics[width=0.32\linewidth]{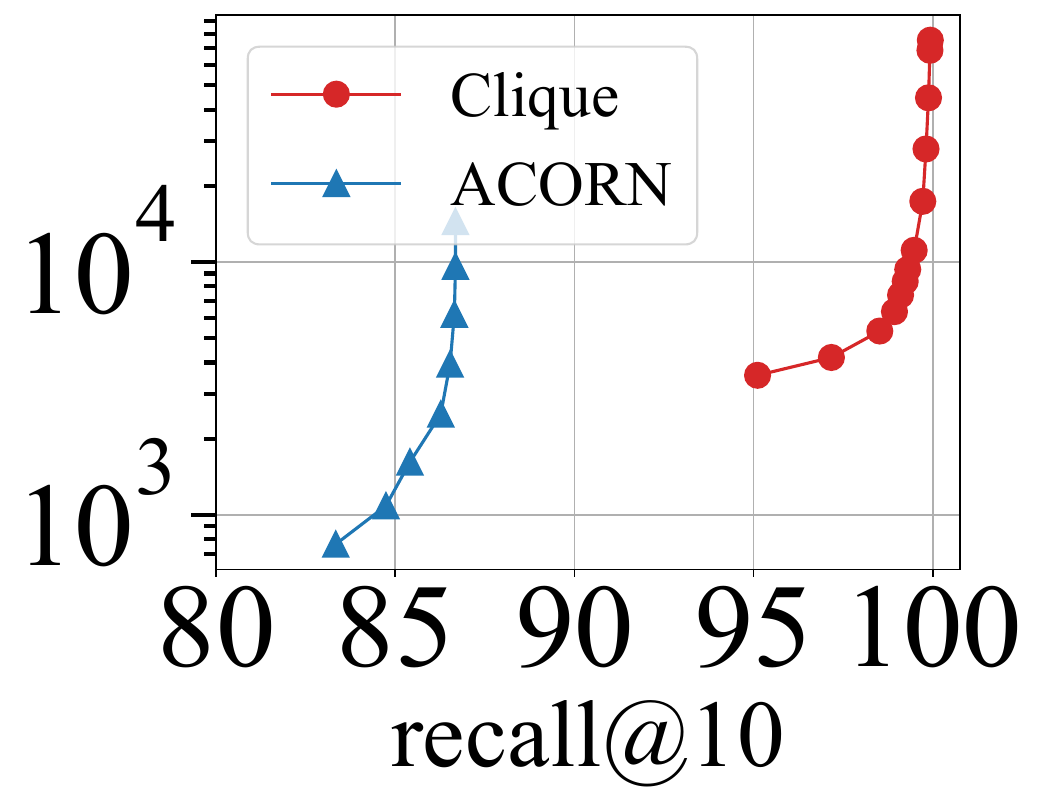}}
\end{center}
\caption{Recall vs distance computation times (y-axis is the times of distance computation)}
\label{fig:cmp} 
\end{figure}

\stitle{Exp-7: recall vs distance computation times.} To further analyze the efficiency of the search process, we compare the recall achieved against the number of distance computations required. This comparison, shown in Figure~\ref{fig:cmp}, is conducted under queries with mixed selectivities.  Similar trends were observed across other datasets. We primarily compare our method, \mci, with \acorn. Algorithms such as Faiss-HNSW, Milvus-HNSW, and IVFPQ are built upon the Faiss and Milvus libraries, which internally manage index traversal and distance calculations.

Consistent with the trends observed in prior results (e.g., Figure~\ref{fig:recall_qps}), \acorn demonstrates an advantage in scenarios requiring lower recall.  Conversely, our method \mci consistently achieves a significantly higher final recall rate. This indicates that \mci is suitable  for high-recall retrieval tasks.

\begin{figure}[t!]
	\vspace*{-0.25cm}
	\begin{center}
		\includegraphics[width=0.98\linewidth]{graph/legendQPS.pdf}\vspace*{-0.4cm}\\
		\subfigure[deep1M, k=50]{\label{sfig:K50_deep1M}\includegraphics[width=0.32\linewidth]{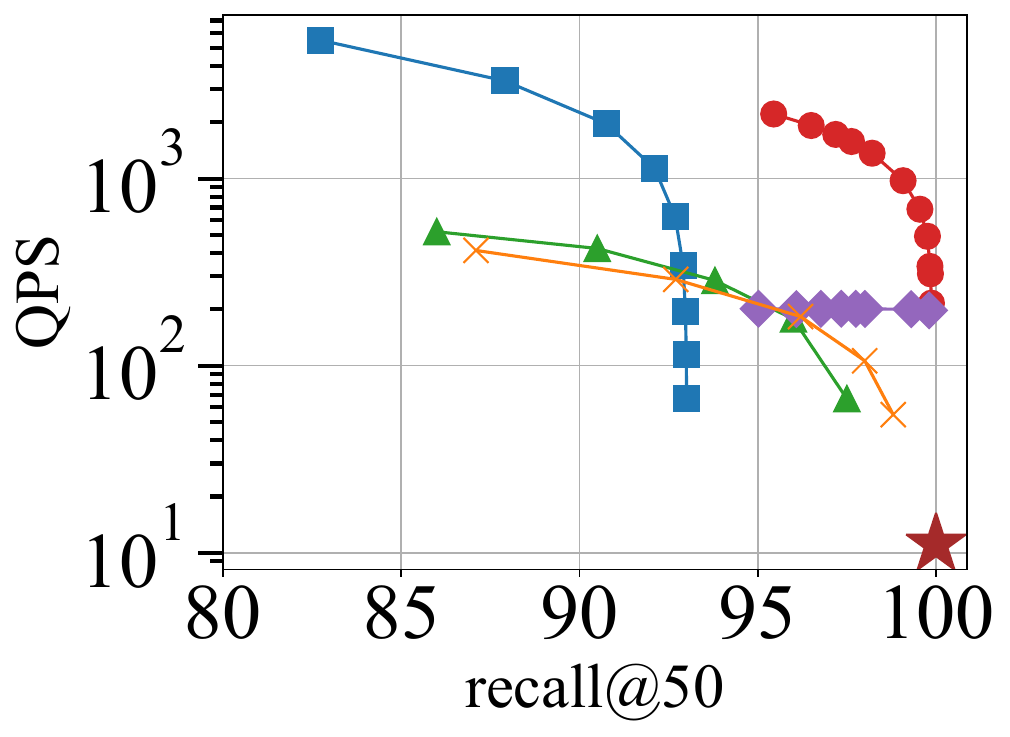}}
        \subfigure[tripclick, k=50]{\label{sfig:K50_tripclick}\includegraphics[width=0.32\linewidth]{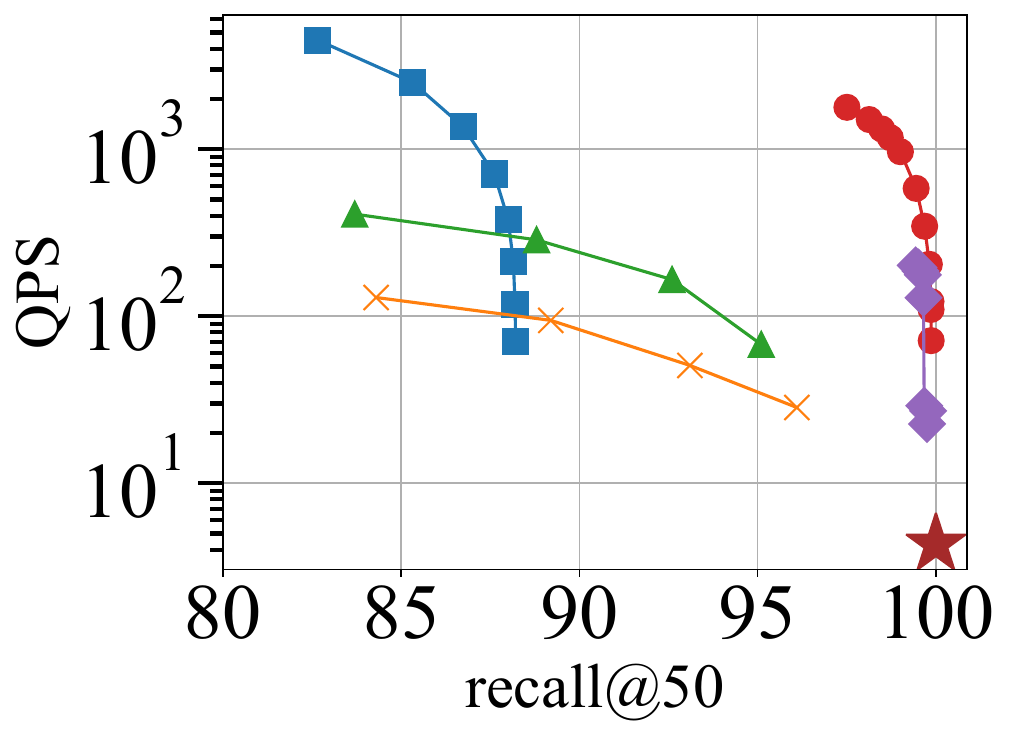}}
		\subfigure[gist1M, k=50]{\label{sfig:K50_gist1M}\includegraphics[width=0.32\linewidth]{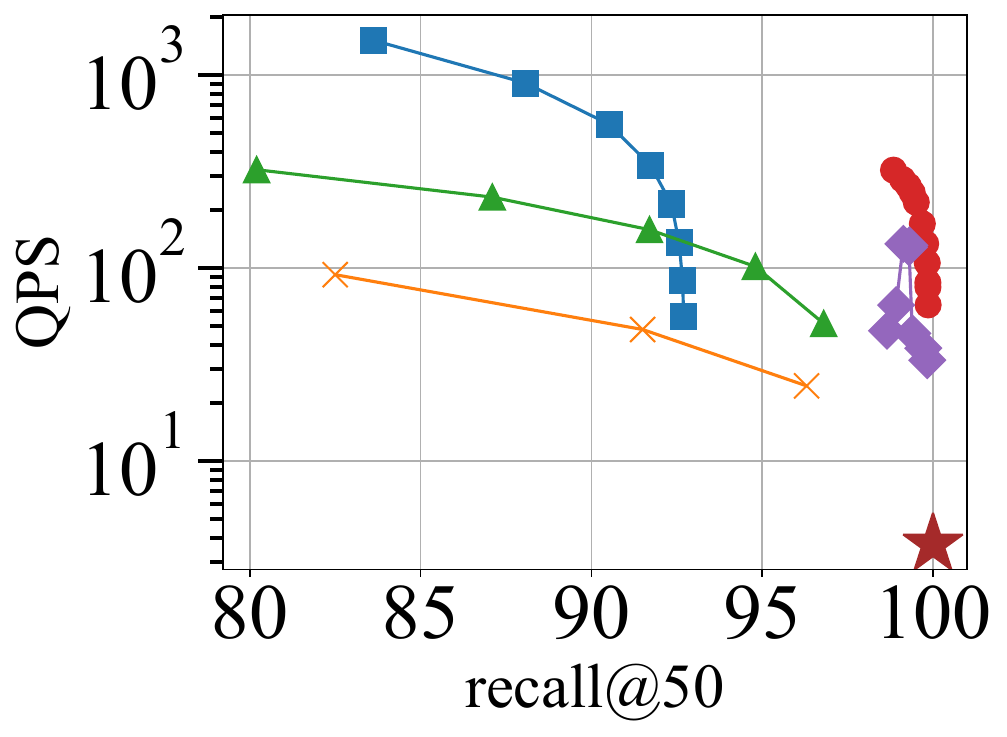}}
		
	\end{center}
	\vspace*{-0.4cm}
	\caption{Recall vs QPS of various methods for $k=50$}
	\label{fig:recall_qps_vark} 
	\vspace*{-0.3cm}
\end{figure}

\stitle{Exp-7: Sensitivity to result set size ($k$).}
To evaluate the impact of the target number of nearest neighbors ($k$), we extended our evaluation to $k=50$, as shown in Figure~\ref{fig:recall_qps_vark}. The overall trend indicates that the relative performance of all evaluated algorithms remains consistent across different $k$ values. Our method, \mci, maintains nearly identical performance profiles regardless of $k$. In contrast, \acorn exhibits minor degradation as $k$ increases. For instance, on the \texttt{deep1M} dataset at a throughput of $10^3$ QPS, the recall for \acorn drops slightly from approximately 93\% at $k=10$ (Figure~\ref{sfig:time_deep1M}) to about 92\% at $k=50$ (Figure~\ref{sfig:K50_deep1M}).
These results confirm that the performance characteristics of AFANNS are generally robust with varying $k$.

\begin{figure}[t!]
	\begin{center}
		\includegraphics[width=1.2\linewidth]{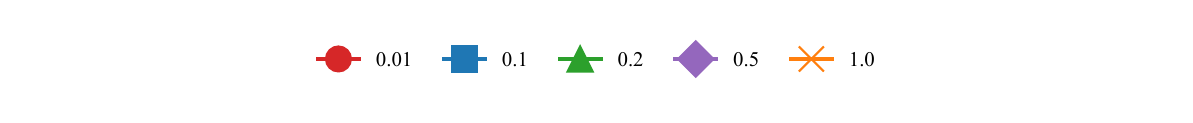}\vspace*{-0.5cm}\\
\subfigure[glove2.2m, s=0.5]{\label{sfig:vareps_glove2.2m0.5}\includegraphics[width=0.32\linewidth]{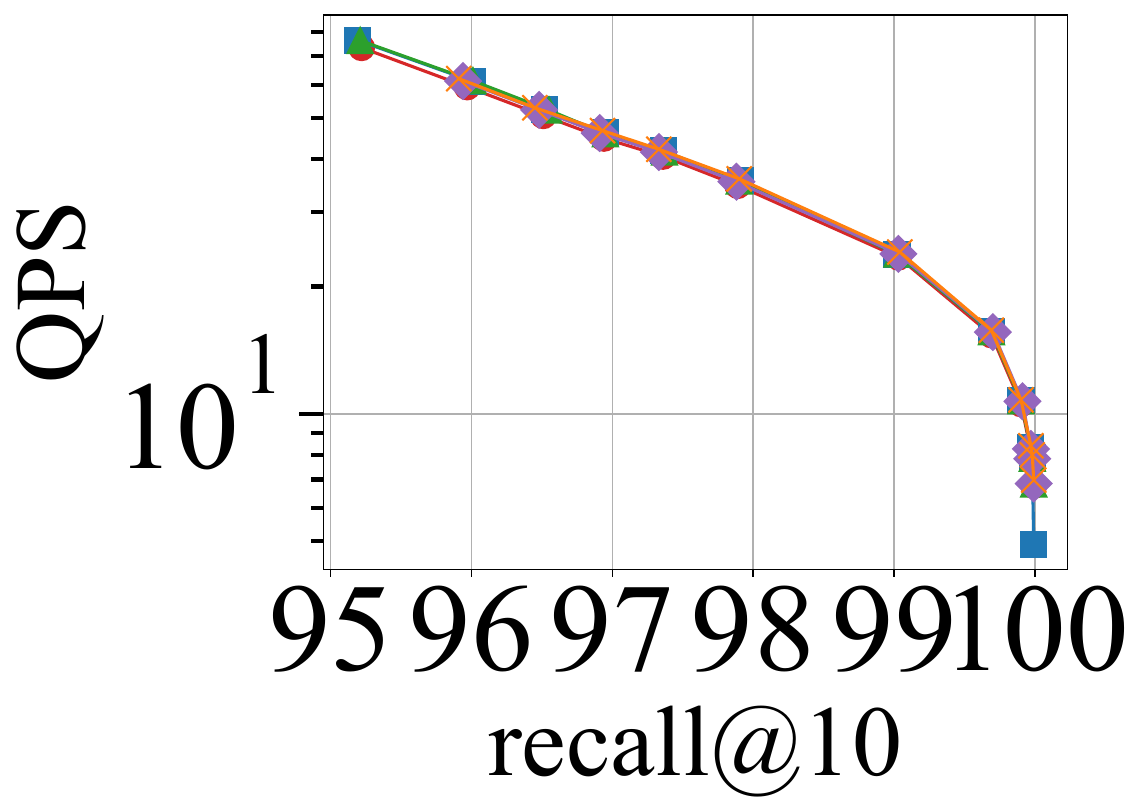}}
\subfigure[glove2.2m, s=0.05]{\label{sfig:vareps_glove2.2m0.05}\includegraphics[width=0.32\linewidth]{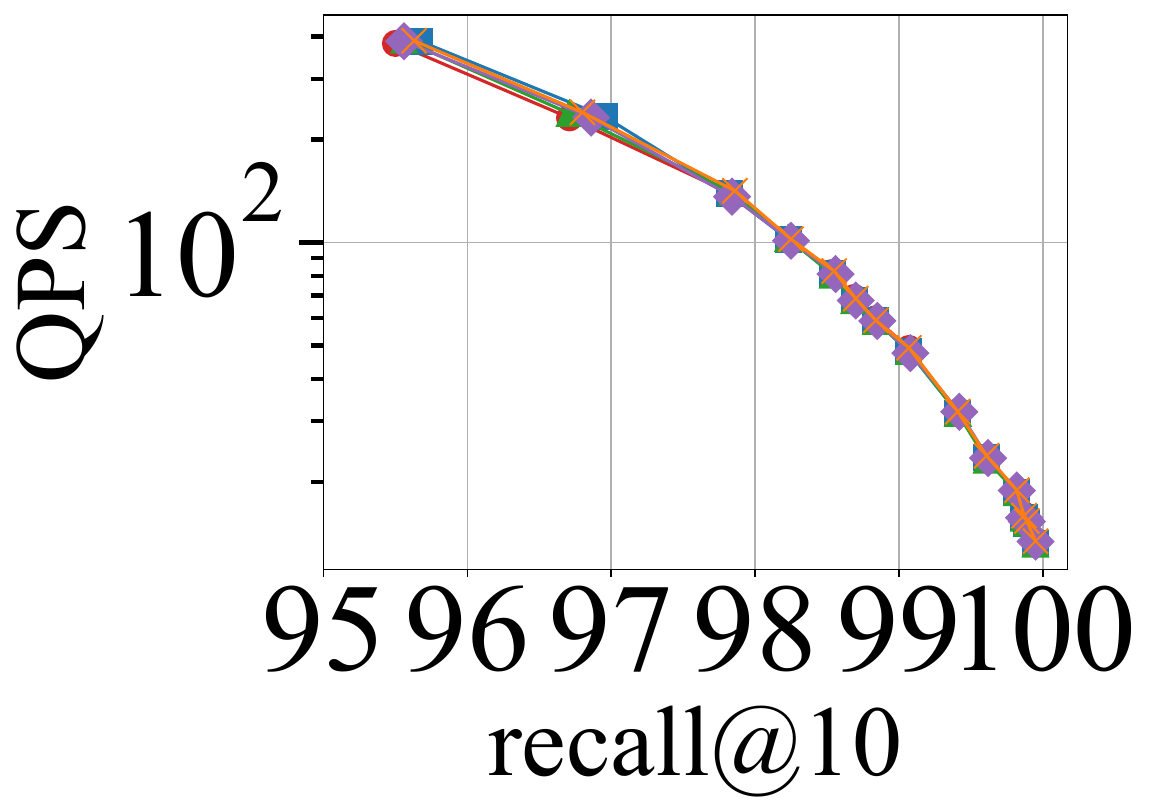}}
\subfigure[glove2.2m, s=0.005]{\label{sfig:vareps_glove2.2m0.005}\includegraphics[width=0.32\linewidth]{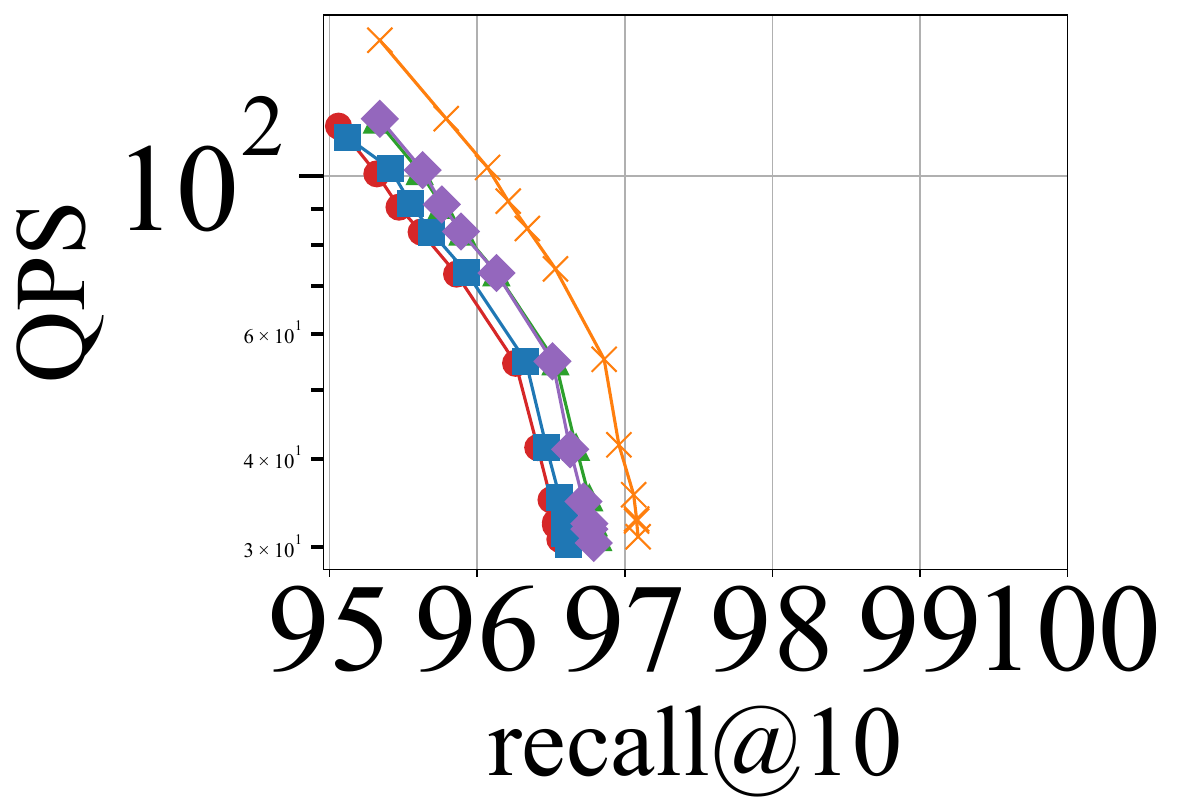}}\vspace*{-0.3cm}
\subfigure[tripclick, s=0.5]{\label{sfig:vareps_tripclick0.5}\includegraphics[width=0.32\linewidth]{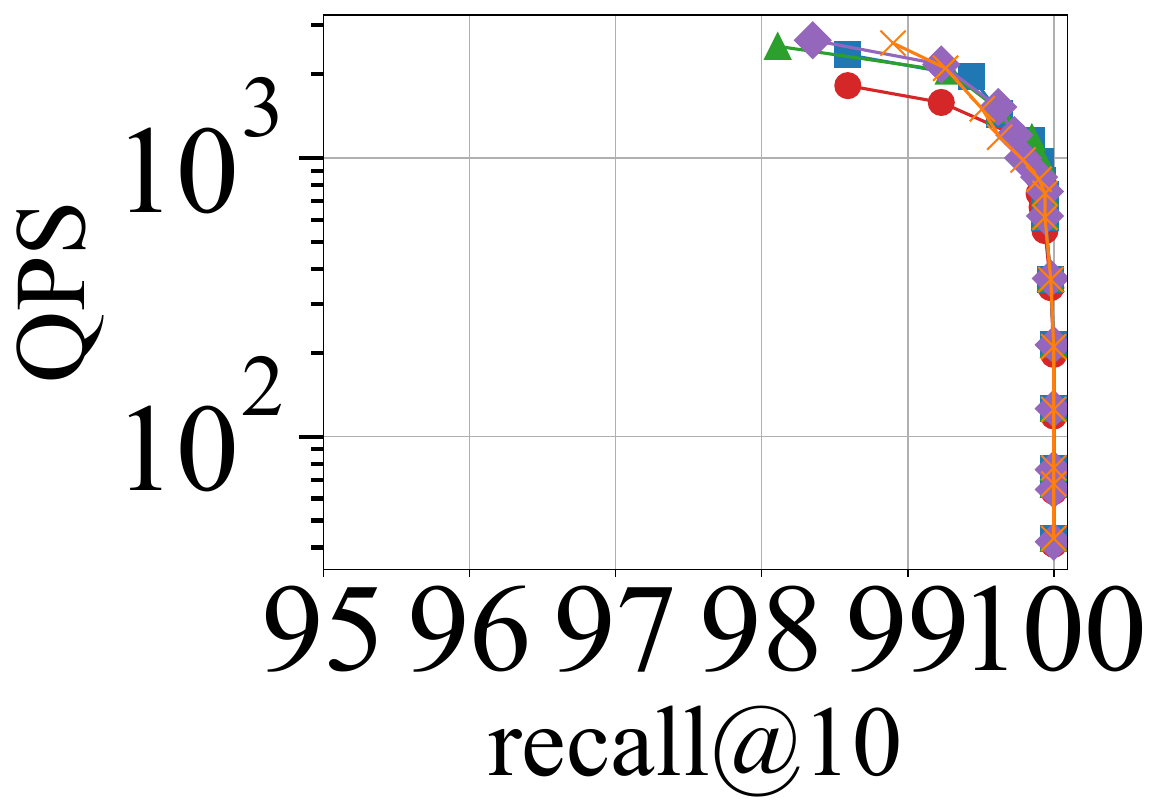}}
\subfigure[tripclick, s=0.05]{\label{sfig:vareps_tripclick0.05}\includegraphics[width=0.32\linewidth]{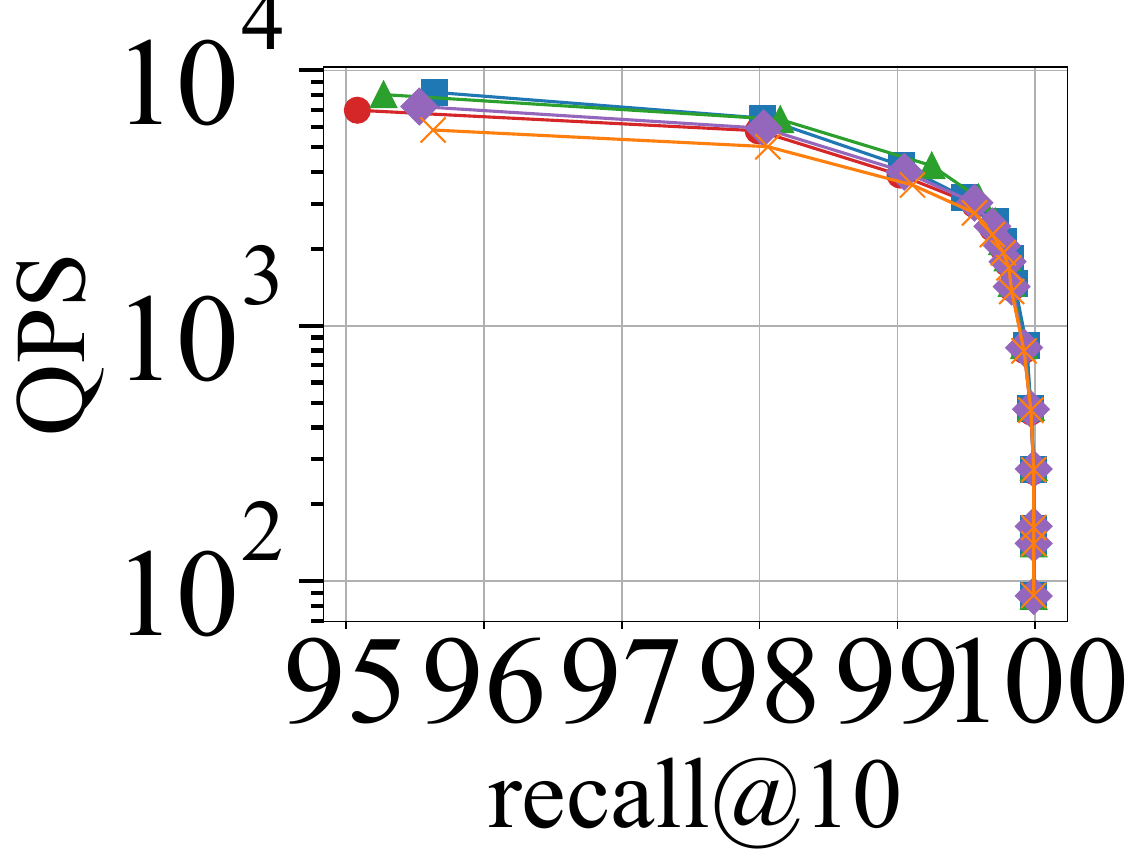}}
\subfigure[tripclick, s=0.005]{\label{sfig:vareps_tripclick0.005}\includegraphics[width=0.32\linewidth]{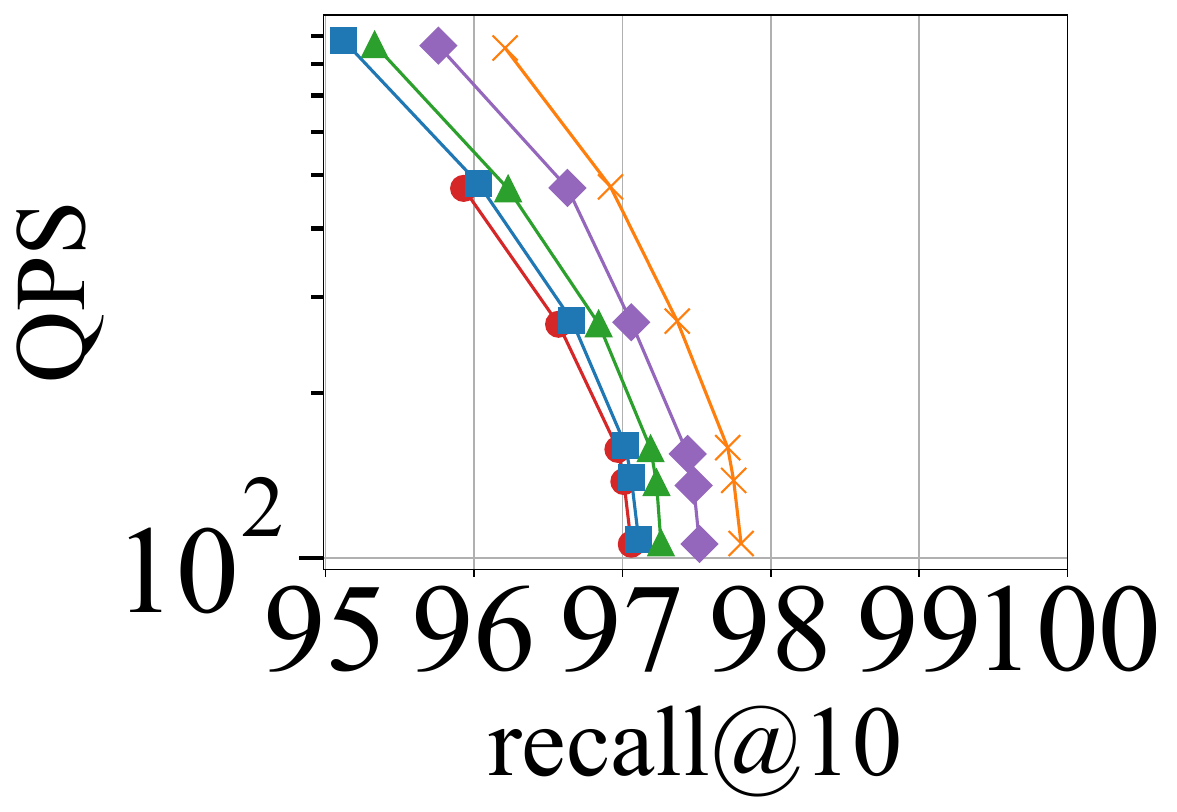}}
\end{center}
\vspace*{-0.3cm}
\caption{Results of our method with varying $\epsilon$}
\label{fig:varepsilon} 
\vspace*{-0.4cm}
\end{figure}

\stitle{Exp‑8: Sensitivity to sampling parameter $\epsilon$.}
We investigate the impact of $\epsilon$, which controls the initial seed count $\epsilon\sqrt{n}$ (Algorithm~\ref{alg:search}). Figure~\ref{fig:varepsilon} reveals a clear pattern: the optimal seed count is inversely correlated with query selectivity $s$. For high selectivity, a small $\epsilon$ suffices because the filtered subgraph is dense and well‑connected. For low selectivity, increasing $\epsilon$ yields significant gains. On \texttt{tripclick} with $s=0.005$, raising $\epsilon$ from $0.01$ to $1$ lifts Recall@10 from $97.06\%$ to $97.8\%$—a non‑trivial improvement in the challenging high‑recall regime. These results suggest a practical tuning heuristic: set $\epsilon$ inversely proportional to the expected selectivity.

\begin{figure*}[t!]
	\vspace*{-0.3cm}
	\begin{center}
\includegraphics[width=0.6\linewidth]{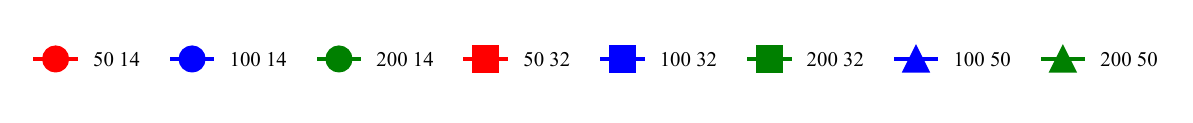}\vspace*{-0.5cm}\\
\subfigure[sift1M]{\label{sfig:vareps_sift1M}\includegraphics[width=0.195\linewidth]{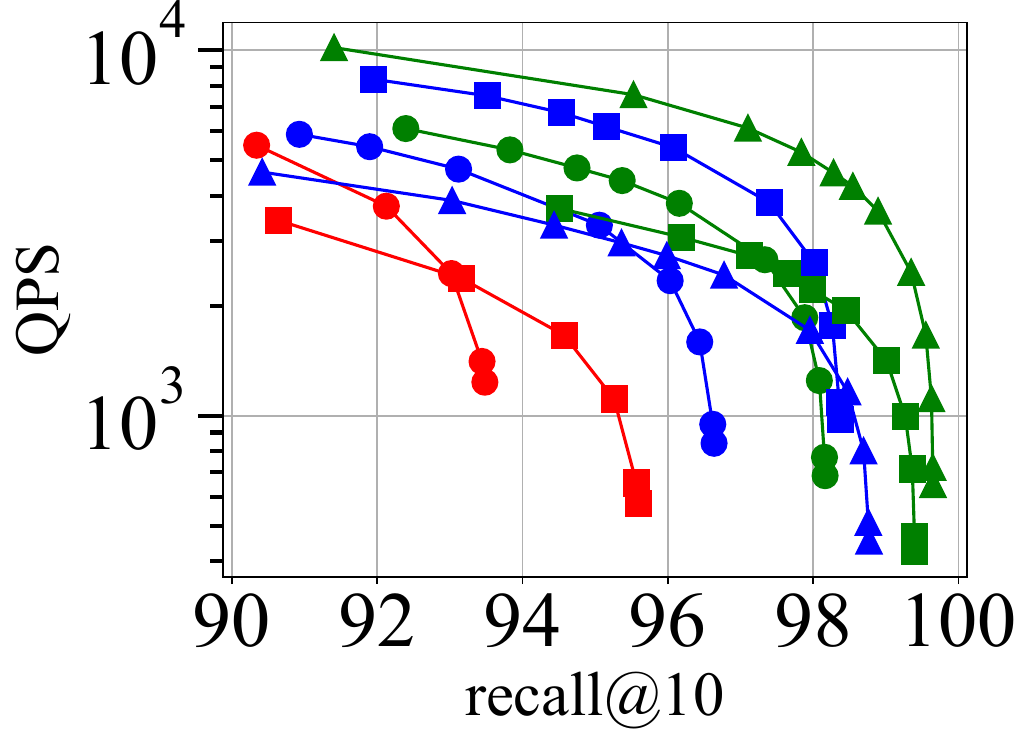}}
\subfigure[deep1M]{\label{sfig:vareps_deep1M}\includegraphics[width=0.195\linewidth]{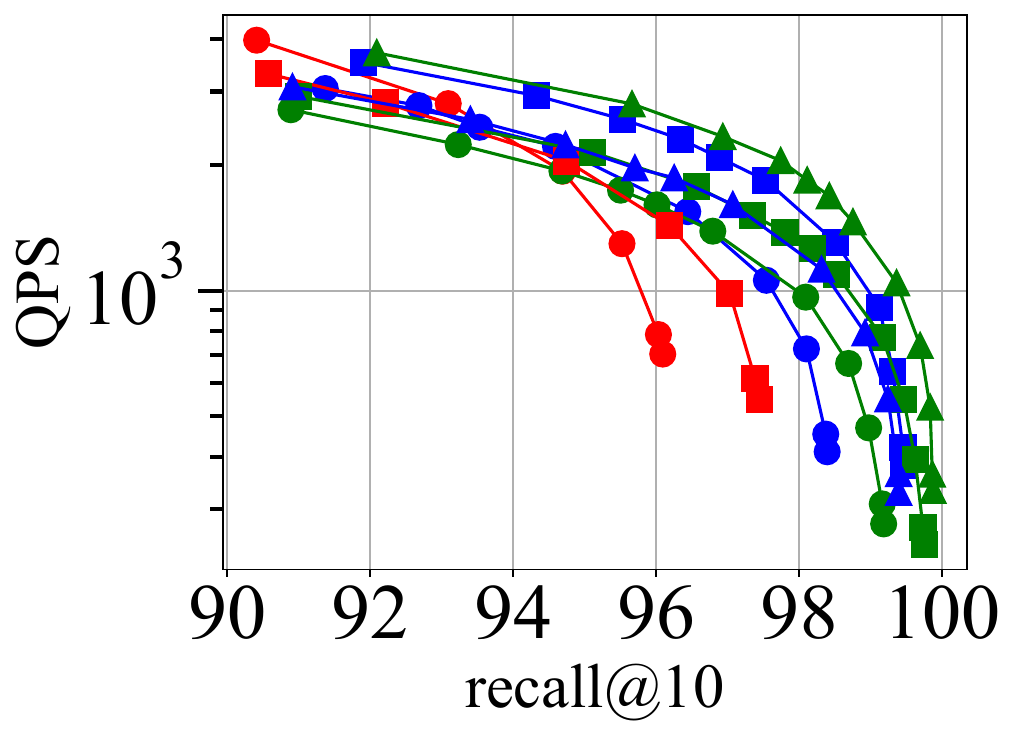}}
\subfigure[glove2.2m]{\label{sfig:vareps_glove2.2m}\includegraphics[width=0.195\linewidth]{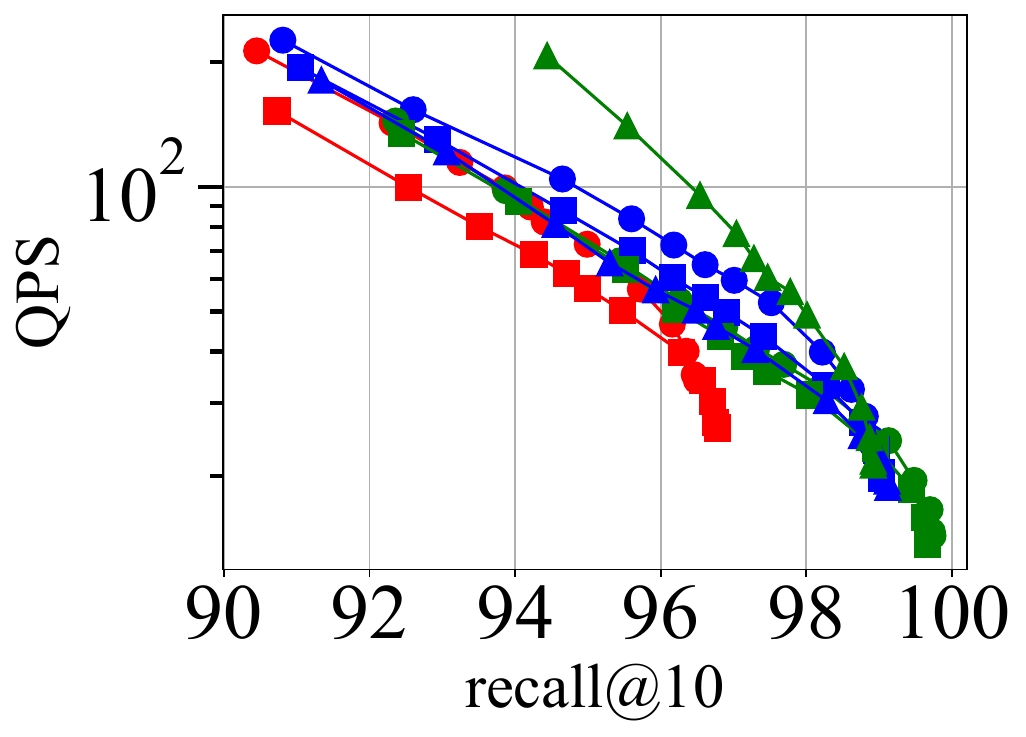}}
\subfigure[gist1M]{\label{sfig:vareps_gist1M}\includegraphics[width=0.195\linewidth]{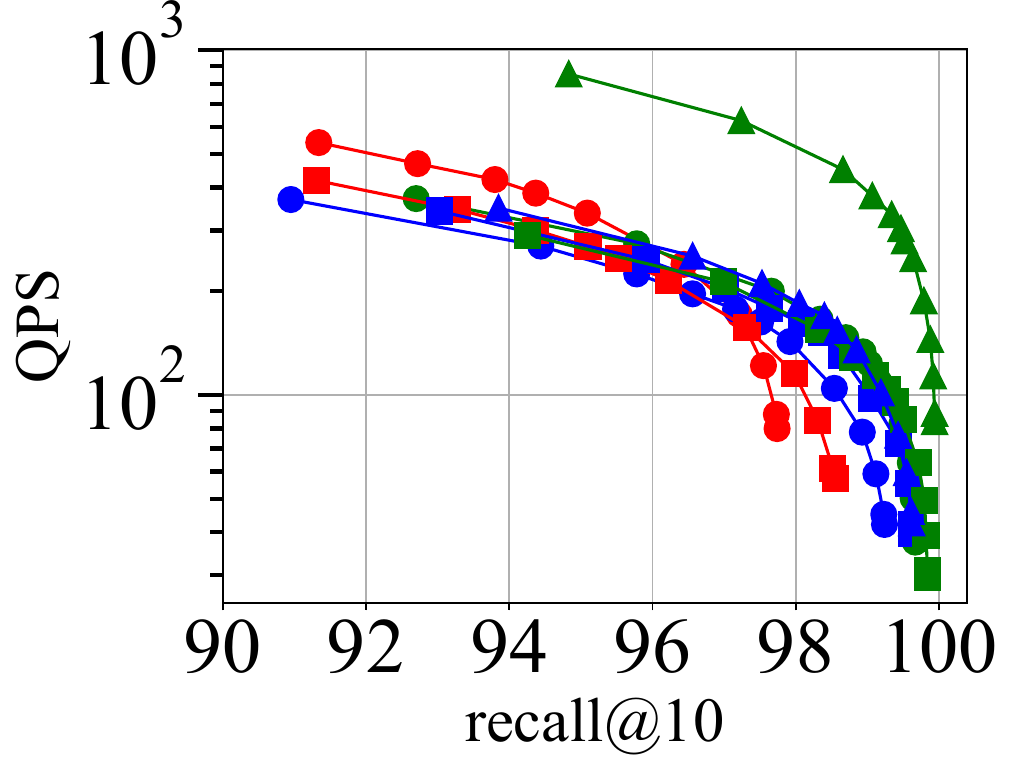}}
\subfigure[wit]{\label{sfig:vareps_wit}\includegraphics[width=0.195\linewidth]{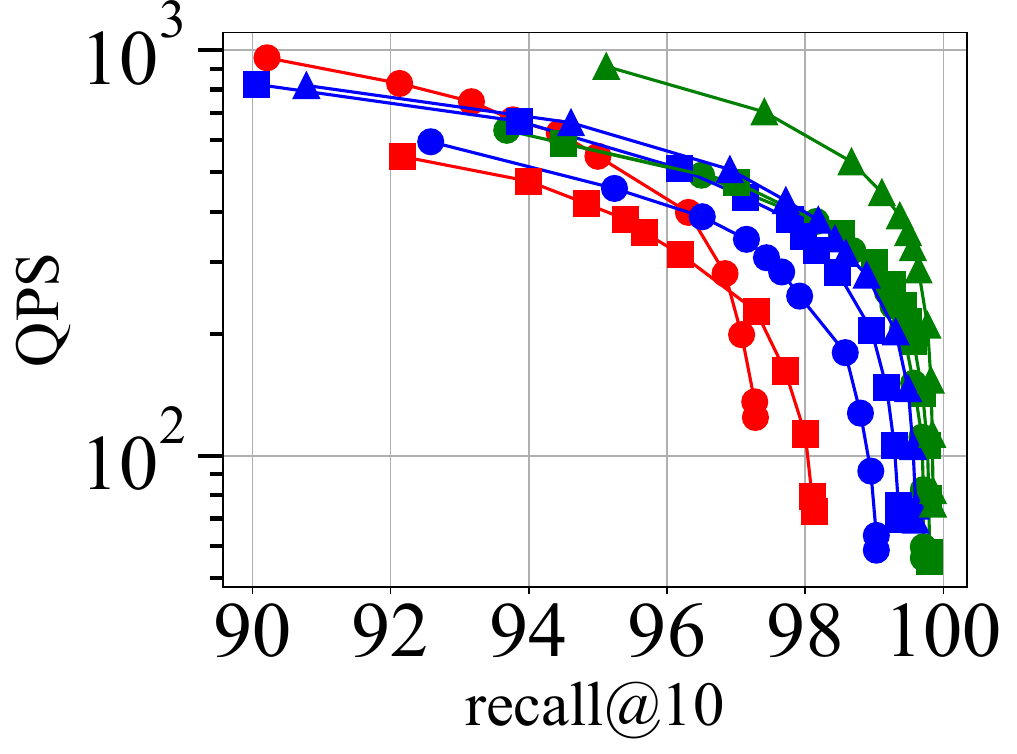}}

\end{center}
\vspace*{-0.4cm}
\caption{Results of our method with varying ($k'$, $\tau$)}
\label{fig:varcliqueminlubild} 
\vspace*{-0.2cm}
\end{figure*}

\stitle{Exp-9: Sensitivity to construction parameters ($k'$ and $\tau$).}
Figure~\ref{fig:varcliqueminlubild} analyzes the impact of the construction parameters $k'$ (initial neighborhood size) and $\tau$ (clique overlap threshold). In the plots, varying $k'$ values ($50, 100, 200$) are represented by red, blue, and green lines, respectively, while different $\tau$ thresholds ($14, 32, 50$) are distinguished by circular, square, and triangular markers.

The results indicate that performance is primarily determined by $k'$. As illustrated, the green curves ($k'=200$) consistently yield better recall/QPS trade-offs than the blue curves ($k'=100$), which in turn outperform the red curves ($k'=50$). In contrast, \mci demonstrates significant robustness to the parameter $\tau$. For instance, on the \texttt{wit} dataset (Figure~\ref{sfig:vareps_wit}), trajectories with the same color (constant $k'$) exhibit nearly identical performance profiles despite differences in $\tau$. This implies that $k'$ is the critical determinant of index quality, while the algorithm is relatively insensitive to the precise tuning of $\tau$.

\begin{table}[t!]
		\small
	\centering
		\vspace*{-0.2cm}
	\caption{Average out-degree of \mci}
		\vspace*{-0.3cm}
	\begin{tabular}{c|c|c|c|c|c|c}
		\toprule
		\multirow{2}{*}{\textbf{Datasets}}	 & \multicolumn{3}{c|}{$k'=100$}   & \multicolumn{3}{c}{$k'=200$}  \\
	\cmidrule{2-7}
	&$\tau=14$	& $\tau=32$ & $\tau=50$& $\tau=14$& $\tau=32$& $\tau=50$ \\
		\midrule
sift1M & 297.1 & 405.4 & 454.2 & 532.8 & 797.5 & 908.9  \\
deep1M & 415.2 & 644.2 & 616.8 & 726.3 & 1176.2 & 1305.7  \\
glove2.2m & 512.8 & 535.5 & 561.9 & 1953.2 & 2067.8 & 1093.7  \\
gist1M & 657.2 & 782.2 & 892.9 & 1457.8 & 1676.7 & 1871.3  \\
wit & 553.4 & 636.1 & 715.7 & 1192.4 & 1329.5 & 1453.7  \\

\bottomrule
\end{tabular}
	\vspace*{-0.4cm}
\label{tab:outdegree}
\end{table}

\stitle{Exp-10: Average out-degree of \mci.}
Table~\ref{tab:outdegree} reports the average out-degree of the \mci graph structure across various parameter settings. A key strength of \mci is its ability to produce an NNG whose effective average degree far exceeds the initial parameter $k'$.
For $k'=100$, the average out‑degree consistently surpasses 300 across all datasets, reaching over 600 in some configurations. This structural densification is pivotal for search performance, particularly in low-selectivity scenarios. Even when a significant fraction of neighbors are pruned by the query predicate, the expanded neighborhood ensures that a sufficient number of valid neighbors remain. Consequently, this high connectivity directly contributes to the algorithm's robustness and high recall. 

\begin{figure}[t]
	\vspace*{-0.3cm}
	\begin{center}
	\end{center}
	\includegraphics[width=0.5\linewidth]{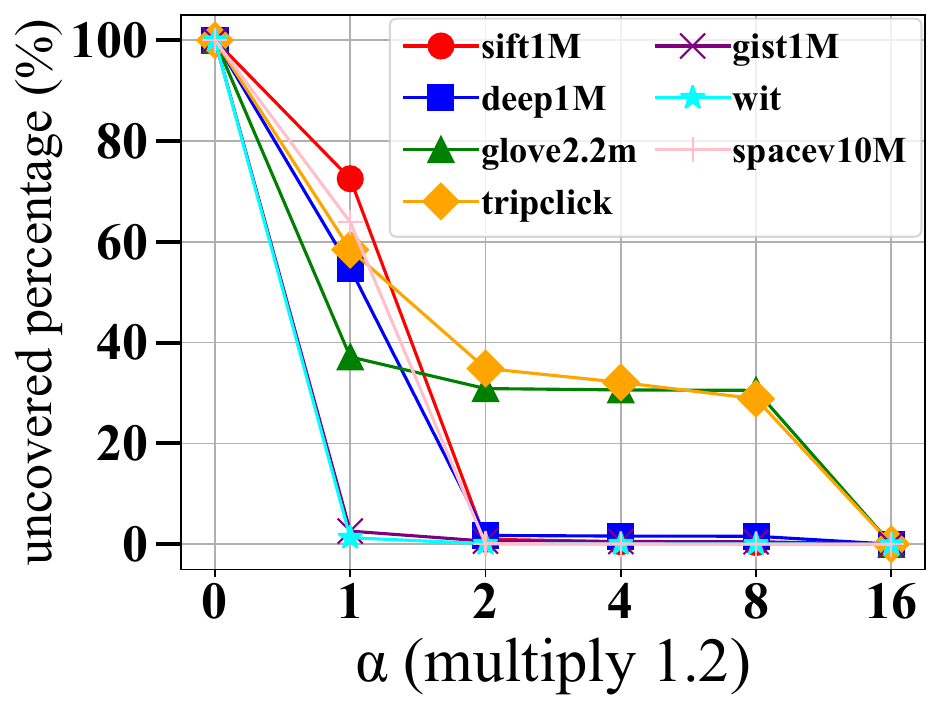}
    \vspace*{-0.3cm}
	\caption{Results of $\alpha$ vs. the percentage of uncovered nodes}
	\label{fig:alphaChange} 
	\vspace*{-0.4cm}
\end{figure}

\stitle{Exp‑11: $\alpha$ vs. percentage of uncovered nodes.}
Figure~\ref{fig:alphaChange} illustrates how the percentage of uncovered nodes decreases as $\alpha$ increases (with $\alpha_{\max}=10$). At $\alpha=2$, over $60\%$ of nodes are already covered across all datasets; the remaining uncovered nodes belong to small, isolated clusters. Full coverage is achieved when $\alpha \ge \alpha_{\max}$, as guaranteed by the fallback mechanism in Lines~17--18 of Algorithm~\ref{alg:construction}.
These results highlight two important properties: (1) the construction algorithm converges in few iterations, and (2) the first two iterations are sufficient to cover the dense clusters, leaving only sparse outliers for later rounds.

\begin{figure}[t!]
	\begin{center}
		\subfigure[sift1M]{\label{sfig:parallelB_sift1M}\includegraphics[width=0.4\linewidth]{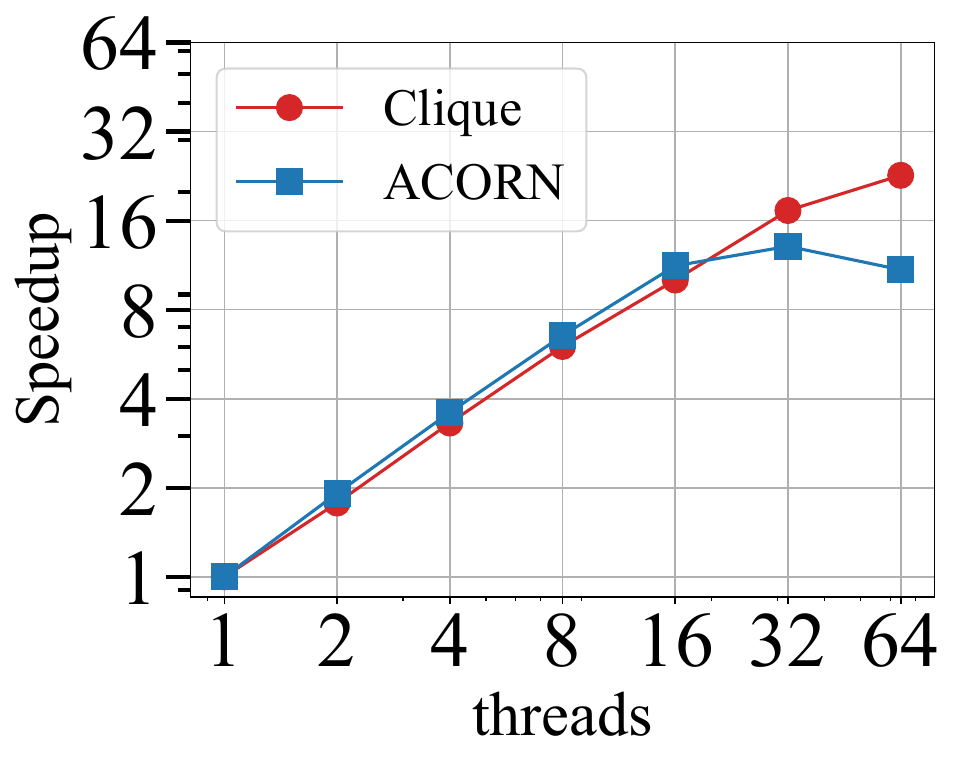}}
		\subfigure[deep1M]{\label{sfig:parallelB_deep1M}\includegraphics[width=0.4\linewidth]{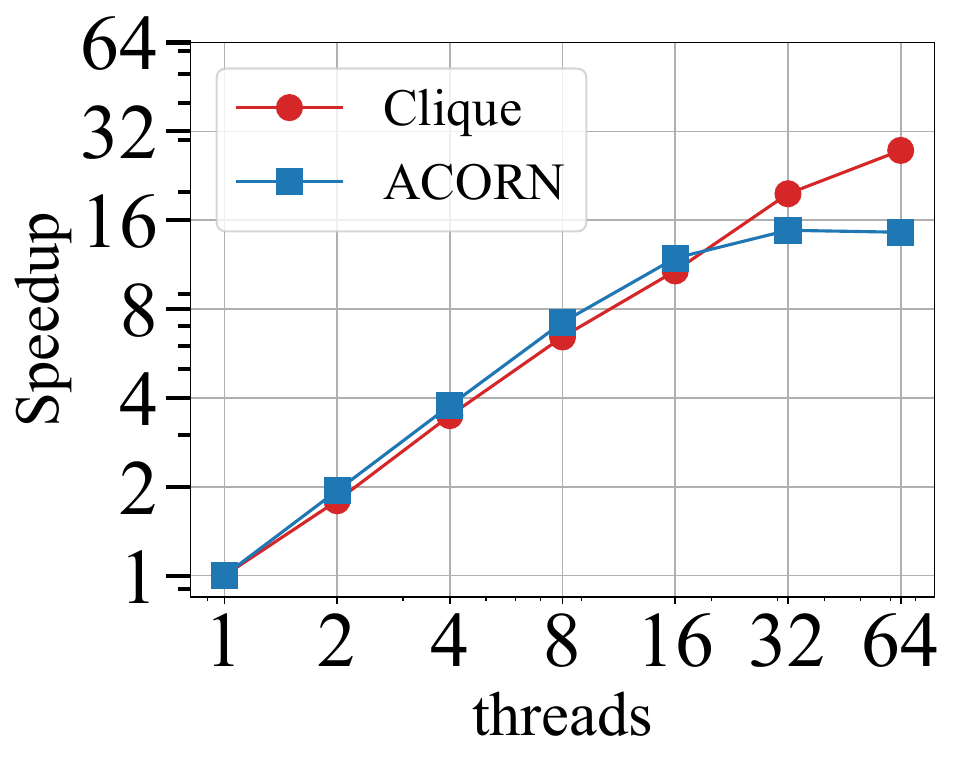}}

	\end{center}
	\vspace*{-0.4cm}
	\caption{Parallel performance of index building algorithms}
	\label{fig:parallelBuilding} 
	\vspace*{-0.3cm}
\end{figure}

\stitle{Exp-12: Parallel speedup of index construction.} Figure~\ref{fig:parallelBuilding} plots the parallel speedup of \mci and \acorn. With 16 or fewer threads, both algorithms exhibit similar speedup ratios. However, at 32 and 64 threads on the sift1M dataset, the speedup of \acorn drops to 13.10 and 10.91, respectively, whereas \mci achieves 17.31 and 22.76. The reduced parallel efficiency of \acorn at higher thread counts stems from its reliance on locks in the HNSW structure~\cite{HNSW}. In contrast, \mci benefits from a lock-free design, allowing it to scale better with available resources. These results confirm our analysis in Section~\ref{subsec:index-construct-alg}.

\subsection{Summary of experimental results}
Table~\ref{tab:summary} presents a systematic comparison of \mci and state-of-the-art methods across key performance metrics, rated using a star system ($\star$: more stars = better performance). The analysis reveals \mci's distinctive advantages: First, \mci matches the storage efficiency of quantization-based methods (like IVFPQ), a notable achievement for a graph-based index. This compact representation reduces memory overhead while preserving rich connectivity. Second, in index construction time (ICT), \mci significantly outperforms ACORN (43\% of ACORN's time) and remains competitive with HNSW. This efficiency stems from our lock-free parallel design and the geometric expansion strategy.
Most importantly, \mci delivers the best query performance—achieving the highest recall at given latency levels, especially under low selectivity. This addresses a critical limitation of existing filtered-ANNS methods. Together, these strengths give \mci the highest overall robustness rating, demonstrating its suitability for diverse real-world workloads where no single dimension can be sacrificed.

\comment{
Table~\ref{tab:summary} presents a qualitative comparison of \mci against state-of-the-art methods across key performance metrics. We employ a star-rating system ($\star$) to denote relative proficiency, where a higher count indicates superior performance (e.g., smaller memory footprint, faster construction, or higher QPS). In terms of Index Size, \mci achieves top-tier efficiency, matching the compactness of quantization-based methods like IVFPQ. Regarding the Index Construction Time (ICT), \mci remains highly competitive; while slightly slower than pure HNSW implementations, it builds significantly faster than ACORN. Crucially, \mci secures the highest rating in Query performance, offering the best trade-off between recall and latency. Consequently, the analysis highlights that \mci provides the most robust overall performance.
}

\comment{
\section{Related Work} \label{sec:rw}
\stitle{Unfiltered ANNS.}
Approximate Nearest Neighbor Search (ANNS) has been extensively studied, with existing indexing techniques broadly categorized into four primary families: tree-based structures \cite{DBLP:journals/pvldb/AroraSK018, DBLP:conf/nips/AndoniILRS15, DBLP:journals/pvldb/ChatzakisFKPP23, DBLP:journals/pami/JegouDS11, DBLP:conf/cvpr/Silpa-AnanH08}, proximity graphs \cite{DiskANN, VISTA, HCNNG, DBLP:conf/www/JungPLOL25, NSW14, CSPG, DEG, NSG, DBLP:journals/pvldb/ChenZHJW24, SteinerHardness, DBLP:journals/pacmmod/PengCCYX23, FANNG, DPG}, learning-based methods \cite{DBLP:conf/sigir/VecchiatoLNB24, DBLP:journals/jmlr/HyvonenJR24, DBLP:conf/dac/0003LYCWM024}, and Locality-Sensitive Hashing (LSH) \cite{DBLSH, PMLSH, DETLSH, DBLP:journals/pvldb/ZhaoTHZZ23}. Among these, Proximity Graph (PG)-based methods have consistently demonstrated superior performance in terms of recall and latency \cite{survey21}.

A central design principle in this domain is the approximation of ideal geometric structures. Many state-of-the-art graph indices, including HNSW \cite{HNSW}, NSG \cite{NSG}, and Vamana \cite{DiskANN}, are engineered as efficient approximations of the Relative Neighborhood Graph (RNG) \cite{DBLP:journals/pr/Toussaint80}. While some methods, such as $\tau$-MNG \cite{DBLP:journals/pacmmod/PengCCYX23} and FANNG \cite{FANNG}, offer theoretical guarantees for greedy routing success, others like HNSW \cite{HNSW} and DPG \cite{DPG} prioritize empirical efficiency over strict guarantees.

Parallel to algorithmic structural improvements, significant research has focused on mitigating the distance computation bottleneck. Key directions include dimensionality reduction and quantization techniques \cite{SymphonyQG, SubspaceCollision, DimReducSurvey, DBLP:journals/pvldb/DengCZWZZ24, FINGER, Gao23, DBLP:conf/asplos/LiuNLFGCLG024, DBLP:conf/nips/JaasaariHR24, DBLP:conf/icde/YueX0TL24, RaBitQ, Rabitqv2, DBLP:conf/icml/Lu0I24}, as well as hardware-aware optimizations targeting CPU pipelining \cite{DBLP:conf/ppopp/ManoharSBD0S024}, SIMD vectorization \cite{DBLP:conf/mir/AndreKS17}, GPU acceleration \cite{DBLP:conf/icde/OotomoNNWFW24}, and I/O efficiency \cite{DBLP:conf/isca/WangLZSLCLC24, DBLP:conf/usenix/TianLDL0024}.

\stitle{Filtered ANNS.}
The challenge of ANNS intensifies when queries incorporate filtering predicates. Prior research has largely focused on optimizing for specific filter types.
For \textit{range filtering}, where vectors possess numerical attributes and queries specify target ranges, approaches typically construct specialized indices tailored to the attribute distribution. Prominent examples include SeRF \cite{SeRF}, which embeds range-aware edges into the graph, and hybrid structures like $\beta$-WST \cite{DBLP:conf/icml/EngelsLYDS24} and iRange \cite{iRangeGraph}, which integrate binary search trees over attribute ranges \cite{DBLP:conf/icml/EngelsLYDS24, SeRF, DSG, iRangeGraph, UNIFY}. A related variant, range-containment filtering \cite{DBLP:conf/kdd/0043CZ25}, further extends this to scenarios where both data and queries possess range extents.
For \textit{categorical or label-based filtering} \cite{filteredDiskann,NHQ}, state-of-the-art methods such as UNG \cite{UNG} physically partition the dataset by label and utilize Directed Acyclic Graphs (DAGs) to encode label-set inclusion relationships, achieving high efficiency for this specific use case.

However, these specialized algorithms lack the flexibility to handle \textit{arbitrary} filtering conditions, where predicates are unpredictable and not confined to a single attribute type. General-purpose systems like FAISS \cite{FAISS} and Milvus \cite{milvus} support filtered search typically via post-filtering or simple pre-filtering heuristics, which often yield suboptimal performance. Recent learning-based approaches, such as SIEVE \cite{SIEVE}, attempt to optimize this process but rely on stable historical query workloads for training, severely limiting their applicability in dynamic or cold-start scenarios. While ACORN \cite{ACORN} represents a significant advance in handling arbitrary filters via HNSW-based navigation, our empirical analysis reveals limitations in achieving high recall under low-selectivity conditions. This gap motivates the design of \mci: a robust, graph-based index explicitly engineered to deliver high-performance arbitrary filtered ANNS across diverse selectivity regimes.
}

\section{Related work} \label{sec:rw}

\stitle{Unfiltered ANNS.}
Approximate Nearest Neighbor Search (ANNS) is a well‑studied problem, with methods falling into four main categories: tree‑based structures \cite{DBLP:journals/pvldb/AroraSK018, DBLP:conf/nips/AndoniILRS15, DBLP:journals/pvldb/ChatzakisFKPP23, DBLP:journals/pami/JegouDS11, DBLP:conf/cvpr/Silpa-AnanH08}, proximity graphs (PG) \cite{DiskANN, VISTA, HCNNG, DBLP:conf/www/JungPLOL25, NSW14, CSPG, DEG, NSG, DBLP:journals/pvldb/ChenZHJW24, SteinerHardness, DBLP:journals/pacmmod/PengCCYX23, FANNG, DPG}, learning‑based methods \cite{DBLP:conf/sigir/VecchiatoLNB24, DBLP:journals/jmlr/HyvonenJR24, DBLP:conf/dac/0003LYCWM024}, and Locality‑Sensitive Hashing (LSH) \cite{DBLSH, PMLSH, DETLSH, DBLP:journals/pvldb/ZhaoTHZZ23}. Among these, PG‑based indexes have shown the best recall‑latency trade‑off in practice \cite{survey21}. A key design principle is approximating ideal geometric structures. Many top‑performing graphs (e.g., HNSW \cite{HNSW}, NSG \cite{NSG}, Vamana \cite{DiskANN}) are efficient approximations of the Relative Neighborhood Graph (RNG) \cite{DBLP:journals/pr/Toussaint80}. Some methods (e.g., $\tau$‑MNG \cite{DBLP:journals/pacmmod/PengCCYX23}, FANNG \cite{FANNG}) provide theoretical routing guarantees, while others (e.g., HNSW \cite{HNSW}, DPG \cite{DPG}) prioritize empirical efficiency. Complementary research focuses on reducing the distance‑computation bottleneck via quantization \cite{SymphonyQG, SubspaceCollision, DimReducSurvey, DBLP:journals/pvldb/DengCZWZZ24, FINGER, Gao23, DBLP:conf/asplos/LiuNLFGCLG024, DBLP:conf/nips/JaasaariHR24, DBLP:conf/icde/YueX0TL24, RaBitQ, Rabitqv2, DBLP:conf/icml/Lu0I24} and hardware‑aware optimizations (CPU pipelining \cite{DBLP:conf/ppopp/ManoharSBD0S024}, SIMD \cite{DBLP:conf/mir/AndreKS17}, GPU \cite{DBLP:conf/icde/OotomoNNWFW24}, I/O \cite{DBLP:conf/isca/WangLZSLCLC24, DBLP:conf/usenix/TianLDL0024}).

\stitle{Filtered ANNS.}
Incorporating filtering predicates makes ANNS substantially more challenging because the index must preserve both geometric proximity and predicate reachability. As summarized in Table~\ref{tab:methods}, most prior work is specialized to one predicate family. For \textit{range filtering} on numerical attributes, static indexes such as SeRF \cite{SeRF}, SuperPF/the $\beta$‑WST framework of Engels et al. \cite{DBLP:conf/icml/EngelsLYDS24}, iRangeGraph \cite{iRangeGraph}, and UNIFY \cite{UNIFY} exploit attribute order through range-aware graph structures or hybrid trees. More recent \textit{dynamic} RFANNS methods, including DSG \cite{DSG}, DIGRA \cite{DIGRA}, and RangePQ+ \cite{RangePQ}, further support online updates while still assuming scalar attributes and interval predicates. For \textit{keyword, label, or tag filtering}, Filtered-DiskANN \cite{filteredDiskann}, NHQ \cite{NHQ}, UNG \cite{UNG}, and TFANNS \cite{TFANNS} leverage discrete set-membership semantics via label-aware graphs, attribute indexes, DAG-style inclusion structures, or tag-frequency-aware graph construction; these methods are highly effective in their target settings but typically assume static categorical metadata.

For \textit{arbitrary filtering}, the classical baselines are pre-filtering and post-filtering, which are effective only in favorable selectivity regimes. General systems such as FAISS \cite{FAISS} and Milvus \cite{milvus} usually implement filtered search through these heuristics or lightweight variants, often sacrificing robustness across mixed workloads. Learning-based methods such as SIEVE \cite{SIEVE} assume stable historical query distributions, which limits applicability in dynamic or cold-start settings. ACORN \cite{ACORN} is the main predicate-agnostic graph index in this line, extending HNSW with enlarged predicate-robust neighborhoods; however, as also confirmed by our experiments, it still incurs substantial memory overhead and can struggle to maintain high recall under low selectivity. This gap motivates \mci, which seeks to retain the generality of arbitrary filtering without relying on predicate-specific structures or workload-specific assumptions.

\newcommand{\stars}[1]{\ifcase#1\or $\star$\or $\star\star$\or $\star\star\star$\or $\star\star\star\star$\or $\star\star\star\star\star$\fi}

\begin{table}[t!]
     \small
    \centering
    \caption{Summary of performance comparison}
    \vspace*{-0.2cm}
    \begin{tabular}{lcccc}
        \toprule
        \textbf{Method} & \textbf{Index Size} & \textbf{ICT} & \textbf{Query} & \textbf{Overall} \\
        \midrule
        \acorn \cite{ACORN}         & \stars{3} & \stars{3} & \stars{3} & \stars{3} \\
        Faiss-HNSW \cite{FAISS}     & \stars{3} & \stars{5} & \stars{3} & \stars{3} \\
        Milvus-HNSW \cite{milvus}   & \stars{2} & \stars{5} & \stars{3} & \stars{3} \\
        IVFPQ \cite{FAISS}          & \stars{5} & \stars{4} & \stars{2} & \stars{2} \\
        \textbf{\mci (Ours)}        & \stars{5} & \stars{4} & \stars{5} & \stars{5} \\
        \bottomrule
    \end{tabular}
    \vspace*{-0.5cm}
    \label{tab:summary}
\end{table}

\section{Conclusion} \label{sec:conclusion}
This work addresses the Arbitrary Filtered Approximate Nearest Neighbor Search (AFANNS) problem. We present the Maximal Clique Index (\mci), a novel graph-based index designed for AFANNS. By employing a Maximal Clique Cover strategy and a  neighborhood graphs geometric densification strategy, \mci exploits the inherent clique structure of the graph to achieve  high compression and connectivity. Our results demonstrate that \mci significantly outperforms state-of-the-art baselines in both high-recall regimes and storage efficiency, offering a robust solution for the complex filtering requirements of modern vector databases.

\bibliographystyle{ACM-Reference-Format}
\balance
\bibliography{reference}

\newpage
\appendix
\section{Discussion on dynamic updates}

Although our current implementation focuses on static construction, \mci is designed to support efficient dynamic operations without global re-indexing. We propose the following lightweight strategies, which we leave for future work.

\stitle{Insertion.} We utilize a heuristic grounded in geometric transitivity (Theorem~\ref{the:nnarenn}). To insert a node $u$, we first retrieve its approximate nearest neighbors $R$ using the existing index (cost $\approx O(n^{0.55})$). We then append $u$ to any maximal clique $C$ satisfying the overlap condition $|C \cap R| \ge \sqrt{|C|}$. This overlap serves as a proxy for full connectivity, leveraging distance concentration to avoid expensive verification. If $u$ is an outlier (i.e., no such clique exists), we locally mine a new maximal clique from $R \cup \{u\}$ using the \minCliques procedure in Algorithm~\ref{alg:construction}.

\stitle{Deletion.} Deletion is straightforward: the target node is simply removed from all maximal cliques containing it. If a clique's size drops below the threshold $\tau$, locally mine the maximal cliques for remaining nodes using the \minCliques procedure in Algorithm~\ref{alg:construction}.

\section{Additional experiments}

\stitle{Exp-13: Sensitivity to the quality of $k'$-NNG.} The NN-Descent algorithm relies on iterations to converge; generally, a higher number of iterations yields a higher-quality $k'$-NNG. We measure the quality of a $k'$-NNG using Recall@$k'$. At 6, 8, 10, and 12 iterations, the quality (Recall@$k'$) is 0.39, 0.67, 0.81, and 0.91, respectively. Although iterations 8, 10, and 12 produce $k'$-NNGs of varying qualities, the query performance (Recall-QPS) of \mci remains stable, as shown in Figure~\ref{sifg:nnvarite}. These results demonstrate that \mci is not sensitive to the quality of the underlying $k'$-NNG.

\stitle{Exp-14: Sensitivity to $\alpha$ expansion ratio.} In the default setting, $\alpha$ doubles at each step (Line~8 of Algorithm~\ref{alg:construction}). Figure~\ref{sfig:alphaExpansionRatio} plots the performance under various expansion ratios: $1.25, 1.5, 1.75,$ and $2.0$. An expansion ratio of $1.25$ implies the update $\alpha \leftarrow \alpha \times 1.25$ at Line~8 of Algorithm~\ref{alg:construction}. The results in Figure~\ref{sfig:alphaExpansionRatio} indicate that \mci is insensitive to the specific choice of expansion ratio.

\begin{figure}[t!]
	\begin{center}
		
		\subfigure[Various quality of $k'$-NNG]{\label{sifg:nnvarite}\includegraphics[width=0.4\linewidth]{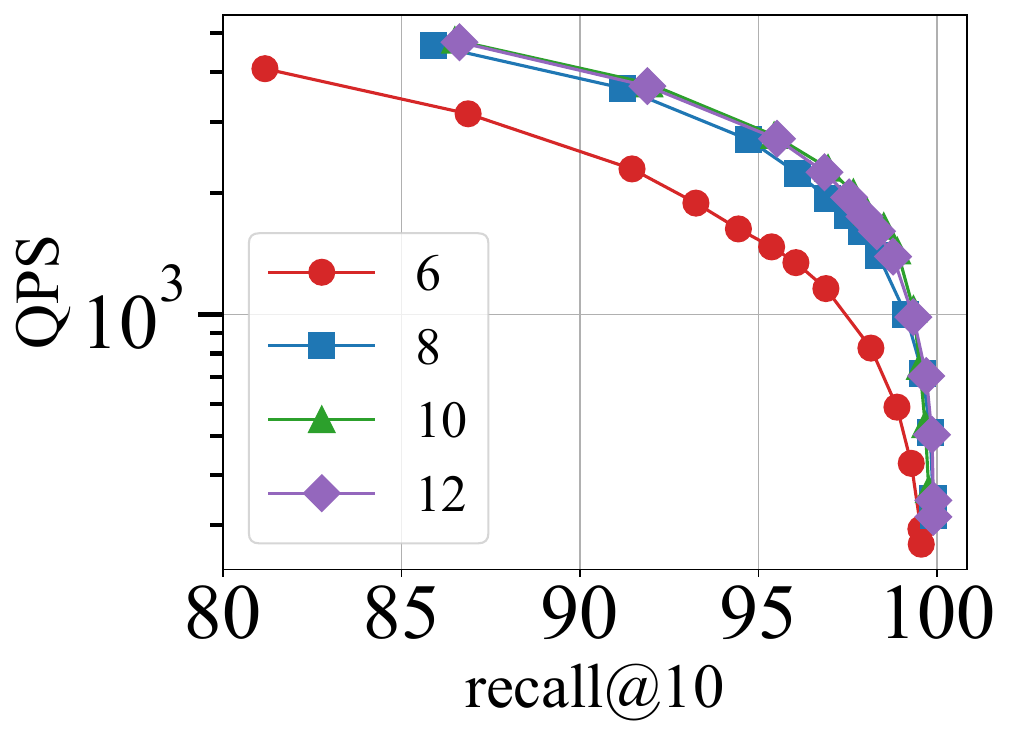}}
		\subfigure[Various $\alpha$ expansion ratio]{\label{sfig:alphaExpansionRatio}\includegraphics[width=0.4\linewidth]{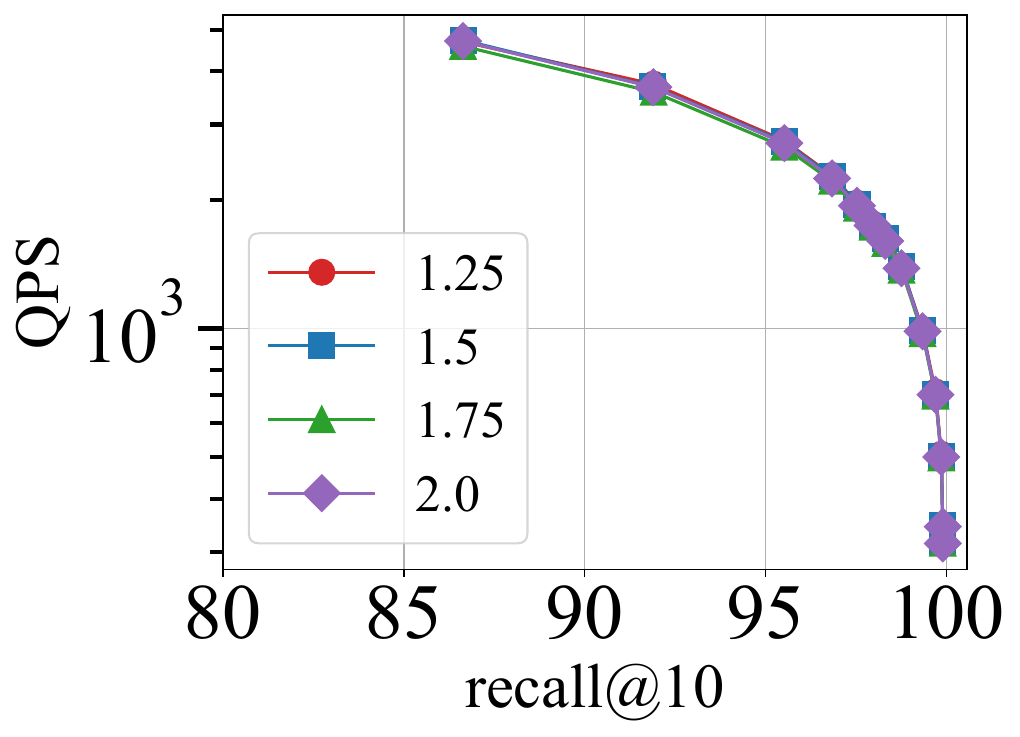}}
		
	\end{center}
	\caption{Additional experiments on \texttt{deep1M}}
	\label{fig:appendixexp} 
\end{figure}


\end{document}